\documentclass[10pt]{article}
\usepackage{amsmath,amssymb,fullpage,graphicx, xcolor,
enumerate, pdflscape}
\usepackage{comment, subfig, makecell,capt-of}
\usepackage{soul}
\usepackage{booktabs, siunitx}
\let\hat\widehat
\let\tilde\widetilde

\usepackage[ruled,linesnumbered]{algorithm2e}
\usepackage[colorlinks=true,linkcolor=blue,citecolor=blue]{hyperref}
\hypersetup{
   colorlinks=true, %set true if you want colored links
   linktoc=all,     %set to all if you want both sections %and subsections linked
   linkcolor=blue  %choose some color if you want links to %stand out
}

\usepackage[nottoc]{tocbibind}

\usepackage[authoryear,round]{natbib}
\usepackage{graphicx}

\newcommand{\lortomas}[1]{\textbf{\color{orange}[Lorenzo comment: #1]}}
\newcommand{\ltomas}[1]{{\color{black}#1}}

\newtheorem{theorem}{Theorem}
\newtheorem{lemma}[theorem]{Lemma}

\usepackage{url}

\newenvironment{proof}{{\bf Proof.}}{$\Box$}
\DeclareMathOperator*{\argmax}{argmax}
\DeclareMathOperator*{\argmin}{argmin}

\renewcommand{\P}{\mbox{$\mathbb{P}$}}
\newcommand{\E}{\mbox{$\mathbb{E}$}}
\allowdisplaybreaks
\definecolor{myblue}{RGB}{50,50,150}

\usepackage{soul}

\usepackage{comment}

\def\V{\textcolor{magenta}}
\def\VV{\textcolor{cyan}}

\begin{document}

\begin{center}
\LARGE{\bf Robust Simulation Based Inference}\\
\vspace{.25cm}
\large{Lorenzo Tomaselli, Val\'erie Ventura, Larry Wasserman}\vspace{.1cm}\\ 
Department of Statistics \& Data Science,
Carnegie Mellon University\vspace{-.1cm}\\ 
%July 13 2025 
%\vspace{1cm} 
\end{center}

\begin{quote}
{\em
Simulation-Based Inference (SBI) is an approach to statistical 
inference where simulations from an assumed model are used 
to construct estimators and confidence sets. 
SBI is often used when the likelihood is intractable and to construct 
confidence sets that do not rely on asymptotic methods or regularity 
conditions. Traditional SBI methods assume that the model is correct, 
but, as always, 
this can lead to invalid inference when the model is misspecified. 
This paper introduces 
robust methods that allow for valid frequentist inference 
in the presence of model misspecification.
We propose a framework where the target of inference is a projection parameter that 
minimizes a discrepancy between the true distribution and the assumed model. 
The method guarantees valid inference, even when the model is incorrectly specified
and even if the standard regularity conditions fail.
Alternatively, we introduce model expansion through exponential tilting 
as another way to account for model misspecification. 
We also develop an SBI 
based goodness-of-fit test to detect model misspecification.
Finally, we propose two ideas that are useful in the SBI framework beyond robust inference:
an SBI based method to obtain closed form approximations of intractable models and an active 
learning approach to more efficiently sample the parameter space.
}
\end{quote}

\vspace{-.2cm}

\vspace{-.4cm}

\tableofcontents

\newpage

\section{Introduction}

Simulation based inference (SBI)
is an approach to statistical
inference in which 
simulations from
an assumed model
facilitate inference.
SBI can be used for two
distinct purposes.
The first, and most common,
is to perform inference
when the likelihood function is intractable.
The second is to
construct confidence sets
when standard regularity conditions do not hold.
In some cases, SBI is used for both tasks.

Perhaps the earliest use of SBI
was for approximate Bayesian computation (ABC) \citep{rubin_bayesianly_1984, beaumont_approximate_2002}. 
This approach for inference compares summary statistics from  observed data with those extracted 
from simulations using generative models, endowed with a prior distribution.
Bayesian SBI approaches have been implemented in astrophysics
\citep{mishra-sharma_neural_2022}, high-energy physics 
\citep{cranmer_approximating_2016}, 
genetics \citep{beaumont_approximate_2002}, epidemic models \citep{mckinley_simulation-based_2014, ionides_inference_2015, 
minter_approximate_2019,
hao_reconstruction_2020, golightly_accelerating_2023}, and ecology \citep{beaumont_approximate_2010}, to cite a few. 
However, 
Bayesian methods
do not yield
valid confidence intervals.
The focus of this paper is instead on frequentist inference.

The literature on likelihood-based SBI
is large and growing fast.
Some key references include:
\cite{thomas_likelihood-free_2022, dalmasso_likelihood-free_2023,
mishra-sharma_neural_2022, brehmer_mining_2020, cranmer_frontier_2020,
cranmer_approximating_2016}.
SBI for likelihood or quasi-likelihood estimation has gained popularity in many fields from
epidemic models, via particle filtering or sequential Monte Carlo \citep{ionides_inference_2006, king_inapparent_2008, breto_time_2009},
to econometrics, where it is commonly named ``indirect inference" \citep{jiang_indirect_2004}.
SBI confidence sets are considered in
\cite{dalmasso_likelihood-free_2023, cranmer_frontier_2020, walchessen_neural_2024}, \cite{lenzi_neural_2023,lenzi_towards_2024},
\cite{xie_repro_2022}.

The usual likelihood-based
SBI methods lead to valid inference
under the assumption that the model is correct.
In this paper, we consider robust SBI methods that
allow model misspecification
and failure of regularity assumptions.
Our target of inference
is the {\em projection parameter}
$\theta^*$ that minimizes
$d(P_\theta,P)$
for some discrepancy $d$,
where $P$ denotes the true distribution,
which need not be contained in the assumed 
model ${\cal P} = \{ P_\theta:\ \theta\in\Theta \}$,
$\Theta\subset\mathbb{R}^d$. 
We also write
$d(P_\theta,P)$
as
$d(p_\theta,p)$,
where
$p_\theta$ and $p$
are the densities of
$P_\theta$ and $P$.
In tractable models that do not require SBI,
projection estimators have been studied in
\cite{beran_minimum_1977, lindsay_efficiency_1994, basu_robust_1998}.
In particular,
\cite{beran_minimum_1977}
emphasized the important role of Hellinger distance
because it yield efficient inference when
the model happens to be correct.
These works also assume that the model
satisfies many regularity conditions
which we try to avoid.
In intractable models that do require SBI,
\cite{nickl2010}
also consider projection estimators
when standard regularity holds
and the densities are estimated by B-splines.
A related method is repro sampling
\citep{xie_repro_2022}.
%
%\V{Larry, Do we need the following sentences? This is too much detail in intro.} In that approach, we assume that
%$(Y_1,\ldots,Y_n) = G(\theta,U_1,\ldots, U_n)$ for some function $G$
%and where $U$ has a known distribution $H$.
%One finds a set $B_\alpha$ such that
%$P( (U_1,\ldots, U_n)\in B_\alpha)\geq 1-\alpha$.
%Then
%$\{\theta:\ 
%(Y_1,\ldots, Y_n) = G(\theta,U_1,\ldots, U_n)\  
%\mathrm{for\ some\ } (U_1,\ldots, U_n)\in B_\alpha\}$
%is a $1-\alpha$ confidence set
%for $\theta$.
%The method is implemented by simulating from $H$.

In SBI, constructing confidence sets
that do not require regularity conditions typically 
relies on inverting a hypothesis test
\citep{dalmasso_likelihood-free_2023}.
But this approach does not yield
valid confidence sets when the model
is misspecified, 
because the null hypothesis $H_0: \theta = \theta_{\rm true}$ is false for
every $\theta$. 
Because minimum discrepancy estimators
are $M$-estimators, an alternative approach 
could use
$M$-estimation asymptotics
to find confidence sets.
But then this approach
assumes that the model
satisfies substantial regularity
conditions.
Since our goal is to have valid confidence sets
for projection parameters 
whether the model is regular or not, we will instead 
extend to the SBI framework the relative fit approach in
\cite{park_robust_2023,takatsu2025bridging}.
The method uses much weaker
regularity conditions.

A different approach to handle a misspecified model 
$p_\theta$ is to expand $p_\theta$ using an exponential tilt,
so it is more flexible, and apply existing 
likelihood-based SBI to make inference about $\theta$.
The exponential form of the model expansion
leads to some simplifications that reduce the computational 
burden.
If the expanded model remains misspecified, we can make robust inference 
about the projection parameter of the expanded model.

The purpose of this paper is to provide a robust SBI comprehensive framework 
to perform valid statistical inference 
without necessarily assuming a tractable likelihood, a correct model, 
or regularity conditions on the model.
In this paper we focus on the iid case.
In a companion paper we deal with
dependent data.

\medskip

{\bf Our Contributions.}
This paper makes the following contributions:

(1) We develop discrepancy based SBI (point and confidence set estimation) 
without assuming the model 
is correct and without making regularity assumptions on the model 
(Sections~\ref{sec::robust_sbi} and~\ref{sec::cs_misspecified}).

(2) We use one-step semiparametric estimators for
the discrepancies.

(3) We develop SBI based inference on 
the exponentially tilted model expansion
(Section~\ref{sec::exponential_tilt}).

(4) We propose an SBI based goodness of fit test
for the model (Section~\ref{sec::gof}).

(5) We compare three discrepancies and show their
advantages and disadvantages (Section~\ref{sec::applications_iid}).

In Section~\ref{sec::accoutrements} we also propose two ideas that are useful in the SBI framework beyond robust inference:

(6) In cases where SBI is used to estimate intractable likelihoods,
we show how SBI can be used to obtain
a closed form approximation to the model (Section~\ref{sec::approx}).
%Specifically, we find functions $\beta_1(\theta),\ldots, \beta_k(\theta)$ such that we can approximate $p_\theta(x)$ with the varying coefficient model $p_\theta(x) \propto e^{\sum_j \beta_j(\theta) b_j(x)}.$

(7) We develop an active learning approach to more efficiently sample the parameter space (Section~\ref{sec::AL}).

\medskip

But first, we introduce the basics of SBI
for point and confidence estimation in Section~\ref{sec::SBI}, and in Section~\ref{sec::density_ratio_estim} we review
techniques for density ratio estimation, which we use throughout.

\section{Simulation Based Inference}\label{sec::SBI}

We now review SBI in the case of a correctly specified model.
Let $Y_1,\ldots, Y_n \sim P$ be the observed data and
let ${\cal Y}_{obs} = (Y_1,\ldots, Y_n)$.
We consider parametric models
consisting of densities
${\cal P} = \{ p_\theta:\ \theta\in\Theta \}$
where $\Theta\subset\mathbb{R}^d$.
We let
${\cal L}(\theta) \equiv {\cal L}(\theta; {\cal Y}_{obs}) = \prod_i p_\theta(Y_i)$
denote the likelihood function and
$\ell(\theta) = \log {\cal L}(\theta)$
the log-likelihood function.
Let
$\hat\theta_{mle} = \argmax_\theta \ell(\theta)$
denote the maximum likelihood estimator.
Let $Y_1(\theta),\ldots, Y_m(\theta)$
denote a sample of size $m$ from $p_\theta$ 
and 
let ${\cal Y} \equiv {\cal Y}(\theta) = (Y_1(\theta),\ldots, Y_m(\theta))$.
In some cases,
one can take
$Y_i(\theta) = G(U_i,\theta)$
for some $G$, where
$U_1,\ldots, U_m$ is a draw
from a fixed distribution $F$.
In these cases,
samples from different $p_\theta$'s are
obtained from the same base sample $U_1,\ldots, U_m$.

%As mentioned in the previous section, SBI is used for two different tasks: estimating the likelihood function and constructing confidence sets.

\paragraph{Estimating the Likelihood Function.}

%Estimating ${\cal L}(\theta; {\cal Y})$ for any $\theta$ and ${\cal Y}$ is based on a simulated dataset and then a binary classifier is used to estimate a density ratio.

Let $\theta_1,\ldots, \theta_N\sim \pi$,
where $\pi$ is some distribution with full support on $\Theta.$
Let ${\cal Y}_j = {\cal Y}(\theta_j)$
be a sample of size $m$ from $p_\theta$ with 
$\theta=\theta_j$, $j = 1, \ldots, N$.
We simulate a dataset
$$
\{(Z_j,{\cal Y}_j,\theta_j): 1 \leq j \leq 2N\} = \{ (1,{\cal Y}_1,\theta_1),\ldots, (1,{\cal Y}_N,\theta_N),
(0,{\cal Y}_1,\theta_{q(1)}),\ldots, (0,{\cal Y}_N,\theta_{q(N)}) \},
$$
where $Z_j=1$ for $j\leq N$ and
$Z_j=0$ for $j > N$, and 
$q$ is a permutation of $1,\ldots, N$.
This is summarized in Algorithm~\ref{alg:data_SBI}.
The second half of the dataset is the same as the first
except that the $\theta_j$'s have been randomly permuted. 
The distribution of
$({\cal Y},\theta)$ given $Z=1$ is
$p(\theta,{\cal Y}) = p_\theta(Y_1,\ldots, Y_n)\pi(\theta)$,
%See Algorithm 1 for a summary.
while the distribution of
$({\cal Y},\theta)$ given $Z=0$ is
$p(Y_1,\ldots, Y_n) \pi(\theta)$,
where
$p(Y_1,\ldots, Y_n) = \int p_\theta(Y_1,\ldots, Y_n)\pi(\theta)d\theta$.
From Bayes' theorem, we have that
\begin{align*}
h({\cal Y},\theta) &\equiv 
P(Z=1|{\cal Y},\theta) =
\frac{p(\theta,{\cal Y})}{p(\theta,{\cal Y}) + \pi(\theta)p({\cal Y})}
=\frac{\pi(\theta)p_\theta({\cal Y})}
{\pi(\theta)p_\theta({\cal Y}) + \pi(\theta)p({\cal Y})}
=
\frac{p_\theta({\cal Y})}
{p_\theta({\cal Y}) + p({\cal Y})},
\end{align*}
so that
\begin{equation}
    \frac{h({\cal Y},\theta)}{1-h({\cal Y},\theta)} =
\frac{p_\theta({\cal Y})}{p({\cal Y})}\propto
p_\theta({\cal Y}) = {\cal L}(\theta; {\cal Y}).
\label{eq::likelihood_trick}
\end{equation}
Thus the binary classifier 
$h({\cal Y},\theta)$ estimates the likelihood function using 
the so-called ``likelihood ratio trick" 
\citep{cranmer_approximating_2016, cranmer_frontier_2020,walchessen_neural_2024}.

An alternative approach
\citep{thomas_likelihood-free_2022}
is to draw a sample
${\cal Y}'=(Y_1',\ldots, Y_\ell')$
from a fixed reference density
$g$ and then
fit a separate classifier
for each $\theta_j$
between
${\cal Y}_j$ and
${\cal Y}'$.
This requires more computation but might be
more accurate
since the classifier is focused on a single $\theta_j$.
Furthermore, we can use a different
reference density $g$ for each $\theta_j$ if desired. 
The tradeoff between improved accuracy and increased
classification is an open question.

\paragraph{Constructing Confidence Sets \citep{dalmasso_likelihood-free_2023}.}
%Now we explain how confidence sets are constructed as in \cite{dalmasso_likelihood-free_2023} assuming the model is correct.
Let $\theta^*$ denote the (unknown) true value of $\theta$.
Let $T(\theta,{\cal Y})$
denote any statistic
which is allowed to depend on the parameter as well as the data.
This could be, but need not be, the likelihood function.
Let
$$
B(\theta,{\cal Y}(\theta),{\cal Y}_{obs}) = \mathbb I \, 
\Bigl\{ T(\theta, {\cal Y}(\theta)) \geq T(\theta, {\cal Y}_{obs}) \Bigr\},
$$
where $\mathbb I$ is the indicator function.
Now, 
\begin{equation}
    \label{eq::B}
\text{pv}(\theta, {\cal Y}_{obs})=\E_\theta[B(\theta,{\cal Y}(\theta),{\cal Y}_{obs})]
\end{equation}
is precisely
the p-value for testing
that the true value of the parameter is $\theta$.
(The expected value is over the randomness of
${\cal Y}(\theta)$ with ${\cal Y}_{obs}$ and $\theta$ treated as fixed.)
Thus,
$$
C= \{\theta:\ \text{pv}(\theta,{\cal Y}_{obs}) \geq \alpha\}
$$
is an exact $1-\alpha$ confidence set for $\theta^*$, that is
$
\inf_\theta P_\theta(\theta\in C) \geq 1-\alpha.
$

In SBI, we use simulation to estimate
(\ref{eq::B}):
we simulate
$\theta_1,\ldots, \theta_N$ from some distribution $\pi$.
For each $j$,
we simulate
${\cal Y}_j \equiv {\cal Y} (\theta_j)$ from $p_{\theta_j}$.
Let
$$
B_{j} = \mathbb I \,\Bigl\{ T(\theta_j,{\cal Y}_j) \geq T(\theta_j,{\cal Y}_{obs} )\Bigr\}.
$$
Now we perform
nonparametric regression of
$B_1,\ldots, B_N$
on
$\theta_1,\ldots, \theta_N$,
which gives an estimate
$\widehat{\text{pv}}(\theta, \mathcal{Y}_{obs})$ 
of 
(\ref{eq::B}). 
The estimated confidence set is
$\hat C = \{ \theta:\ \widehat{\text{pv}}(\theta, \mathcal{Y}_{obs}) \geq \alpha\}$.
($C$ can also be obtained by using quantile regression to estimate the
$1-\alpha$ quantile of the test statistic rather than using the $p$-value.)
Assuming $\text{pv}$ is $\gamma$-Holder smooth,
typical nonparametric regression methods
achieve
$$
\hat{\text{pv}}(\theta,{\cal Y}_{obs})-  \text{pv}(\theta,{\cal Y}_{obs}) = O_P(N^{-\gamma/(2\gamma+d)}).
$$
(Recall that, in simple terms,
$\gamma$-Holder smooth means that the function has $\gamma$ continuous derivatives.)
In many cases,
$\text{pv}(\theta,{\cal Y}_{obs})$ is infinitely differentiable so that
$\hat{\text{pv}}(\theta,{\cal Y}_{obs})-  \text{pv}(\theta,{\cal Y}_{obs}) = 
O_P (\sqrt{\log N/N})$.
As long as $N > n \log n$,
the error added by estimating the p-value function is then
negligible.
When the model is correct, this approach yields valid and efficient confidence sets. However, when the model is misspecified,
inverting the test does not
yield valid confidence sets.

\section{Densities and Density Ratios}\label{sec::density_ratio_estim}

SBI typically requires estimating densities
\citep{nickl2010}
or density ratios 
\citep{cranmer_frontier_2020}.
We saw that for the likelihood based approach described earlier
where we used a classifier to estimate the ratio
${p(\theta,{\cal Y})}/{( \pi(\theta)p({\cal Y}) )} $.
In some cases
we have a choice of estimating a density or a density ratio.
Current practice seems to
focus mostly on
density ratio estimation rather than
density estimation. Indeed, there seems to be a unspoken assumption
in much of the SBI literature
that density ratios can be easier to estimate than
densities.
There are, perhaps, two reasons
that users prefer density ratio estimation
to density estimation.
The first is that density ratios
can be estimated by using
classification methods and there
is a plethora of available methods. For example,
random forests, boosting, neural nets
and deep learning
are popular classification methods
that have been shown to be very effective in practice.
Especially in multivariate cases,
this could be a benefit.
With certain assumptions on the density ratio,
it has been shown that
deep learning methods can
achieve fast rates of convergence
that might even be dimension independent
\citep{bauer2019deep, schmidtheiber2020nonparreg, kohlerlanger2021ratedeep}.
Such results need to be treated with caution, as they do
make extra assumptions on the function being estimated.
Nonetheless, such results provide a strong motivation
for neural net methods.
A second reason for
preferring density ratios is that
they can sometimes be less complex than
densities.
For example,
consider two densities $p$ and $q$.
We might have that
$p,q\in \text{Holder}(\beta)$ while
$p/q\in \text{Holder}(\xi)$
with $\xi > \beta$.
An extreme example is when $p$ is highly
nonsmooth,
but $q=p$ so that $r = 1$.
In this case,
$p$ and $q$ are complex but the ratio is
simple.
The text by \cite{sugiyama_density_2012} presents many effective techniques for
estimating density ratios.

Whether it is better
to focus on estimating densities
or density ratios
is an open question.
Given the current preference
for density ratios,
we will express our methods
in terms of density ratios
to be consistent with common practice
but one could replace density ratio estimation with
density estimation in what follows.

Suppose that
$Y_1,\ldots, Y_n \sim p$,
$Y_{n+1},\ldots, Y_{n+m} \sim q$
and that we want to estimate
$r(y) = p(y)/q(y)$. 

{\bf Classifier Approach.}
We define
the pair $(Z,Y)$ where
$Z=1$ if $Y\sim P$ and
$Z=0$ if $Y\sim Q$.
Then, by Bayes' theorem,
$$
r(y) = \frac{1-a}{a} \frac{1-h(y)}{h(y)}
$$
where
$a = n/(n+m)$ and
$h(y) = P(Z=1|Y=y)$.
The function $h$ is estimated using a classifier
to get $\hat h$ and then we set
$$
r(y) = \frac{1-a}{a} \frac{1-\hat h(y)}{\hat h(y)}.
$$
The classifier can be
logistic regression, a random forest, a neural net etc.

{\bf Least Squares Approach.}
Classifiers like neural nets
are very flexible
but they require careful training,
very large sample sizes
and can require choosing many tuning parameters.
An alternative is to use
an $L_2$ approach
\citep{kanamori_least-squares_2009}.
The goal is to choose $\hat r$
to minimize
$$
\int (\hat r - r)^2 q =
\int \hat r^2 q - 2 \int \hat r r q + \int r^2 q =
\int \hat r^2 q - 2 \int \hat r p + \int r^2 q.
$$
The last term does not involve $\hat r$
so it suffices to choose $\hat r$ to
minimize
$$
L(\hat r) = \int \hat r^2 q - 2 \int \hat r p
$$
which can be estimated by
\begin{equation}\label{eq::theloss}
\hat L(\hat r) =
\frac{1}{m}\sum_{i=n+1}^{n+m} \hat r^2(Y_i) -
\frac{2}{n}\sum_{i=1}^n \hat r(Y_i).
\end{equation}
Following \citep{kanamori_least-squares_2009},
we assume that
$r$ is contained in a reproducing kernel Hilbert space (RKHS)
${\cal H}$ defined by a kernel $K$.
More precisely,
we minimize the penalized loss
$$
\hat L(\hat r) + \lambda ||r||_{\cal H}^2.
$$
The minimizer $\hat r$ of
$\hat L(\hat r)$
subject to $r\in {\cal H}$
has the form
$\hat r(y) =\sum_i \beta_j K(Y_i,y)$ 
for some $\beta_1,\ldots, \beta_{n+m}$.
It usually suffices to
restrict $r$ to be of the form
\begin{equation}
r(y)=\sum_{i=1}^{n_c} \beta_i K(y,\sigma, c_i) 
\label{eq::kernel_density_ratio}
\end{equation}
where $c_1,\ldots, c_{n_c}$
are centers,
$K(y, \sigma, c)$ is a Gaussian kernel with 
center $c$ and scale $2\sigma^2$ evaluated at a point $y$. 
Define 
$\widehat{H}_{ij}=
\dfrac{1}{m}\sum_{k=n+1}^{n+m}
\exp{\left(-\frac{\|Y_k-c_i\|^2+\|Y_k -c_j\|^2}{2\sigma^2}\right)}$ 
and 
$\widehat h\in \mathbb{R}^{c}$ with 
$\hat h_{i}=\frac{1}{n}\sum_{k=1}^{n} 
\exp{\left(-\frac{\|Y_k-c_i\|^2}{2\sigma^2}\right)}$ 
for $i,j=1,\dots,n_c$.  The centers of the Gaussian kernels are either selected at random 
or can be selected over a grid. 
The values of the hyperparameters $\sigma$ and $\lambda$ are chosen by cross-validation \citep{sugiyama_density_2010}. Alternatively, $\sigma$ can be 
chosen by the median heuristic
\citep{garreau2017}.

Inserting this into
(\ref{eq::theloss})
and minimizing over $\beta$ yields:
\begin{equation*}
\hat \beta := \arg\min_{\beta \in \mathbb{R}^c}\ 
\frac{1}{2}\,\beta^\top \hat H \beta - 
\hat h^\top \beta + \frac{\lambda}{2} \beta^\top\beta
\end{equation*}
To ensure non-negativity of the density ratios, 
we set $\hat \beta_i=\max(0, \beta_i)$ elementwise. 
\cite{sugiyama_density_2010}
show that,
if $r\in Holder(\beta)$ for $\beta>1/2$, then
$||\widehat{r}-r|| =  O_P\left(n^{-\frac{\beta}{2\beta+d}}\right)$.

{\bf Remark:}
{\em
We can reduce the computation by constructing
only one density ratio estimator.
Generate
$(Z_j,\theta_j,Y_j)$
for $j=1,\ldots, 2N$
as follows.
For $1\leq j \leq N$
set $Z_j=1$,
draw $\theta_1,\ldots, \theta_N \sim \pi$
and $Y_j \sim p_{\theta_j}$.
For $N+1\leq j \leq 2N$,
set $Z_j=0$,
draw $\theta_{N+1},\ldots, \theta_{2N} \sim \pi$
and
$Y_j \sim g$.
Then, as in Section~\ref{sec::SBI} we have 
$p(y,\theta|z=1) \propto \pi(\theta)p_\theta(y)$
and
$p(y,\theta|z=0) \propto \pi(\theta)g(y)$,
and estimating the ratio of these two densities gives
$p_\theta(y)/g(y)$.
Hence a single density ratio estimator
yields $p_\theta(y)/g(y)$ for all values of $\theta$
simultaneously.
Similarly, if we prefer to estimate densities rather than density ratios,
a single density estimator applied to the first sample
yields $p(\theta,y)\propto p_\theta(y)$
and hence estimates the density for each $\theta$
simultaneously.
There is a tradeoff.
One can do many density estimates of the dimension of $Y$ are one
density estimate of the dimension of $(Y,\theta)$.
}

\section{Robust SBI Using Discrepancies}\label{sec::robust_sbi}

If the true distribution $p$
is not contained in the model 
${\cal P} = \{ p_\theta:\ \theta\in\Theta \}$
then we say that the model is misspecified.
In this case, we take as our target of inference
\begin{equation}
\theta^* = \argmin_\theta d(p_\theta,p)\label{eq::proj_param}
\end{equation}
where $d(\cdot,\cdot)$ is some discrepancy.
We call $\theta^*$ the
{\em projection parameter}
(this corresponds to the true value
when the model is correctly specified).
Under regularity conditions,
the mle converges to the value that minimizes
the Kullback-Leibler discrepancy
$D(p,p_\theta) = \int p \log (p/p_\theta)$.
But this discrepancy
leads to non-robust estimators \citep{beran_minimum_1977}.
Instead, we consider 
three other discrepancies:
the Hellinger discrepancy,
the power divergence and
the kernel distance.
Each
has advantages and disadvantages; see Table
\ref{table::compare}.
In this section we discuss point estimation for the projection parameter and,
in the next section, we provide confidence sets.

We next define the discrepancies
and their estimators.
The estimators proposed in this section are one step estimators,
which
have the form: plugin estimator plus influence function.
Under smoothness conditions,
these estimators are efficient
and asymptotically
Normal.
In what follows,
we will often use a sample
$Y_1^*,\ldots, Y_k^*$ from a convenient reference density $g$.
We assume that $g$ is known in closed form.

\paragraph{The Hellinger discrepancy.}
The Hellinger discrepancy between $p_\theta$ and $p$ is
$$
h^2(p_\theta,p) = \int (\sqrt{p_\theta}-\sqrt{p})^2 = 2 - 2\psi(p_\theta,p)
$$
where
$\psi(p_\theta,p)=\int \sqrt{p_\theta p}$.
Given an estimate 
$\hat\psi(p_\theta,p)$ of $\psi(p_\theta,p)$ 
we take
$\hat h^2(p_\theta,p) = 2 - 2 \hat\psi(p_\theta,p)$. 
Most work on estimating the Hellinger distance
has used the plugin estimate
$\psi(p_\theta,\hat p)$,
where $\hat p$ is a nonparametric density estimate
\citep{beran_minimum_1977, basu_robust_1998}.
This can lead to $n^{-1/2}$ consistent
estimates in some cases
if the density estimate is carefully undersmoothed.
But in general 
these estimates are asymptotically biased.
Instead, we use the
semiparametric one-step estimator.

\begin{lemma}\label{lemma::hellinger}
Let
$Y_1,\ldots, Y_n \sim p$,
$Y_1(\theta),\ldots, Y_m(\theta) \sim p_{\theta}$, 
and $Y_1^*,\ldots, Y_k^* \sim g$.
Let
$r(x)=\frac{p(x)}{g(x)}$ and $s_\theta(x)=\frac{p_\theta(x)}{g(x)}$.
Let
$\hat r$ be an estimate of $r$ based on
$Y_1,\ldots, Y_n$ and
$Y_1^*,\ldots, Y_k^* $,
as discussed in
the previous section.
Similarly, let
$\hat s_\theta$ be an estimate of $s_\theta$ based on
$Y_1(\theta),\ldots, Y_m(\theta)$
and
$Y_1^*,\ldots, Y_k^*$. 
The one-step estimator is 
\begin{align}
\hat\psi(p_{\theta}, p) &=
\frac{1}{2n} \sum_i \sqrt{\frac{\hat s_\theta(Y_i)}{\hat r(Y_i)}} +
\frac{1}{2k} \sum_i \sqrt{\frac{\hat r(Y_i(\theta))}{\hat s_\theta(Y_i(\theta))}}\label{eq::est_MHDE}.
\end{align}
Suppose that $r\in \text{Holder}(\beta_1)$ and 
$s_\theta\in \text{Holder}(\beta_2)$ for each $\theta$,
where $\beta_1,\beta_2 > d/2$  
(where we recall that $d$ is the dimension of $\theta$)
and $m,k\geq n$.
Assume that
$||\hat r - r|| = o_P(n^{-1/4})$ and
$||\hat s_{\theta^*} - s_{\theta^*}|| = o_P(n^{-1/4})$.
Then
$$
\sqrt{n}(\hat \psi -\psi)\rightsquigarrow N(0,\sigma^2) \ \ \mbox{where} \ \ \sigma^2 = \dfrac{1-\psi^2}{2}.
$$
\end{lemma}

The proof of lemma~\ref{lemma::hellinger} and all proofs henceforth are provided in Appendix~\ref{appendix::proofs}.

{\bf Remark:}
{\em The ratios in the sums can become unstable
in the tails so, in practice,
we trim the ratios.
The Normal approximations in the two theorems
breaks down if $p = p_{\theta^*}$ 
because the variance of the estimator tends to 0.
But if we add $1/n$ to the
estimated variance of $\hat\psi$
then the confidence intervals are still valid even in this case.}

\paragraph{The Power Divergence.}
The power divergence \citep{basu_robust_1998} is 
$$
d_\gamma(p,p_\theta) = 
\int \Biggl\{ p_\theta^{1+\gamma}(x)- \left( 1 + \frac{1}{\gamma}\right) 
p(x)p_\theta^\gamma(x)+ \frac{1}{\gamma}p^{1+\gamma}(x)\Biggr\} dx, 
$$
parametrized by $\gamma\in (0, 1]$. This is a wide family of divergences that 
balances efficiency (low $\gamma$) and robustness (large $\gamma$). 
Robustness, here, means that the projection is not sensitive to small
changes in $p$.
This includes
the KL distance ($\gamma \to 0$) and $L_2$ distance ($\gamma=1$).
For the purposes of this paper, it is only necessary
to estimate the first two terms 
\begin{equation}
\psi_\gamma(p,p_\theta)=
\int \Biggl\{ p_\theta^{1+\gamma}(x)- \left( 1 + \frac{1}{\gamma}\right) 
p(x)p_\theta^\gamma(x)\Biggr\} dx \label{eq::mdpd_discrepancy}
\end{equation}
since the third term is a constant of $\theta$.

\begin{lemma}\label{lemma::power}
Let
$Y_1,\ldots, Y_n \sim p$,
$Y_1(\theta),\ldots, Y_m(\theta) \sim p_\theta$ and $Y_1^*,\ldots, Y_k^* \sim g$. Considering the ratio-based approach, the one-step estimator is 
\begin{align}
\hat{\psi}_\gamma(p_\theta, p) &= \left(1+\gamma\right) \dfrac{1}{m}\sum_i \widehat{r}_\theta^\gamma (Y_i(\theta)) g^\gamma (Y_i(\theta))  \left(1 - \dfrac{\widehat{r}(Y_i(\theta)) }{\widehat{r}_\theta (Y_i(\theta)) }\right) - \dfrac{\left(1+\gamma \right)}{\gamma} \dfrac{1}{n}\sum_i \widehat{r}_\theta^\gamma (Y_i) g^\gamma (Y_i)  -
\gamma \hat\psi_{\gamma}(p_\theta, p,\gamma)\nonumber\\
&=  \dfrac{1+\gamma}{m} \sum_i   \widehat{r}_\theta^{\gamma}(Y_i(\theta))g^{\gamma}(Y_i(\theta)) -   \dfrac{1+\gamma}{m} \sum_i   \widehat{r}_\theta^{\gamma-1}(Y_i(\theta))\widehat{r}(Y_i(\theta))g^{\gamma}(Y_i(\theta))\nonumber\\
	&\quad- \left(1+\dfrac{1}{\gamma}\right)\dfrac{1}{n} \sum_i \widehat{r}_\theta^\gamma (Y_i) g(Y_i)^\gamma 
 -  \dfrac{\gamma}{k} \sum_i \widehat{r}^{1+\gamma}_\theta(Y_i^*) g^\gamma(Y_i^*)
     +  \dfrac{1+\gamma}{k} \sum_i 
     \widehat{r}(Y_i^*)\widehat{r}^{\gamma}_\theta(Y_i^*) g^\gamma(Y_i^*) 
\label{eq::est_MDPD}
\end{align}
which, for the $L_2$ loss ($\gamma=1$), simplifies to
\begin{align}
    \hat{\psi}_1 (p_\theta, p) 
&=  \dfrac{2}{m} \sum_i   \widehat{r}_\theta(Y_i(\theta))g(Y_i(\theta)) -  \dfrac{2}{m} \sum_i  \widehat{r}(Y_i(\theta))g(Y_i(\theta))- \dfrac{2}{n} \sum_i \widehat{r}_\theta (Y_i) g(Y_i)\nonumber\\
	&\quad - \dfrac{1}{k} \sum_i \widehat{r}^{2}_\theta(Y_i^*) g(Y_i^*) +
     \dfrac{2}{k} \sum_i 
     \widehat{r}(Y_i^*)\widehat{r}_\theta(Y_i^*) g(Y_i^*) 
\label{eq::est_L2}
\end{align}
Under the conditions of lemma~\ref{lemma::hellinger},
$$\sqrt{n}(\hat\psi_\gamma - \psi_\gamma)\rightsquigarrow N(0,\sigma^2)$$
where 
\begin{align*}
    \sigma^2 &=\E_{pp_\theta} \Biggl[\Biggl( \Bigl(1+\frac{1}{\gamma}\Bigr) s_\theta^{\gamma}(Y) g^{\gamma}(Y) - 
    (1+\gamma)^2 s_\theta^{\gamma}(Y(\theta))g^{\gamma}(Y(\theta))\Biggl)^2\Biggr]
    + (1+\gamma)^2 \E_g \Bigl[r^2(Y^*) s_\theta^{2\gamma-1}(Y^*) g^{2\gamma}(Y^*) \Bigr] \\
    & \quad - 2(1+\gamma)^2\E_g \Bigl[  r (Y^*) s_\theta^{2\gamma}(Y^*) g^{2\gamma}(Y^*)\Bigr]  
    + \frac{2(1+\gamma)^2}{\gamma}  \Bigl(\E_g \Bigl[ r(Y^*) s_\theta^\gamma(Y^*) g^{\gamma}(Y^*)\Bigr]\Bigr)^2
     - (1+\gamma)^2\psi_\gamma^2.
\end{align*}
\end{lemma}

For future reference we note that the discrepancy estimates can be written 
in the form
\begin{equation}\label{eq::UV}
\hat d(p_\theta,p) = \frac{1}{n}\sum_i U(Y_i,\theta) + 
\frac{1}{m}\sum_i V(Y_i(\theta),\theta) + 
\frac{1}{k}\sum_i W(Y_i^*,\theta),
\end{equation}
where the last term is absent in the Hellinger discrepancy.

\paragraph{Kernel Distance (MMD -- maximum mean discrepancy).} 
MMD has been used for as a minimum distance
estimator in
\cite{cherief2022, briol_statistical_2019}.
For random variables $X$, $Y$ defined on the sample space $\Omega$, let $K:\Omega\times\Omega\mapsto\mathbb{R}$ be a symmetric, positive-definite kernel function defining a reproducing kernel for the associated reproducing kernel Hilbert space (RKHS), $\cal H$.
We define the squared kernel distance 
$$
d^2(p_\theta,p) =
\E[K(X,X')] - 2\E[K(X,Y)] + \E[K(Y,Y')]
$$
with $X,X'\sim p_\theta$ and
$Y,Y'\sim p$. 
This quantity measures the distance between distributions $P_\theta$ and $P$ 
with densities $p_\theta$ and $p$.
Unlike the previous two discrepancies,
it is not necessary to adjust the estimator
using influence functions because the estimator is unbiased
and it is not necessary to estimate the densities
or the density ratios since the density does not appear in the distance.
The MMD can be estimated using the 
standard estimator in \cite{gretton_kernel_2012}, namely,
\begin{equation}
\hat d^2(p_\theta,p)=
\frac{1}{m(m-1)}\sum_{i\neq j}K(Y_i(\theta),Y_j(\theta)) +
\frac{1}{n(n-1)}\sum_{i\neq j}K(Y_i,Y_j) -
\frac{2}{mn}\sum_{i,j}K(Y_i(\theta),Y_j). \label{eq_mmd}
\end{equation}
Then when $p\neq p_\theta$ and $n=m$, we have the convergence result [Corollary 16, \cite{gretton_kernel_2012}] 
$$
\sqrt{n}(\hat d^2(p_\theta,p) - d^2(p_\theta,p))\rightsquigarrow N(0,\sigma^2)
$$
with $\sigma^2=4\left(\mathbb{E}_{p}[\left(\mathbb{E}_{p_\theta}[h(Y,Y(\theta))|Y]\right)^2]-\left(\mathbb{E}_{p,p_\theta}[h(Y,Y(\theta))]\right)^2\right)$, 
where $h(w_i,w_j)=K(x_i, x_j)+K(y_i, y_j)-K(x_i, y_j)-K(x_j, y_i)$ for $w_i=(x_i, y_j)\sim p\times p_\theta$.

However, when $P$ is close to or equal to $P_\theta$, 
(\ref{eq_mmd}) is a degenerate U-statistic
and its asymptotic distribution is not Normal and hard to work with
\citep{shekhar_permutation-free_2023}. 
This is problematic since our approach to build confidence sets consists in inverting  relative fit-type of tests over the parameter space \citep{park_robust_2023} and depends on whether CLT holds for the test statistic (a linear function of the MMD discrepancy). 
We will instead use the studentized MMD estimator proposed in 
\cite{shekhar_permutation-free_2023, kim2024}, because it has more convenient 
asymptotic properties than (\ref{eq_mmd}). 
We proceed by splitting the observed data
in two subsets $\mathcal{I}_i$ of size $n_i$,
and we compute the  kernel mean embedding for the true distribution, 
$\hat \mu_i = \frac{1}{n_i} \sum_{i \in \mathcal{I}_i} K(Y_i, \cdot)$, $i=1,2$.
Similarly, we split the simulated datasets in two 
subsets $\mathcal{I}_i(\theta)$ of size $m_i$,
and 
compute the  kernel mean embedding for the model, 
$\hat \mu^\theta_i = \frac{1}{m_i} \sum_{i \in \mathcal{I}_i(\theta)} K(Y_i(\theta), \cdot)$, $i=1,2$.
The variance of the MMD estimator is then defined as  
the weighted sum of the variance of the model and true distribution mean embeddings
$$\hat\sigma^2=\dfrac{1}{m_1}\hat\sigma_\theta^2+\dfrac{1}{n_1}\hat\sigma_Y^2$$
where $\hat\sigma_\theta^2=\frac{1}{m_1}\sum_j\left(H^\theta_{j2}- \overline H^\theta_{2} \right)^2$ and $\hat\sigma_Y^2=\frac{1}{n_1}\sum_{j'}\left(H_{j'2}-\overline H_{2} \right)^2$, 
with $H^\theta_{ji}=\langle K(Y_j(\theta),\cdot), \hat\mu_2^\theta-\hat\mu_2\rangle$ and 
$H_{j'i}=\langle K(Y_{j'},\cdot), \hat\mu_2^\theta-\hat\mu_2\rangle$ 
for $i=1,2$, $j=1,\dots,m_1$, $j'=1,\dots,n_1$. 
This construction is not symmetric in the two splits
and this is needed to get a Normal limit
even when $p_\theta=p$.
We can now define the studentized MMD estimator as
\begin{align}
   \hat d^2(p_\theta,p)=\dfrac{1}{\hat\sigma}&\Bigl(
\frac{1}{m_1 m_2}\sum_{i\neq j}K(Y_i(\theta),Y_j(\theta)) +
\frac{1}{n_1 n_2}\sum_{i\neq j}K(Y_i,Y_j) \nonumber\\ 
&\quad - \frac{1}{m_1 n_2}\sum_{i,j}K(Y_i(\theta),Y_j) -
\frac{1}{n_1 m_2}\sum_{i,j}K(Y_i(\theta),Y_j)\Bigr). \label{eq::cross_mmd}
\end{align}
Standardizing the original MMD estimator renders it asymptotically standard normal, 
regardless of whether the true (unknown) distribution $P$ is equal or close to the model $P_\theta$, 
even in high-dimensional settings.

\begin{table}
\centering
\begin{tabular}{l|lll}
           & \makecell[l]{
           requires density \\(or density ratio) estimation?} & \makecell[l]{efficient?} & \makecell[l]{need extra sample?}\\ \hline
Hellinger        & Yes                             & Yes        & No \\
Power Divergence & Yes / (No if using densities)   & No         & Yes/(No if using densities)\\
MMD              & No                                & No         & No\\
\end{tabular}
\caption{\em Comparison of discrepancies. 
The MMD has the advantage that it does not require
density estimation.
The Hellinger discrepancy leads
to an estimator that is efficient if the model is correct.
Estimating the power divergence requires
an extra sample for estimating density ratios.} 
\label{table::compare}
\end{table}

\section{Confidence Sets for Misspecified Models}\label{sec::cs_misspecified}

When the model is correctly specified, inverting a
test as described in Section~\ref{sec::SBI} yields confidence sets with valid coverage. 
Here we construct
confidence sets for the projection parameter when the model is 
misspecified.

Let $Y_1,\dots,Y_n\sim P$ and $Y_1(\theta),\dots,Y_m(\theta)\sim P_\theta$ and
$Y_1^*,\ldots, Y_k^*\sim G$. 
Recall that
estimators of $d(p_\theta,p)$ have the form in (\ref{eq::UV}).
%$$
%\hat d(p_\theta,p)=
%\frac{1}{n}\sum_i U(Y_i, \theta) +
%\frac{1}{m}\sum_i V(Y_i(\theta), \theta) .
%$$
%\V{If this is eq 10, then we just need to cite eq 10 (note that eq 10 has an extra term).}

\begin{comment}
    Since $\hat\theta$ is an $M$-estimator
it is tempting to use
the standard expansion \V{Need to add all the $V$ terms to that expansion.}
$$
0 = \hat d(\hat\theta) = \frac{1}{n}\sum_i U(Y_i,\hat \theta) =
\frac{1}{n}\sum_i U(Y_i,\theta^*) +
\frac{1}{n}\nabla \hat U(Y_i,\theta^*) (\hat\theta - \theta^*) + R
$$
for a remainder $R$.
This implies, under the strong regularity conditions, that
$$
\sqrt{n}(\hat\theta - \theta^*)\rightsquigarrow N(0,\Sigma)
$$
where 
$$
\Sigma = (\E[\nabla U(Y,\theta^*)])^{-1} \E[-U(Y,\theta^*)]
((\E[\nabla U(Y,\theta^*)])^{-1})^T.
$$
This leads to the standard confidence set
$$
C = \Bigl\{ \theta:\ n(\theta-\hat\theta)^T \hat\Sigma^{-1} (\theta-\hat\theta)\leq \chi^2_{\alpha,d}\Bigr\}.
$$
\end{comment}

If regularity conditions hold,
we can use standard $m$-estimator asymptotic methods
to get confidence sets.
We discuss this in the appendix.
However, recall that one of our
goals is to have confidence sets
that do not require the regularity conditions.
So our preferred approach is to adapt the 
idea from \cite{park_robust_2023, takatsu2025bridging}
based on tests of relative fit.
For each $\theta$, we test
$$
H_0: d(p_\theta,p) \leq d(p_{\hat\theta},p)
$$
where $\hat\theta$ is some preliminary estimator
based on a separate sample and is regarded here as fixed.
By definition, the projection parameter $\theta^*$ satisfies this null hypothesis,
so inverting the test
yields a confidence interval for the projection parameter $\theta^*$.
What makes this method attractive
is that it only requires a central limit theorem
for $\hat d(p_{\theta},p)$
and this will typically hold since
$\hat d(p_{\theta},p)$ is a sample average.
In contrast, using the asymptotic distribution of the
$M$-estimator $\hat\theta$
relies strongly on regularity conditions for the model.

Let
$\Delta(\theta_1,\theta_2) = 
d(p_{\theta_1},p)-d(p_{\theta_2},p)$.
From (\ref{eq::UV}),
the estimated difference of discrepancies for the samples $Y_1, \ldots, Y_n\sim p$, 
$Y_1(\theta), \ldots, Y_m(\theta)\sim p_{\theta}$ 
and
$Y_1^*,\ldots, Y_k^*\sim g$,
can be written as
\begin{align}
\hat \Delta(\theta_1, \theta_2) =
\Biggl(&\dfrac{1}{n}\sum_i U(Y_i,\theta_1)+\dfrac{1}{m}\sum_i V(Y_i(\theta_1),\theta_1)+\dfrac{1}{k}\sum_i W(Y_i^*,\theta_1) \Biggr)\nonumber\\ &-
\Biggl(\dfrac{1}{n}\sum_i U(Y_i,\theta_2)+\dfrac{1}{m}\sum_i V(Y_i(\theta_2),\theta_2) + \dfrac{1}{k}\sum_i W(Y_i^*,\theta_2) \Biggr) \label{eq::diff_discr}
\end{align}
%\V{this is not what i get from 10; 10 has an extra term} \lortomas{some discrepancies do not have the extra term} 
which are sample averages
so we can use the central limit theorem.
When we use this idea,
we take $\theta_1 =\theta$ and $\theta_2 = \hat\theta$
where $\hat\theta$ is based on a separate sample.
The use of sample splitting is crucial
since it allows the use of the central 
limit theorem.
Let 
$s(\theta_1, \theta_2)$ be the estimated standard error of $\hat\Delta(\theta_1,\theta_2)$.
The steps are in Algorithm \ref{alg::relfit}.

\begin{algorithm}[t!]
\caption{SBI Relative Fit Confidence Set}\label{alg::relfit}
\vspace{.5cm}
\begin{enumerate}
\item Split the data into two groups
${\cal D}_0$ and ${\cal D}_1$ 
each of size $n_0=n_1= n$.
\item Construct a preliminary estimator $\hat\theta$
from ${\cal D}_0$.
\item Draw $\theta_1,\ldots, \theta_N\sim \pi$.
\item Calculate 
$\hat\Delta(\theta_j,\hat\theta)$ and its standard error $s(\theta_j, \hat \theta)$ from ${\cal D}_1$, for $j=1, \ldots, N$.
\item Let
$Z_j = -\Phi \left(-\hat\Delta(\theta_j,\hat\theta) / s(\theta_j, \hat \theta) \right)$ 
be the p-value for the test with null hypothesis
$\Delta(\theta_j, \hat\theta) \le 0$.
\item Smooth the $Z_i$'s to obtain estimated p-values for all $\theta$
$\hat{\text{pv}}(\theta) = \sum_j Z_j K_h(\theta_j - \theta)/\sum_j K_h(\theta_j - \theta)$
where $K_h$ is a kernel with bandwidth $h$.
\item Return the estimated confidence set
$\hat C = \{ \theta:\ \hat{\text{pv}}(\theta) \geq \alpha\}$.
\end{enumerate}
\end{algorithm}

\begin{theorem} \label{thm::whatisthis}
Suppose that,
conditional on ${\cal D}_0$,
\begin{equation}\label{eq::need_clt}
\frac{\sqrt{n} \, \hat\Delta(\theta^*,\hat \theta) }{s(\theta^*,\hat\theta)}
\rightsquigarrow N(0,\sigma^2).
\end{equation}
Assume that
$\text{pv}(\theta) \in Holder(\beta)$
and that
$h\sim (1/N)^{1/(2\beta+d)}$.
Then
$$
P(\theta^* \in \hat C) = 1-\alpha + O_P(n^{-1/2}) +O_P(N^{-\beta/(2\beta+d)}).
$$
\end{theorem}

Condition (\ref{eq::need_clt})
holds for our discrepancy estimators
under weak conditions,
even when the model is irregular.
Typically,
$\text{pv}(\theta)$ is infinitely smooth.
In this case, we can take
$h\sim 1/\log n$
and we get
$$
P(\theta^* \in \hat C) = 1-\alpha + O_P(n^{-1/2}) + O_P( \sqrt{\log N/N}).
$$
The term
$O_P( \sqrt{\log N/N})$
is negligible as long as
$N > n/\log n$.

\cite{park_robust_2023},
showed that it is possible to use concentration
inequalities instead of 
the central limit theorem
which then requires essentially no conditions.
However, the central limit version suffices
for our purposes.

There is one problem when the model happens to be correct:
the variance of
$\hat \Delta(\theta^*, \hat\theta)$
may tend to 0 
faster than $O(1/n)$, which invalidates the central limit theorem.
\cite{verdinelli_decorrelated_2024}
showed that adding $1/n$ to the estimated variance
fixes the problem and yields
valid, albeit conservative, confidence intervals.

{\bf Remark:}
{\em We can reduce the randomness due to sample splitting by repeating 
the entire procedure at level $(1-\alpha/2)$ a large number of times $B$,
giving confidence sets
$C_1,\ldots, C_B$, and letting
$$
C = \Biggl\{ \theta:\ \frac{1}{B}\sum_b \mathbb I(\theta\in C_b) \geq 1/2\Biggr\}.
$$
Then by Markov's inequality
$P(\theta^* \in C)=
P \left( \frac{1}{B}\sum_b \mathbb I(\theta\in C_b) \geq 1/2 \right)
\leq
2 \frac{1}{B}\sum_b \E [ \mathbb I(\theta\in C_b)] = 2 (\alpha/2) = \alpha$
\citep{gasparin_merging_2024}.
}

\section{Robust SBI using Model Expansion}\label{sec::exponential_tilt}

Another approach to
model misspecification
is to expand the assumed model so that it is more flexible than the original model, to 
accommodate some misspecification.
We may then assume that the expanded model is correct, so that robust methods
are not required,
or, if we have evidence against this assumption -- for example if the goodness-of-fit
test in Section~\ref{sec::gof} is rejected -- our robust methods 
can also be applied.
In the latter case, 
there may be little benefit compared to applying the robust methods directly to $p_\theta$
and we only pursue the first case, where we regard
the expanded model as correct.

We consider a particular
model expansion, namely,
the exponential tilt
$$
p_{\theta,\beta}(x) =\frac{p_\theta(x) e^{\beta^T b(x)}}{c(\theta, \beta)}
$$
where
$b(x) =(b_1(x),\ldots, b_k(x))$
is a vector of fixed functions
and $c(\theta, \beta )$ is the normalizing constant,
$$
c(\theta, \beta) = \int p_\theta(x) e^{\beta^T b(x)} dx.
$$
Note that
$p_{\theta,\beta} = p_\theta$ when
$\beta = (0,\ldots, 0)^T$.
We assume that $k$ and
$(b_1(x),\ldots, b_k(x))$ are given.
An interesting extension is to use the data to choose these
but we do not pursue that here.
While there are many ways to expand a model,
the exponential tilt has some computational
advantages when doing SBI.
In particular, we will not need to
sample 
from $p_{\theta,\beta}$ for all combinations
of $\theta$ and $\beta$.

Then we base inference on the simulation based profile likelihood
${\cal L}(\theta) = \sup_\beta {\cal L}(\theta,\beta)$.
Let $\Theta\times B$ denote the parameter space for
$(\theta,\beta)$.
We use a two step
procedure where
we first find the maximizer $\hat\beta(\theta)$ of the
likelihood for each fixed
$\theta$
and then approximate the profile likelihood using SBI.
We estimate the profile likelihood
using only samples
from the $p_\theta$;
again, it's not necessary to sample from the expanded model
$p_{\theta,\beta}$.

For a fixed $\theta$, the likelihood for $\beta$ is
$$
{\cal L}_\theta(\beta)\equiv
{\cal L}(\theta,\beta) \propto
\frac{{\cal L}(\theta,\beta)}{{\cal L}(\theta,0)} =
\prod_i \frac{p_{\theta,\beta}(Y_i)}{p_{\theta,0}(Y_i)}
=\frac{e^{\beta^\top \sum_i b(Y_i)}}{c(\theta,\beta)^n},
$$
where $c(\theta,0)=1$, and since $Y_1(\theta),\ldots, Y_m(\theta)$ is a sample from
$p_\theta$,
we can estimate 
$c(\theta,\beta)=\int p_\theta e^{\beta^\top b}$ by
$$
\hat c(\theta,\beta) =
\frac{1}{m} \sum_i e^{\beta^T b(Y_i(\theta))}.
$$
We can thus estimate the log-likelihood 
for $\beta$ (for a fixed $\theta$) 
using only a sample from the model,  $Y_1(\theta), \ldots, Y_m(\theta)\sim p_\theta$,  
by
\begin{equation}
\hat\ell_\theta(\beta) =
n\beta^T \overline{b} - n \log
\left(\frac{1}{m} \sum_i e^{\beta^T b(Y_i(\theta))}\right) \label{eq:exp_tilt_log_to_max} 
\end{equation}
where $\overline{b} = n^{-1}\sum_i b(Y_i)$.
For each $\theta$, we maximize
over $\beta$
using Newton's method to obtain $\hat\beta(\theta)$; see Appendix~\ref{app::NR}.
Now we apply SBI as in Section~\ref{sec::SBI} 
to the model
$p_{\theta,\hat\beta(\theta)}$
to get 
the profile likelihood.
To do so, we would need to sample
from $p_{\theta,\hat\beta(\theta)}$.
Instead,
we reweight the existing sample
$Y_1(\theta),\ldots, Y_m(\theta)$ from $p_\theta$
with weights 
\begin{equation}
  w_i \propto 
\frac{p_{\theta,\hat\beta(\theta)}(Y_i(\theta))}
{p_{\theta}(Y_i(\theta))} \propto
e^{\hat\beta(\theta)^T b(Y_i(\theta))}.  \label{eq::resampl_dist_beta}
\end{equation} 
We can then resample with these weights
and apply Algorithm~\ref{alg:tilt}
or we can simply include these weights when we estimate the density ratio (\ref{eq::theloss}).
%\st{we do the density ratio estimation}.
%\V{where? how? Need equation number.}

\RestyleAlgo{ruled}
\begin{algorithm}[hbt!]
%\textcolor{red}{What is $p_{\tilde\theta}?$}
\caption{SBI profile likelihood for exponentially tilted model.}\label{alg:tilt}
\SetAlgoLined
\SetKwInOut{Input}{Input}
\SetKwInOut{Output}{Output}
\Input{${\cal Y_{\text{obs}}}= (Y_1, \dots, Y_n)$}
\Output{SBI profile log-likelihood ${\cal L}(\theta,\hat\beta(\theta))$}
sample $\theta_1, \dots, \theta_N \sim \pi$\\
\For{$j=1,\ldots, N$}{
  draw ${\cal Y}(\theta_j) = (Y_1(\theta_j),\ldots, Y_m(\theta_j)), \, Y_i(\theta_j)\sim p_{\theta_j}$\\
  find $\hat \beta_j \equiv \hat\beta(\theta_j)$
by maximizing the log-likelihood in~(\ref{eq:exp_tilt_log_to_max}) via Newton-Raphson (Appendix~\ref{app::NR}) 
}
generate a permutation of the index set $s=[I_1, \dots, I_N]$  \\
\For{$j=1,\ldots, N$}{
  draw ${\cal Y}_j \equiv \mathcal{Y} (\theta_j, \widehat\beta_j)\sim p_{\theta_j, \widehat{\beta}_j}$ 
  by resampling ${\cal Y}(\theta_j)$ at random with weights (\ref{eq::resampl_dist_beta}) and set $Z_j=1$ \\
  set $\mathcal{Y}_{N+j} = \mathcal{Y}_j$, $Z_{N+j}=0$, $\theta_{N+j}=\theta_{s_j}$ and $\hat \beta_{N+j} = \hat \beta_{s_j}$
} 
train $h(\mathcal{Y}, \theta, \hat \beta(\theta)) \equiv P(Z=1|{\cal Y}, \theta, \hat \beta(\theta))$ 
in~(\ref{eq::likelihood_trick}) using the
dataset $\{(Z_j, \mathcal{Y}_j, \theta_j, \hat \beta_j): 1\leq j \leq 2N\}$ \\
{\bf return} the estimated profile likelihood $\hat {\cal L}(\theta,\hat\beta(\theta))=\dfrac{\hat h(\mathcal{Y}, \theta, \hat\beta(\theta))}{1-\hat h(\mathcal{Y}, \theta, \hat\beta(\theta))}$ \vspace{.1cm} 
\end{algorithm}

If the expanded model is correctly specified, as we assume here, we can build valid  confidence sets for $\theta$ by inversion of hypothesis  tests of the form
$H_0:\theta=\theta_j$
using the profile likelihood
${\cal L}(\theta,\hat\beta(\theta))$
as the test statistic,
$T(\theta, \mathcal{Y}) = \dfrac{e^{\widehat\beta(\theta)^\top\sum_i b(Y_i)}}
{c(\theta, \widehat{\beta}(\theta))^n}.$
The procedure to obtain
the estimated p-value 
$\widehat{\text{pv}}(\theta_j, \mathcal{Y}_{obs})$
is the same as in Section~\ref{sec::SBI}, except that
we regress 
$B_1,\ldots, B_N$
on
$\theta_1,\ldots, \theta_N$
using the weights
$$
w_j = 
\frac{e^{\sum_{i=1}^m \hat\beta(\theta_j)^T b(Y_i(\theta_j))}}
{c(\theta_j,\hat\beta(\theta_j))}.
$$

\begin{comment}
\V{I want to remove what follows since the procedure is exactly the same as
in section 2, except we have weights}
The procedure is similar to that described in Section~\ref{sec::SBI}, and is described as follows. Let $\widehat\beta(\theta_1),\dots,\widehat\beta(\theta_N)$ be our estimate of the model expansion parameter vectors.
Define
$$
T(\theta_j, \mathcal{Y}) = \dfrac{e^{\widehat\beta(\theta_j)^\top\sum_i b(Y_i)}}
{c(\theta_j, \widehat{\beta}(\theta_j))^n}.
$$
For each $\theta_j$ we simulate 
$\mathcal{Y}_1(\theta_j),\dots, \mathcal{Y}_m(\theta_j)\sim p_{\theta_j}$  and let
$$
J_j =
I\Bigl\{ T(\theta_j,\mathcal{Y^*}) \geq T(\theta_j,\mathcal{Y})\Bigr\}.
$$
Now regress
$J_1,\ldots, J_N$
on
$\theta_1,\ldots, \theta_N$
with weights
$$
w_j = 
\frac{e^{\sum_{i=1}^m \hat\beta(\theta_j)^T b(Y_i(\theta_j))}}
{c(\theta_j,\hat\beta(\theta_j))}
$$
to obtain
the estimated p-value 
$\widehat{\text{pv}}(\theta_j, \mathcal{Y})$.
\textcolor{red}{Lorenzo: did you include the weights
into the nonparametric regression.}
The $1-\alpha$ confidence set is 
$$
C_n=\{\theta_j:\ \widehat{\text{pv}}(\theta_j, \mathcal{Y})>\alpha \}.
$$
\V{END REMOVE}
\end{comment}

(If the standard regularity conditions hold and sample size 
is large we can instead use the asymptotic approximation 
of the confidence set via Wilk's theorem
or use the cheap bootstrap approach described in 
Appendix~\ref{section::cheap}.)

\section{SBI Goodness of Fit Test \label{sec::gof}}
 
To assess the goodness-of-fit (GoF) of 
the model ${\cal P} =(p_\theta:\ \theta\in\Theta)$, we can test the null hypothesis
$H_0: d(P,{\cal P}) = 0$
where
$d(P,{\cal P}) = \inf_\theta d(P,P_\theta)$
and $d$ is some distance.
This could be, but need not be, one of the discrepancies we have
considered so far.
A p-value for this null is
$p=\sup_\theta p(\theta)$, where
\begin{equation}
\label{eq::gof.pv0}
p(\theta) = 
P_\theta ( T_n(\theta) \geq T_n),
\end{equation}
\begin{equation}
T_n(\theta) = \inf_\psi  d(P_\psi,P_n(\theta)),\ \ \ 
T_n =  \inf_\psi  d(P_\psi,P_n),
\end{equation}
$P_n$ is the empirical distribution of the observed data
${\cal Y}_{obs} = (Y_1,\ldots, Y_n)$ and
$P_n(\theta)$
is the empirical distribution of 
${\cal Y}(\theta) = (Y_1(\theta),\ldots, Y_n(\theta)), \ Y_j(\theta) \sim P_\theta$.
Performing this test requires that $P_\theta$ has a known closed form
and that the probability of the event
$\{T_n(\theta) \geq T_n\}$
can somehow be computed or approximated
with an asymptotic approximation.
We can avoid these requirements
using the SBI framework.

As usual, we assume that for each sampled value $\theta_j$ we have a sample
${\cal Y}(\theta_j) = (Y_1(\theta_j),\ldots, Y_n(\theta_j)), \ Y_j(\theta_j) \sim P_{\theta_j}$.
For each $\theta_j$ we draw a second, independent sample
$Y_1^*(\theta_j),\ldots, Y_M^*(\theta_j) \sim P_{\theta_j}$
where $M$ is much larger than $n$. We let  
$P_M^*(\theta_j)$ denote its
empirical distribution and we approximate $P_{\theta_j}$ by $P_M^*(\theta_j)$.
Then we approximate the inf with respect to $\psi$ in (\ref{eq::gof.pv0}) by minimization over the grid
of values of $\theta$.
Specifically, define
\begin{equation}
\hat T_n(\theta) = \min_s d(P_M^*(\theta_s),P_n(\theta)), \ \ \ \hat T_n = \min_s d(P_M^*(\theta_s),P_n), 
\end{equation}
\vspace{-.65cm}
\begin{align}
\label{eq::gof.p} 
\hat p(\theta) &= 
\frac{\sum_r K_h(\theta_r - \theta) I( \hat T_n(\theta_r) \geq \hat T_n)}
{\sum_r K_h(\theta_r - \theta)},\\
\label{eq::gof.pv} \hat p &= \max_j \hat p(\theta_j).
\end{align}

Formally, in this SBI setting,
the null hypothesis that we test is
$H_0: P\in (P_\theta:\ \theta\in C)$
where $C = \{\theta_1,\ldots,\theta_N\}$
are the sampled values.

\begin{theorem} \label{thm::4}
Suppose that:

(1) $\Theta$ is compact and $\pi(\theta)$ is strictly positive.

(2) $\max_{\theta\in C} |d(P_M^*(\theta),P_n)- d(P_\theta,P_n)| = O_P(\sqrt{\log M/M})$
and
$\max_{\theta\in C} |d(P_M^*(\theta),P_n(\theta))- d(P_\theta,P_n(\theta))| = O_P(\sqrt{\log M/M})$. %\lortomas{From proof of Thm 9 both suprema are of the order $O_P(\sqrt{\log M/M})$, while in this assumption the second term is of the order $O_P(\sqrt{\log n/n})$.}

(3) The function $p(\theta)$ is in Holder $(\beta)$.

(4) The functions $d(P_M^*(\theta),P_n)$ and 
$d(P_M^*(\theta),P_n(\theta))$ are Lipschitz in $\theta$.

(5) For some $\xi$, we have that,
uniformly over $\theta$, $T_n(\theta)-T(\theta)$ has a density 
that is $O_P(n^\xi)$ in a neighborhood of 0.

(6) $M \geq \max\{n,N\}$.

Then, if $H_0$ is true,
$$
\P( \hat p > \alpha) \leq \alpha + 
O_P(h^\beta + (Nh^d)^{-1/2}) + O_P\left((1+n^\xi)\Bigl(O_P(\sqrt{\log N/N}) + \sqrt{\log M/M}\Bigr)\right).
$$
If we set the optimal kernel bandwidth the first term is
$O_P\left(N^{-\beta/(d+2\beta)}\right)$.
\end{theorem}

Note that condition (5) allows for the fact that
$T_n$ and $T_n(\theta)$ could concentrate around 0.

For illustration,
we take $d$ to be the Wasserstein distance.
If $P$ and $Q$ are distributions,
the 2-Wasserstein distance $W(P,Q)$ is defined by
$$
W^2(P,Q) = \inf_J \E_J[ ||Y-X||^2]
$$
where $(X,Y)\sim J$
and the infimum is over all joint distributions $J$
with marginals $P$ and $Q$.
This is an interesting choice since it has been used with success
but the asymptotic justification
for computing the p-value is still
an open question \citep{hallin}.
Our approach avoids this issue.

For one-dimensional distributions it can be shown that
\begin{equation}
    W^2(P,Q) = \int | F^{-1}(u) - G^{-1}(u)|^2 du \label{eq::w1d_qntl}
\end{equation}
where $F$ and $G$ are the cdf's of $P$ and $Q$.
This distance has many appealing properties.
%which makes it very useful.
In particular, it is sensitive to the geometry of the sample space,
which is not true of many other distances.
For example, the Wasserstein distance between
a point mass at $y_1$ and a point mass at $y_2$ is
$||y_1-y_2||$
whereas distances like the total variation, Hellinger or Kolmogorov-Smirnov distance
have a value that does not depend on the distance between $y_1$ and $y_2$.
See \cite{chewi2024} for a review.
Also see \cite{hallin} 
who suggested using
Wasserstein-based goodness-of-fit tests. 
Now if we insert $W$ for $d$ in the above method,
we get a valid test without regularity assumptions or
asymptotic approximations.
As noted in \cite{hallin},
the limiting distribution under the null is not known
so SBI plays an especially important role in this case.

\begin{lemma}
\label{lemma::Lip_Wass}
Suppose that
$\sup_\theta \int ||y||^q dP_\theta(y) < \infty$ for some $q>2$
and that
$\sup_\theta\int e^{\gamma ||y||^\alpha} dP_\theta(y) < \infty$
for some
$\gamma>0$ and $\alpha>2$.
Also assume that
the function
$W(P_\theta,Q)$ is Lipschitz in $\theta$.
For each $\theta_j$ let
$P_M(\theta_j)$ denote the empirical distribution
based on
$Y_1(\theta_j),\ldots, Y_M(\theta_j)$.
Then
$$
|\min_j W(P_M(\theta_j),Q) - \inf_\theta W(P_\theta,Q)| =
O_P\left(\sqrt{\log N/N}\right)+O_P\left(\sqrt{\log M/M}\right).
$$
\end{lemma}

\begin{theorem}
\label{thm::6}
Suppose that
$\sup_\theta \int ||y||^q dP_\theta(y) < \infty$ for some $q>2$
and that
$\sup_\theta\int e^{\gamma ||y||^\alpha} dP_\theta(y) < \infty$
for some
$\gamma>0$ and $\alpha>2$.
Also suppose that, for every $Q$,
$W(P_\theta,Q)$ is Lipschitz in $\theta$. For each $\theta_j$ let
$P_M(\theta_j)$ denote the empirical distribution
based on
$Y_1(\theta_j),\ldots, Y_M(\theta_j)$ and let $P_n$ be the empirical distribution of the data $Y_1, \dots, Y_n$.
Then
$$
\sup_\theta |W(P_M^*(\theta),P_n)-W(P_\theta,P_n)| = O_P(\sqrt{\log M/M}).
$$
\end{theorem}

The 2-Wasserstein distance is popular
and is known to have many appealing properties
\citep{chewi2024}.
But it is computationally expensive to estimate when
dim($\theta) > 1$.
We can instead use any other distance such as
the Kolmogorov-Smirnov (KS) statistic
$d(P,Q) = \sup_x |F(x)-G(x)|$.

\section{Applications}\label{sec::applications_iid}

In Section~\ref{sec::tilt} we illustrate the model expansion idea
in Section~\ref{sec::exponential_tilt} to handle 
model misspecification.
The next three examples concern the discrepancy based projection 
methods developed in Sections~\ref{sec::robust_sbi} and~\ref{sec::cs_misspecified}.
In Section~\ref{sec::gaussian} we use a simple
two-dimensional parameter example
to show that these
methods produce confidence sets 
that have the correct coverage 
whether or not the assumed model $P_\theta$ is correctly specified.
In 
Section~\ref{sec::gandk} we 
conduct robust SBI inference for the four-dimensional parameter 
of the intractable G-and-K distribution, and
in Section~\ref{sec::GMM} we illustrate that the 
projection methods produce valid 
confidence sets in a case of unidentifiable parameters, 
when standard asymptotic methods cannot apply.
Appendix~\ref{app::add} contains an additional example.
We finish by applying the SBI goodness of fit test 
developed in Section~\ref{sec::gof}
to three simulated data examples.

\subsection{Robust SBI via Model Expansion -- Tilted Gaussian Location Parameter}
\label{sec::tilt}

We illustrate the expansion method
with a simple, proof of concept, example.
We generated $n=5000$ data points 
from 
\begin{equation}
\label{eq::exponential_tilt_normal}
    p(x) \propto {\cal N} (\theta, \sigma^2) \,e^{(\alpha_1 x^3\cdot I_{\{|x^3|< \tau\}} +\alpha_2 x^4)}
\end{equation} 
with
     $\alpha = (0.025, -0.0025)$, and ${\cal N} (\theta, \sigma^2)$
     the normal distribution with $\theta=2.5$ and $\sigma=2$.
     The cubic term was truncated at
     $\tau =10^3$ to avoid exploding tails.
     The data is shown as the orange histogram 
in Fig.~\ref{fig:exponential_tilt_normal_profile_likelihood}(c).
     The target parameter is $\theta$.
     We assume model $p_\theta(x) = {\cal N}(\theta, 4)$,
     which is clearly inadequate. 
     %\VV{The MLE is $\hat \theta = XXXXX$ and asymptotic 95\% confidence interval is XXXXX, which does not contain the true value.}
We expand $p_\theta(x)$ to account for model misspecification:
\begin{equation}
    p_{\theta,\beta}(x) \propto p_\theta(x) e^{(\beta_1 x^3 +\beta_2 x^4)}.
    \label{eq::exponential_tilt_norm}
\end{equation} 
Fig.~\ref{fig:exponential_tilt_normal_profile_likelihood}(b) shows the values of $\beta_1(\theta)$ and $\beta_2(\theta)$ that maximize the
likelihood for each $\theta$ and Fig.~\ref{fig:exponential_tilt_normal_profile_likelihood}(a) shows the SBI profile likelihood obtained
by Algorithm~\ref{alg:tilt}. The MLE $\hat \theta$ is close to the true value and the 
tilted density (\ref{eq::exponential_tilt_norm}) with parameters $\widehat\theta$ and $\hat \beta(\hat \theta)$, overlaid 
in Fig.~\ref{fig:exponential_tilt_normal_profile_likelihood}(c), matches the observed data well.
(This is not surprising since (\ref{eq::exponential_tilt_normal}) and (\ref{eq::exponential_tilt_norm}) are very close.)

\begin{figure}[t!]
    \centering
    \subfloat[\sl Estimated profile log-likelihood $\hat {\cal L}(\theta,\hat\beta(\theta))$.
]{\includegraphics[width=4.3cm]{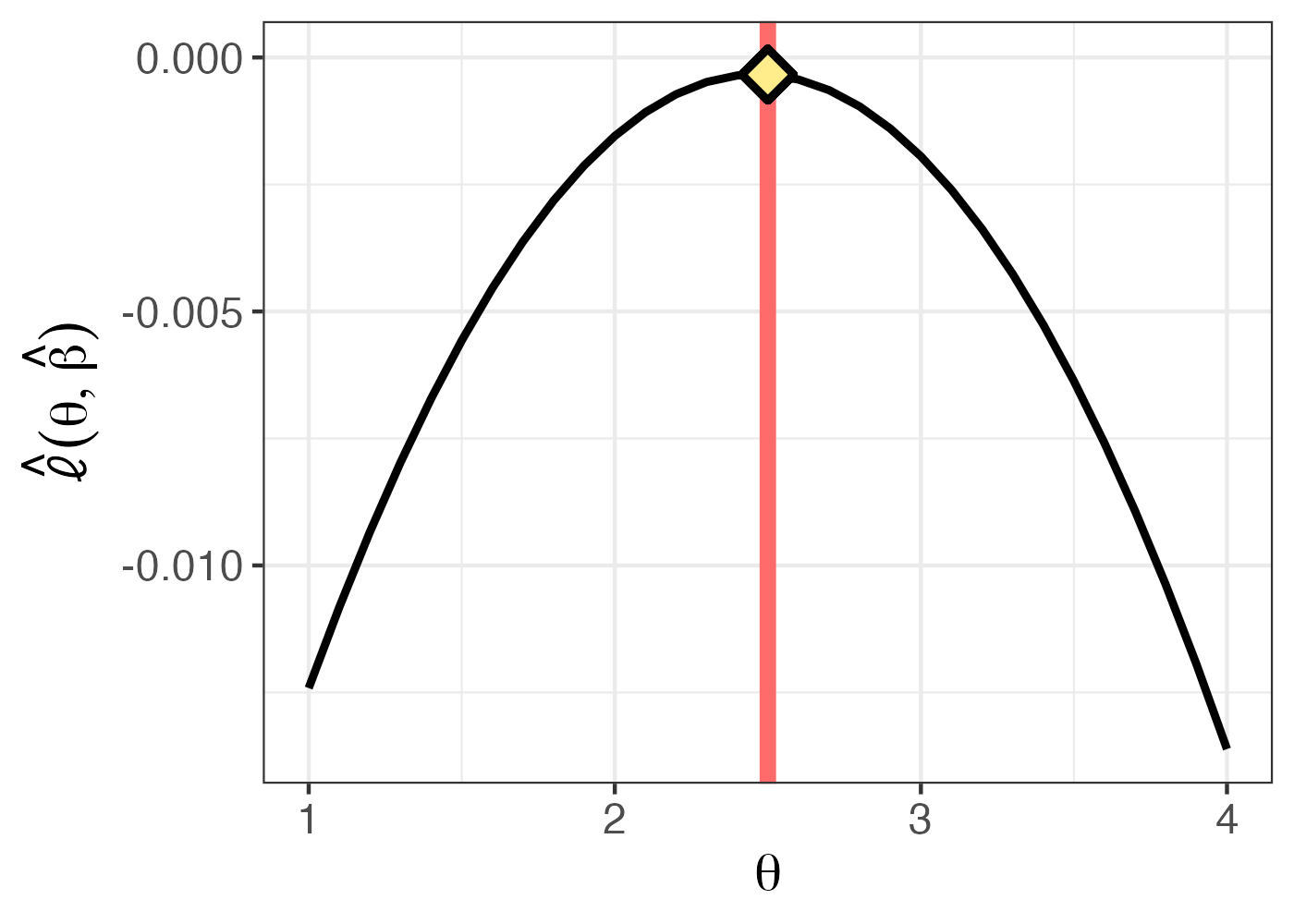}} 
    \qquad
    \subfloat[\sl ML estimates $\hat \beta_1 (\theta)$ and $\hat \beta_2(\theta)$ as functions of $\theta$. ]{
    \includegraphics[width=4.3cm]{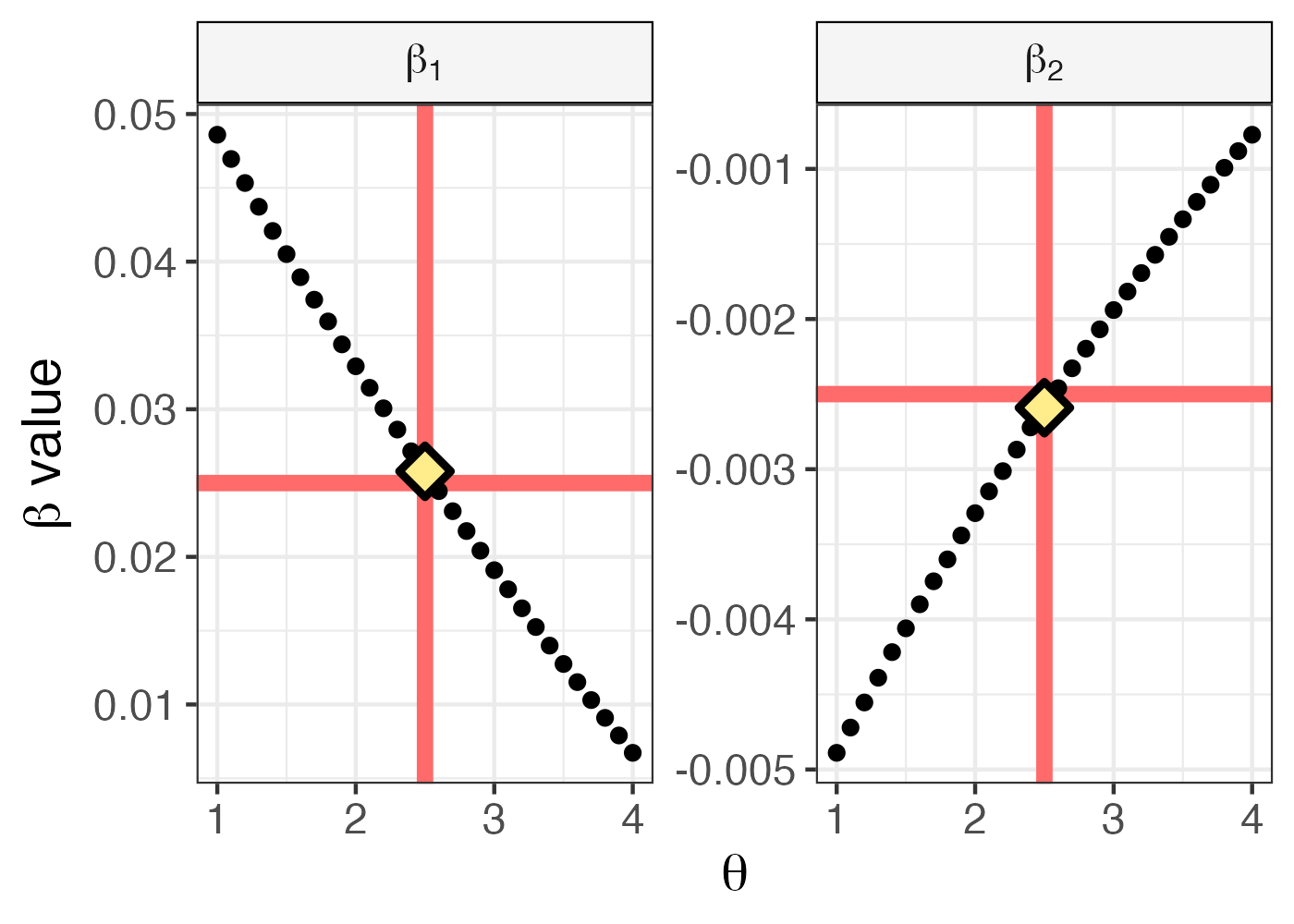}} 
    \qquad
    \subfloat[\sl Histograms of the data with fitted expanded model $p_{\hat\theta, \hat{\beta}(\hat\theta)}$. ]{
    \includegraphics[width=4.3cm]{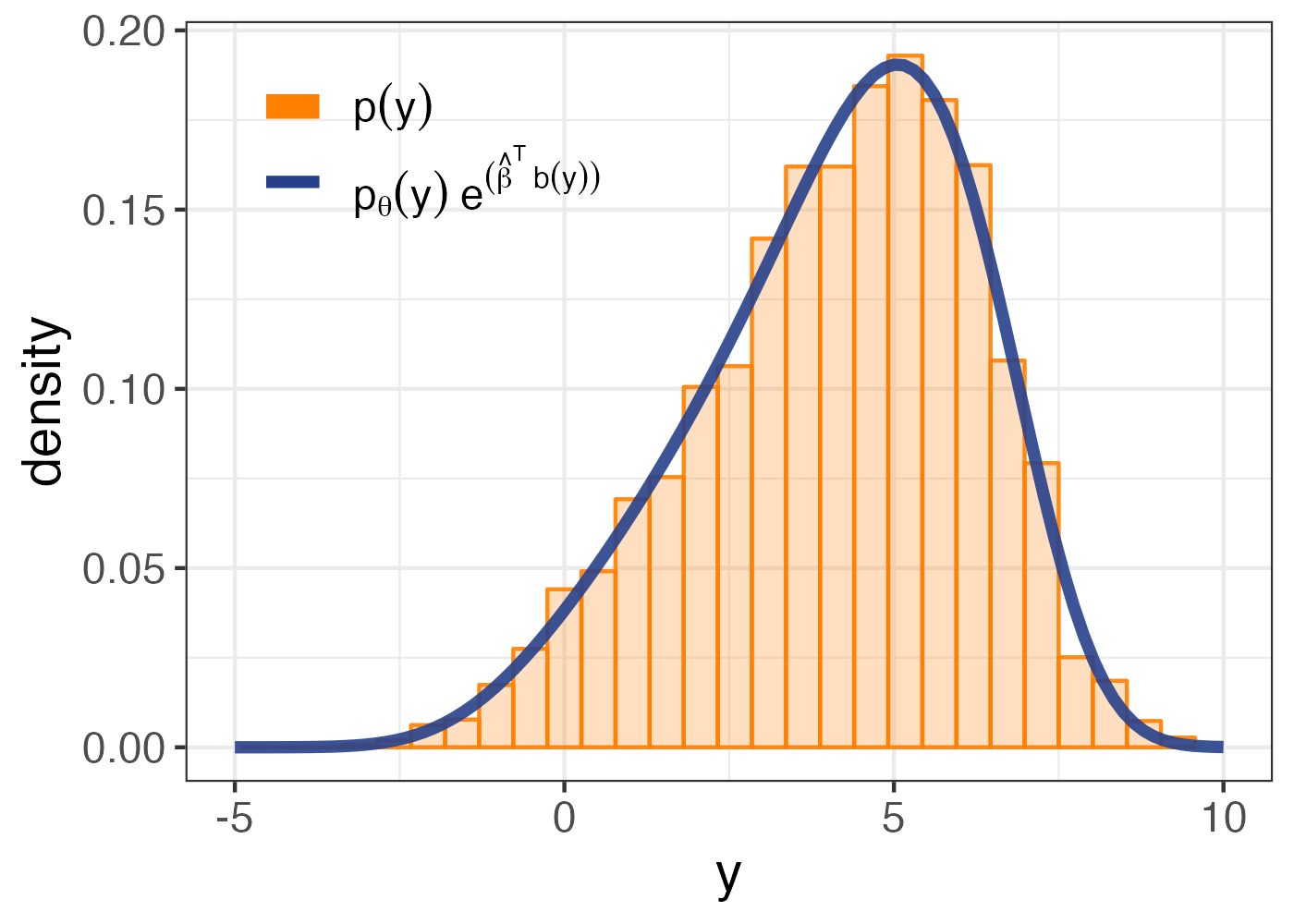}} \\ 
    \caption{\sl {\bf Robust inference via model expansion.} 
     The true model is (\ref{eq::exponential_tilt_normal}) with
     $\alpha_1 = 0.025$, $\alpha_2 = -0.0025$, and ${\cal N} (\theta, \sigma^2) = {\cal N} (2.5, 4)$.
    The target parameter is $\theta$.
     The assumed model is $p_\theta(x) = {\cal N}(\theta, 4)$.
     The expanded model is in (\ref{eq::exponential_tilt_norm}).
     Red lines and gold diamonds indicate true and estimated parameters, respectively.
}
\label{fig:exponential_tilt_normal_profile_likelihood}
\end{figure}

\subsection{Robust SBI via Projection -- Gaussian Location and Scale Parameters}
\label{sec::gaussian}

\begin{figure}[t!]
    \centering
    \subfloat[Gaussian location.  True parameter values $\mu^*=2.5$ $\sigma^*=1$]{
    \includegraphics[width=.7\linewidth]{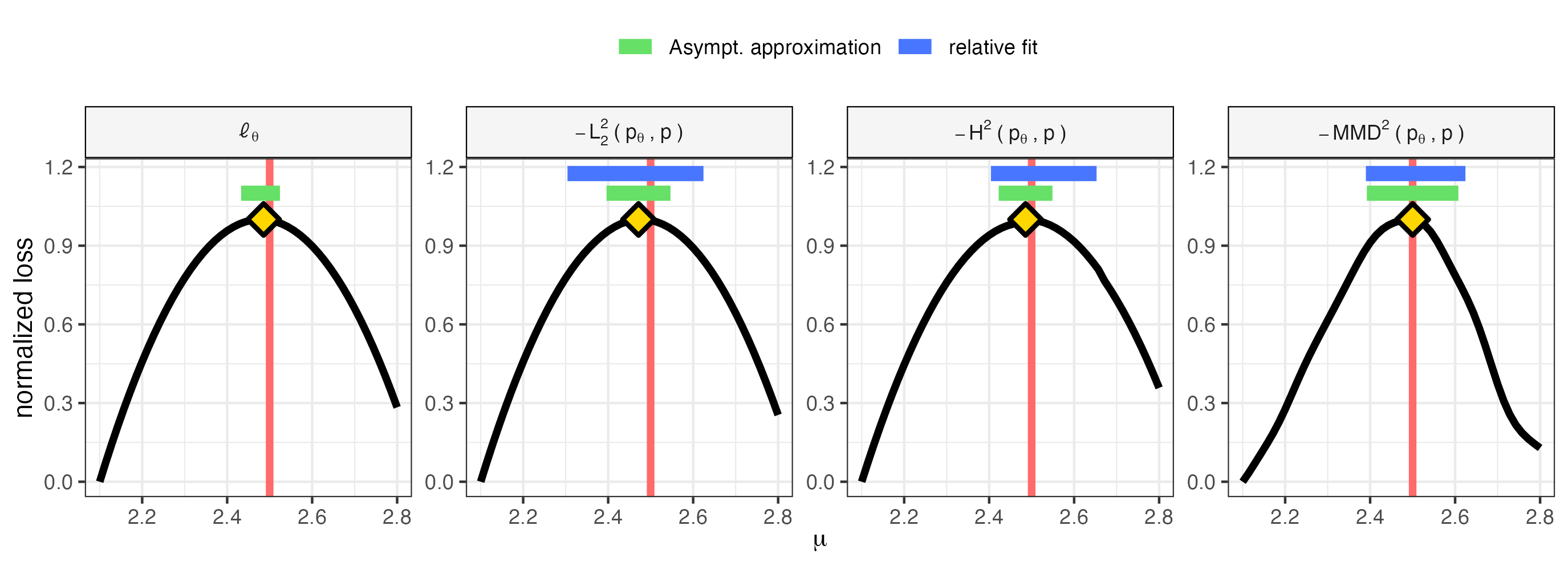}}
    \\
    \subfloat[Gaussian scale. True parameter values $\mu^*=2.5$, $\sigma^*=1$]{
    \includegraphics[width=.7\linewidth]{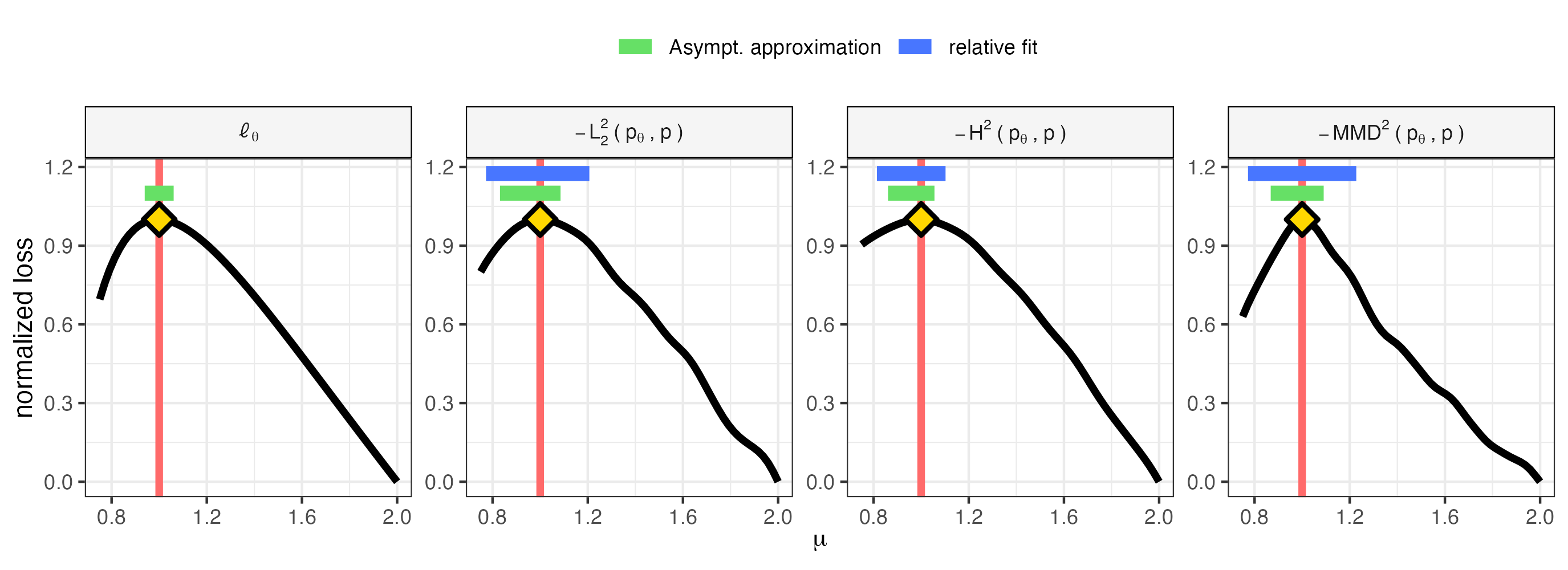}}
     \caption{
     \sl{{\bf SBI for Gaussian location and scale when the model $p_\theta$
     is correctly specified.}
     All discrepancies produce estimates (gold diamonds) close to the true values (red lines).
     Relative fit confidence sets (Section~\ref{sec::cs_misspecified}) are in blue.
     Asymptotic confidence sets (green) for $L_2$ and Hellinger divergences use a sandwich 
     estimator (Section~\ref{sec::cs_misspecified}); for the log-likelihood we inverted 
     the likelihood ratio test;
     for the MMD we applied the theoretical derivations in \cite{briol_statistical_2019}.   
     %\V{The curves need to be smoothed, to match the other figures.}\lortomas{Done}
     }}
\label{fig:gaus_loc_scale}

\bigskip

\centering
%\sisetup{parse-numbers = false}
\begin{tabular}{
  l 
  *{10}{S[table-format=1.5]}
  S[table-format=5]
  %S[table-format=3.2]
}
\toprule
\multicolumn{1}{c}{Discrepancy} &
  \multicolumn{4}{c}{Location parameter ($\mu$)}  &
  \multicolumn{4}{c}{Scale parameter ($\sigma$)} \\
\cmidrule(lr){2-3} \cmidrule(lr){4-5} \cmidrule(lr){6-7}\cmidrule(lr){8-9}  
  & \multicolumn{2}{c}{Asympt. Approx.} & \multicolumn{2}{c}{Relative fit} & \multicolumn{2}{c}{Asympt. Approx.} &  \multicolumn{2}{c}{Relative fit} \\
  & \multicolumn{1}{c}{Coverage} &  \multicolumn{1}{c}{Length} & \multicolumn{1}{c}{Coverage} &  \multicolumn{1}{c}{Length} & 
  \multicolumn{1}{c}{Coverage} &  \multicolumn{1}{c}{Length} &
  \multicolumn{1}{c}{Coverage} &  \multicolumn{1}{c}{Length} \\
\midrule
\midrule
Likelihood  & \multicolumn{1}{r}{0.95 \small $\pm$ .04} & \multicolumn{1}{c}{0.08} & {-} & {-} & \multicolumn{1}{c}{0.94 \small $\pm$ .05} & \multicolumn{1}{c}{0.03} & {--} & {--} \\
%& \multicolumn{1}{c}{\small (0.89, 0.99)} & \multicolumn{1}{c}{\small (0.05, 0.05)} & \multicolumn{1}{c}{--} & \multicolumn{1}{c}{--} &\multicolumn{1}{c}{\small (0.92, 1)} & \multicolumn{1}{c}{\small (0.05, 0.06)} &\multicolumn{1}{c}{--} & \multicolumn{1}{c}{--} \vspace{.2cm} \\
Hellinger & \multicolumn{1}{r}{0.98 \small $\pm$ .03} &  \multicolumn{1}{c}{0.11} & \multicolumn{1}{r}{1 \small $\pm$ .00} &\multicolumn{1}{c}{0.27}  & \multicolumn{1}{r}{0.93 \small $\pm$ .05} & \multicolumn{1}{c}{0.08} & 
\multicolumn{1}{r}{1 \small $\pm$ .00} & \multicolumn{1}{c}{0.27}\\
%& \multicolumn{1}{c}{\small --} & \multicolumn{1}{c}{\small (0.18, 0.18)} & {\small (0.95,1)} & {\small (0.22,  0.24)} & {--}  & {\small (0.40, 0.41)} &\multicolumn{1}{c}{--} &{\small ( 0.46, 0.50)}  \vspace{.2cm}\\ 
$L_2$  & \multicolumn{1}{r}{0.98 \small $\pm$ .03} & \multicolumn{1}{c}{0.14} & \multicolumn{1}{r}{1 \small $\pm$ .00} &\multicolumn{1}{c}{0.32} & \multicolumn{1}{r}{1 \small $\pm$ .00} & \multicolumn{1}{c}{0.21} & \multicolumn{1}{r}{1 \small $\pm$ .00} & \multicolumn{1}{c}{0.51}  \\ 
%& {\small --} & \multicolumn{1}{c}{\small (0.23, 0.23)} & {\small (0.94, 1)} & {\small (0.27, 0.27)} & \multicolumn{1}{c}{--} & \multicolumn{1}{c}{\small (1.15,  1.21)} & \multicolumn{1}{c}{\small (0.97,1)} & {\small ( 1.15, 1.18)}  \vspace{.2cm}\\ 
MMD & \multicolumn{1}{c}{0.97 \small $\pm$ .03} & \multicolumn{1}{c}{0.20} & \multicolumn{1}{r}{0.97 \small $\pm$ .03} & \multicolumn{1}{c}{0.31}& \multicolumn{1}{c}{0.95 \small $\pm$ .04} & \multicolumn{1}{c}{0.14} & \multicolumn{1}{r}{1 \small $\pm$ .00} & \multicolumn{1}{c}{0.33} \\
%& \multicolumn{1}{c}{\small (0.89, 0.99)} & \multicolumn{1}{c}{--} & \multicolumn{1}{c}{--} & {\small (0.28, 0.32)}&  \multicolumn{1}{c}{\small (0.84, 0.96)} & \multicolumn{1}{c}{--} & \multicolumn{1}{c}{--}& {\small (0.51, 0.59)}\\
\bottomrule
\end{tabular}
\captionof{table}{{\bf SBI for Gaussian location and scale when the model $p_\theta$
     is correctly specified.}
\sl Empirical coverages with 95\% simulation bounds and average lengths (accurate up to two digits)
of the 95\% confidence sets in Figure~\ref{fig:gaus_loc_scale} 
in 100 repeat simulations.
}
\label{table::results_gaus_location_scale}
\end{figure}

\begin{figure}[t!]
    \centering
%\begin{figure}[ht!]
    \centering
    %\subfloat[Likelihood]{\includegraphics[width=4.9cm]{figures/exponential_tilt/L2_loss/beta_example/plot_L2_loss_minimizer.png}} 
    \subfloat[L2]{\includegraphics[width=4cm]{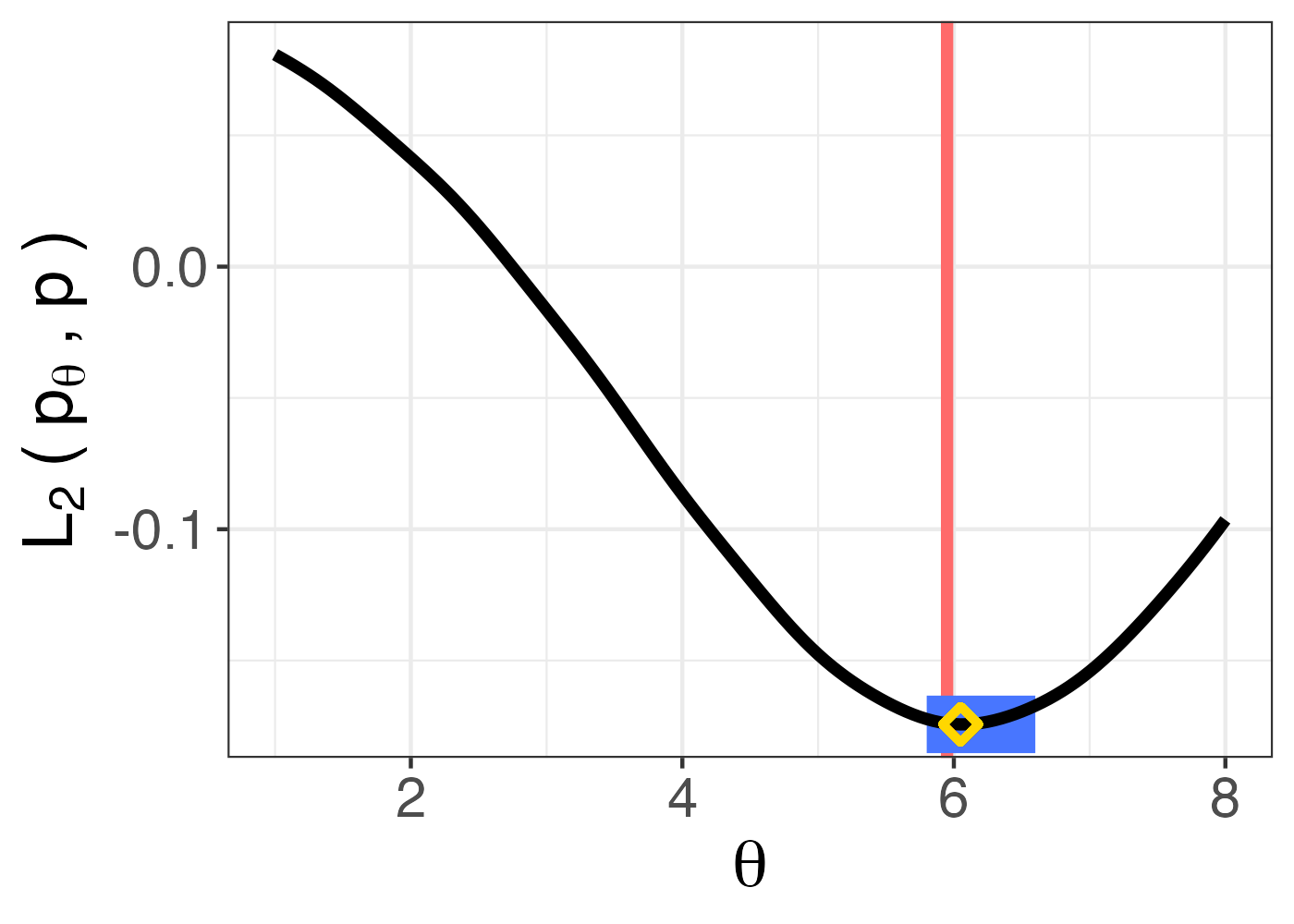}} 
    \subfloat[Hellinger]{\includegraphics[width=4cm]{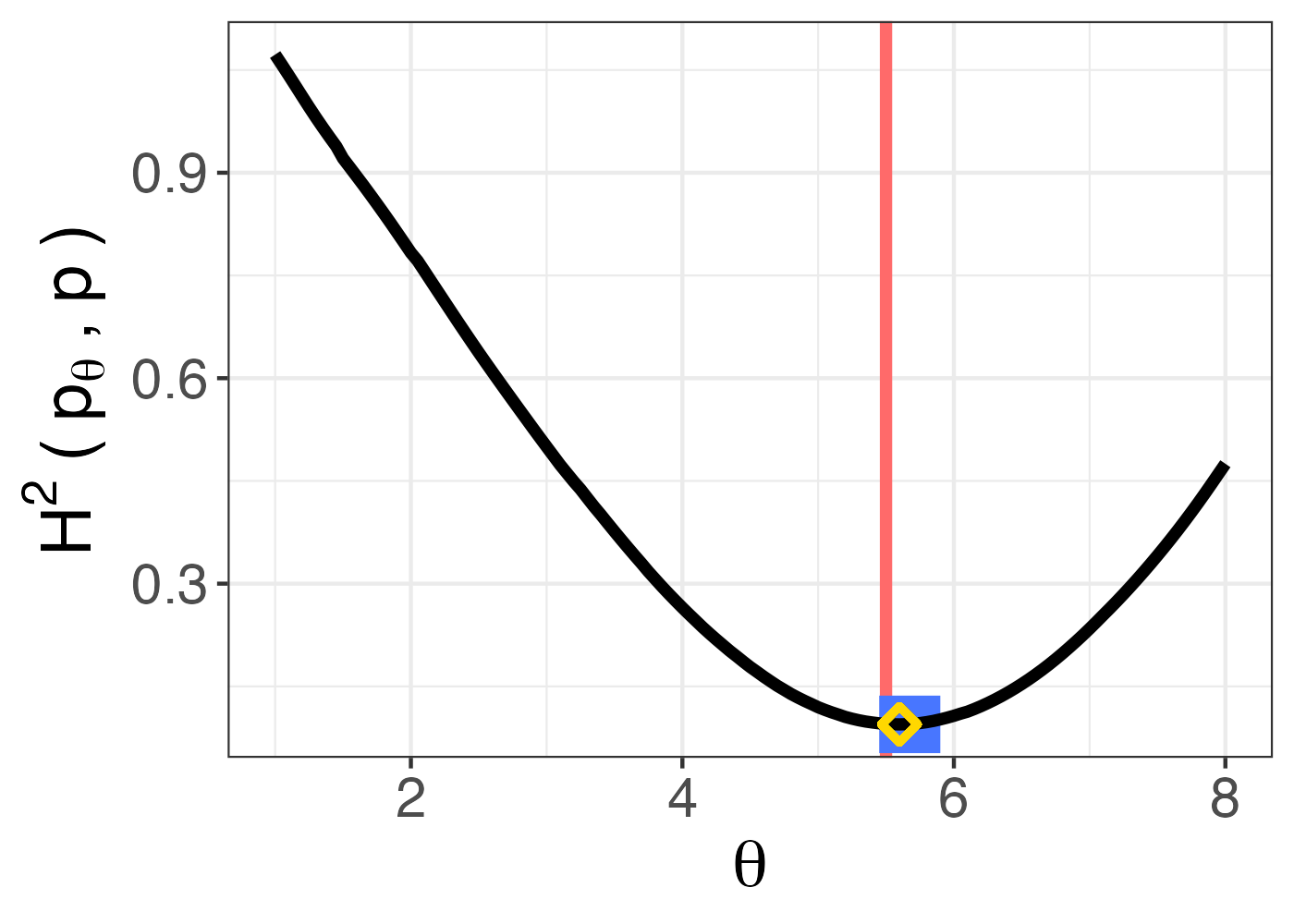}} 
    \subfloat[MMD]{\includegraphics[width=4cm]{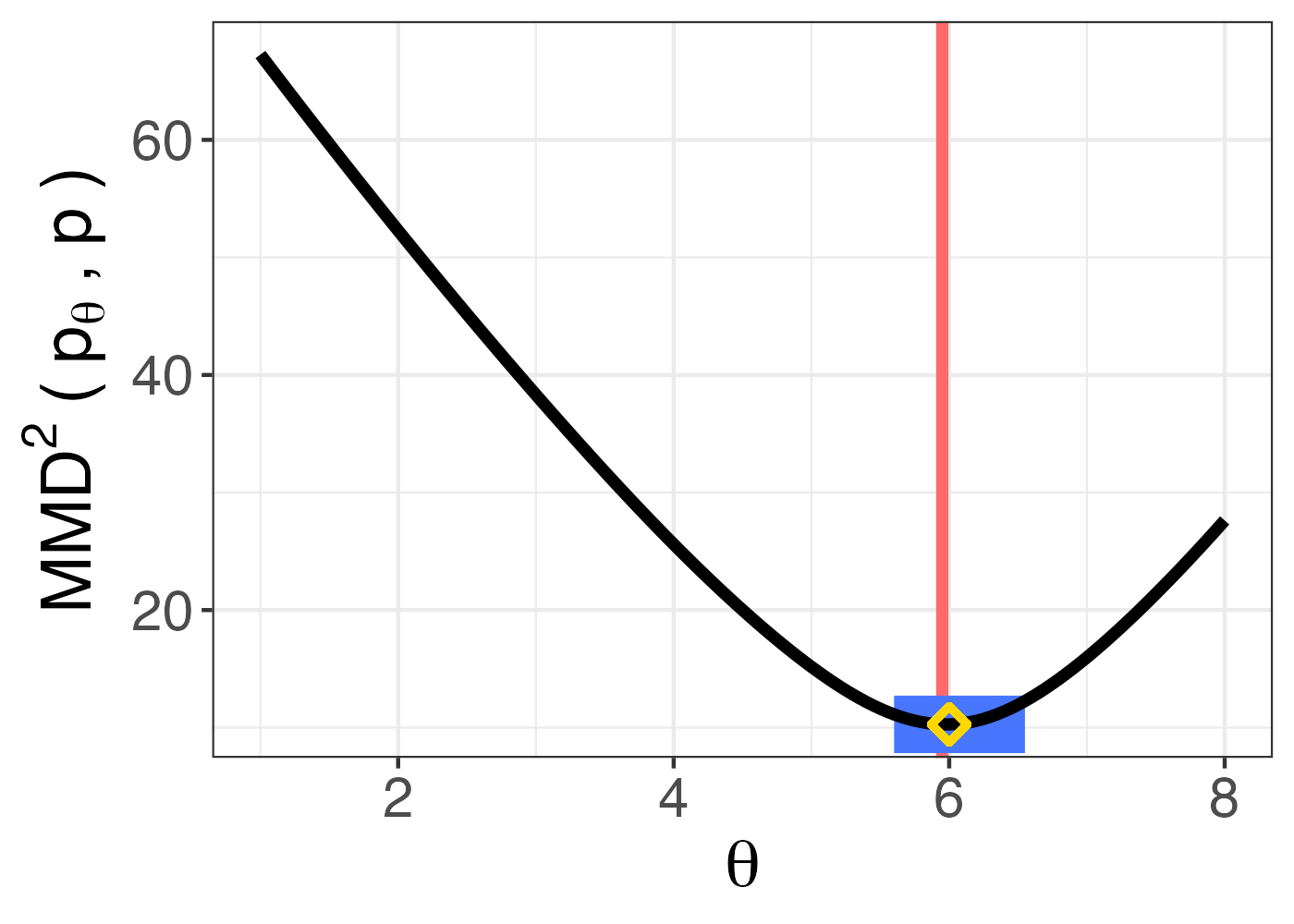}} 
    \caption{\sl{{\bf SBI for Gaussian location under model misspecification.} 
    The data has exponential 
    tilted normal density~(\ref{eq::exponential_tilt_normal}) with $\theta=2.5$, 
    $\sigma=2$ and $\alpha=(0.05, -0.005)$. The assumed model is $p_\theta = {\cal N}(\theta, 2.5^2)$.
    The plot shows the projection parameters for each discrepancy (red), the estimates (gold diamond)
    and the relative fit confidence sets (blue). 
    %The discrepancy produce accurate estimates of the projection parameter and the confidence sets are nearly centered around these values.
}}
\label{fig:gaussian_model_misspecification}

 \bigskip
 
\centering
%\sisetup{parse-numbers = false}
\begin{tabular}{
  l 
  *{10}{S[table-format=1.5]}
  S[table-format=5]
  %S[table-format=3.2]
}
\toprule
\multicolumn{1}{c}{Discrepancy} &
\multicolumn{1}{c}{Length} &
  \multicolumn{4}{c}{Coverage} \\
\cmidrule(lr){3-6} 
& & \multicolumn{1}{c}{$\theta^{proj}_{H}$} &\multicolumn{1}{c}{$\theta^{proj}_{L_2}$} &\multicolumn{1}{c}{$\theta^{proj}_{MMD}$} & \\
\midrule
\midrule
KL (likelihood)  & \multicolumn{1}{c}{0.06 \small $\pm$ .00} & \multicolumn{1}{c}{0.16 \small $\pm$ .10} & {0 \small $\pm$ .00}  & {0 \small $\pm$ .00}\\
%& \multicolumn{1}{c}{(0.064, 0.066)} & \multicolumn{1}{c}{(0.055, 0.265)} & {(0,0)}  & {(0,0)}\vspace{.2cm} \\
Hellinger & \multicolumn{1}{c}{0.39 \small $\pm$ .02} &  \multicolumn{1}{r}{0.98 \small $\pm$ .04} & {--}  & {--} \\
%& \multicolumn{1}{c}{(0.356, 0.374)} & \multicolumn{1}{c}{(1,1)} & {--} & {--}\vspace{.2cm}\\ 
$L_2$  & \multicolumn{1}{c}{0.52 \small $\pm$ .03} & \multicolumn{1}{c}{--} & \multicolumn{1}{c}{0.98 \small $\pm$ .04}  & {--} \\ 
%& {(1.186, 1.252)} & \multicolumn{1}{c}{--} & \multicolumn{1}{c}{(1,1)}  & {--} \vspace{.2cm}\\ 
MMD & \multicolumn{1}{c}{0.68 \small $\pm$ .03} & \multicolumn{1}{c}{--} & {--} & {1 \small $\pm$ .00}  \\
%& \multicolumn{1}{c}{(0.648, 0.722)} & \multicolumn{1}{c}{--} & {--} & {(1,1)} \\
\bottomrule
\end{tabular}
\captionof{table}{
\sl{{\bf SBI for Gaussian location under model misspecification.} 
\sl Empirical coverages with 95\% simulation bounds and average lengths (accurate up to two digits)
of the 95\% relative fit confidence sets in Figure~\ref{fig:gaussian_model_misspecification} 
in 50 repeat simulations.
Results for confidence sets obtained by inversion of the 
likelihood test are provided for comparison.
%The inversion type confidence sets are not robust to model misspecification, while discrepancy-based confidence sets attain the desired nominal levels.
}}
\label{table::results_misspecified_exponential_tilt}

\end{figure}

We now illustrate that the projection methods yield confidence sets 
that have the correct coverage 
whether or not the assumed model is correctly specified. 
We 
begin with the correctly specified case.
We generated a sample of size $n=2000$ from the Gaussian distribution $P = \mathcal{N}(\mu, \sigma^2)$, 
$\theta=\left(\mu, \sigma\right)$ being the target of inference.
We assumed that the model $P_{\theta}$ was Gaussian.
We obtained MLEs and discrepancy-based estimates for $\theta$ (Section~\ref{sec::robust_sbi}),
and confidence sets using asymptotic approximations and the relative fit approach (Section~\ref{sec::cs_misspecified}).
Figure~\ref{fig:gaus_loc_scale} shows the discrepancies, parameter estimates and confidence sets. 
All discrepancies produced estimates that are close to the true $\theta$ 
and confidence sets that cover it.
%Notice that the estimated Hellinger discrepancy is choppy because it is difficult to estimate 
%the density ratio.\lortomas{do we still need the previoussentence as comment?It is true but it may not correspond to current plot.} The MMD discrepancy is also variable due to sample splitting (the effective sample size is halved). 
Table~\ref{table::results_gaus_location_scale} shows the empirical coverages of 
the confidence sets in 100 repeat simulations.
All confidence sets achieve or exceed the nominal 95\% coverage for all the discrepancies.
The MDPD ($L_2$ in this case) yields wider confidence set than Hellinger, which agrees with its lower 
theoretical efficiency.

We now turn to the misspecified case. 
We generated data 
from the tilted Gaussian distribution (\ref{eq::exponential_tilt_normal})
with $\theta = 2.5$, $\sigma=2$ and $\alpha=(0.05, -0.005)$.
The target of inference is $\theta$.
We assumed the Gaussian model $P_\theta$ with unknown mean $\theta$.
The projection parameter is the first component of $\argmin_{\theta,\beta} d(p,p_{\theta,\beta})$.
Fig.~\ref{fig:gaussian_model_misspecification} shows the  discrepancies with 
estimated parameters and relative fit confidence sets. 
All discrepancies produce estimates that are close to the projection 
parameter, with confidence sets nearly centered around it.
Table~\ref{table::results_misspecified_exponential_tilt} contains the empirical 
coverages and average lengths of the confidence sets over 50 
repeat simulations.
All discrepancy-based confidence sets achieve valid coverage. Their widths align with theoretical expectations --
the Hellinger discrepancy is more efficient and thus yields shorter confidence sets compared to the $L_2$ discrepancy.
Likelihood-based confidence sets, while narrower, 
fall short of the desired 95\% coverage levels.

\subsection{Robust SBI for Intractable Likelihood -- G-and-k Distribution} \label{sec::gandk}

The g-and-k distribution cannot be written in closed form, but its quantiles are available 
so it can be simulated from using inverse CDF sampling, 
making it a prime candidate for SBI inference.
The quantiles are \citep{rayner_numerical_2002,prangle_gk_2017}:
\begin{equation}
    q_\theta(p) = l + s\cdot\left[1 + c \cdot\tanh\left(\frac{g \cdot\phi(p)}{2}\right)\right]\phi(p)(1 + \phi(p)^2)^k
\end{equation}
where $\phi(p)$ is the quantile function of the standard normal distribution, 
$c=0.8$, and the parameters $\theta=\left(l, s, g, k\right)$ 
determine the location ($l$) scale ($s$) skewness ($g$) and kurtosis ($k$).

We simulated $n=2000$ observations from the g-and-k distribution and performed robust SBI inference 
based on discrepancies. Figure~\ref{fig:gk_distribution_inference} displays parameter 
and relative fit confidence set estimates obtained using the three discrepancies.
Inference for skewness and kurtosis is notably more challenging than for location and scale, 
as indicated by the wider confidence sets, 
flatter profile loss functions and by the fact that the Hellinger %two out of 
%three 
metric %($L_2$ and Hellinger) 
yields estimate that deviates from the true 
kurtosis. 
Note also that the $L_2$ discrepancy confidence sets are 
wider than those for the other two discrepancies, which is consistent 
with
the discussion on efficiency in Section~\ref{sec::robust_sbi}.

%\VV{The large width of some confidence sets could potentially be due to sparse sampling of the parameter space.
%Eight values per parameter were coarsely sampled, totalling an already substantial grid of 4096 grid. An active learning  approach to sample the space more  efficiently would be invaluable, which we discuss in Section~\ref{sec::AL}.}

\begin{figure}[t!]
    \centering
    \subfloat{\includegraphics[width=12cm, height=3cm]{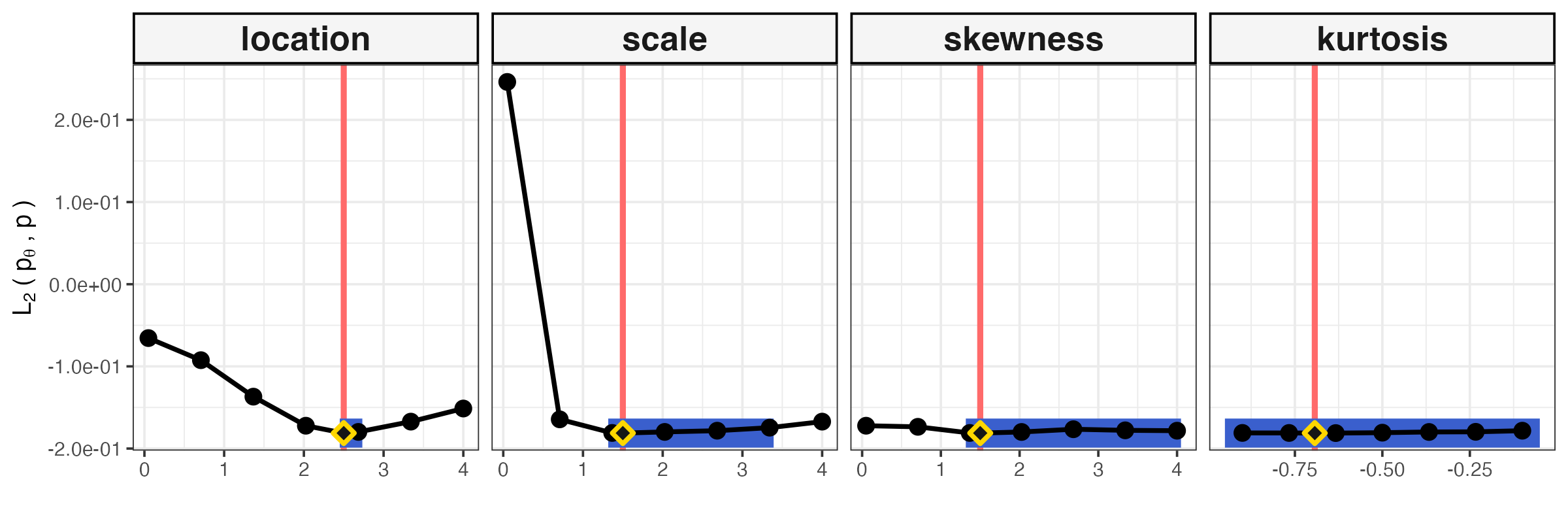}}
    \\
    \subfloat{\includegraphics[width=12cm, height=3cm]{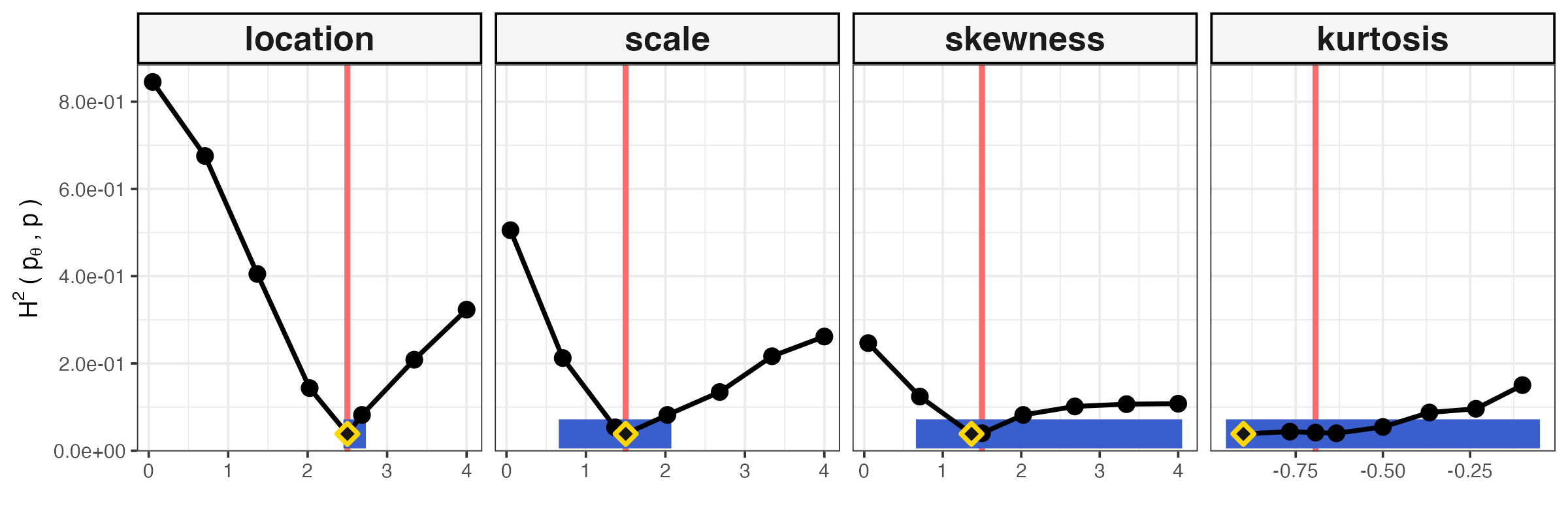}}\\
    \subfloat{\includegraphics[width=12cm, height=3cm]{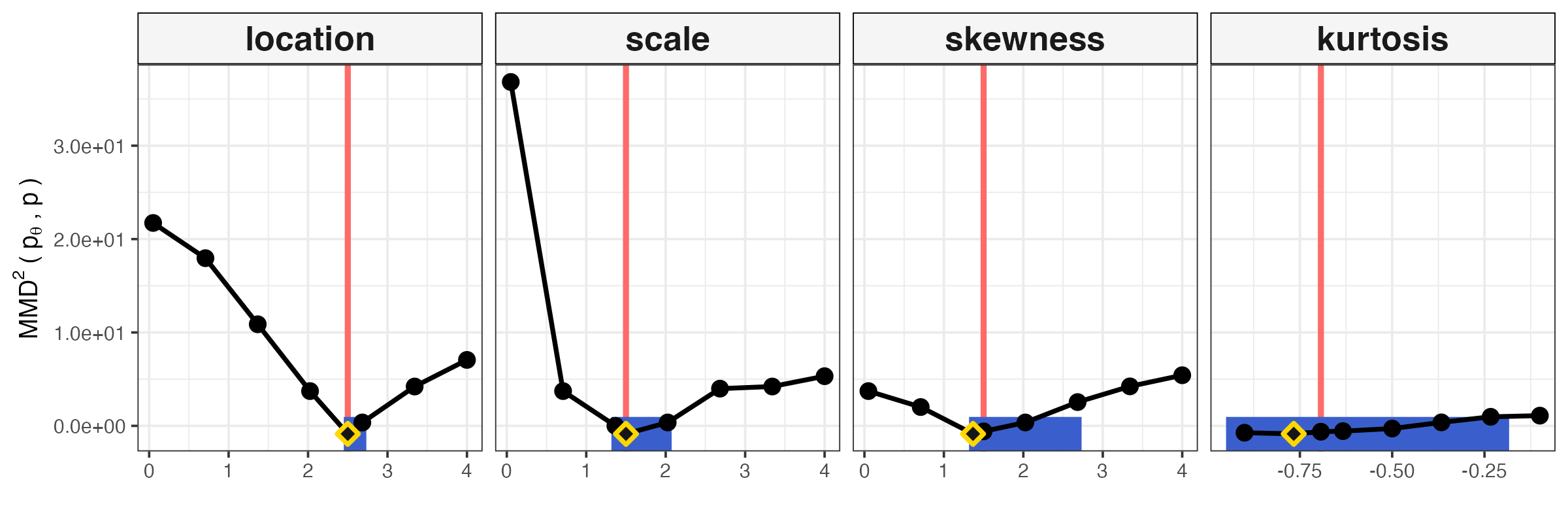}}    
    \caption{\sl{{\bf Inference for the four parameters of the g-and-k distribution} 
    with parameters $\theta=[\, l=2.5,\, s=1.5,\, g=1.5,\, k= - \log(2)\,]$ (red 
 lines).
The three discrepancies (rows) yield estimates (gold diamonds) that are close to the true values. 
Relative fit confidence sets are reported in blue.
They are wider for skewness and kurtosis, suggesting that inference for these parameters is more challenging.
    }}
    \label{fig:gk_distribution_inference}
\end{figure}

\subsection{Robust SBI for Irregular Model -- Gaussian Mixture Model} 
\label{sec::GMM}

\begin{table}[t!]
    \centering
    \begin{tabular}{
  l 
  *{8}{S[table-format=1.5]}
  S[table-format=5]
  %S[table-format=3.2]
}
\toprule
\multicolumn{1}{c}{Discrepancy} &
  \multicolumn{4}{c}{Coverage}  &
  \multicolumn{4}{c}{Length} \\
\cmidrule(lr){2-5} \cmidrule(lr){6-9}  
  & \multicolumn{1}{c}{$\mu_1$} & \multicolumn{1}{c}{$\mu_2$} & \multicolumn{1}{c}{$\sigma$} &  \multicolumn{1}{c}{$p_1$} 
  & \multicolumn{1}{c}{$\mu_1$} & \multicolumn{1}{c}{$\mu_2$} & \multicolumn{1}{c}{$\sigma$} &  \multicolumn{1}{c}{$p_1$} \\
\midrule
\midrule 
KL (likelihood)  & \multicolumn{1}{c}{1} &  \multicolumn{1}{c}{1} & 
 \multicolumn{1}{c}{1} &  \multicolumn{1}{c}{1} & 
 \multicolumn{1}{c}{0.205  \small $\pm$ .011} &
\multicolumn{1}{c}{0.138  \small $\pm$ .005} &
\multicolumn{1}{c}{0.084  \small $\pm$ .003} &
\multicolumn{1}{c}{0.058  \small $\pm$ .00} 
\vspace{.1cm} \\
Hellinger & \multicolumn{1}{c}{1}  & \multicolumn{1}{c}{1} 
& \multicolumn{1}{c}{1} 
& 
\multicolumn{1}{c}1 
&  \multicolumn{1}{c}{0.925 \small $\pm$ .072}
&\multicolumn{1}{c}{0.581 \small $\pm$ .052}  
& \multicolumn{1}{c}{0.350 \small $\pm$ .035}
& \multicolumn{1}{c}{0.209 \small $\pm$ .014} \vspace{.1cm}\\
$L_2$  & \multicolumn{1}{c}{1} 
& \multicolumn{1}{c}{1} 
& \multicolumn{1}{c}{1} 
& \multicolumn{1}{c}{1} 
& \multicolumn{1}{c}{1.255  \small $\pm$ .119} &\multicolumn{1}{c}{0.534  \small $\pm$ .047} 
& \multicolumn{1}{c}{0.412  \small $\pm$ .072}
& \multicolumn{1}{c}{0.186  \small $\pm$ .013}  \\
\bottomrule
\end{tabular}
\caption{\sl{{\bf Inference for the four parameters of the GMM model in the identifiable case.} 
%The true parameter is $\theta^*=[\mu_1=-1.5,\, \mu_2=2.5,\, \sigma=1,\, p_1=0.3]$. 
Empirical coverages with 95\% simulation bounds and average lengths (accurate up to two digits)
of 95\% confidence sets obtained via the relative fit approach (discrepancies) and inversion of likelihood ratio tests (likelihood),
based on 50 repeat simulations.}}
\label{table::results_GMM}

\bigskip
    \captionsetup[subfloat]{labelformat=empty}
    \subfloat{\includegraphics[width=3.2cm]{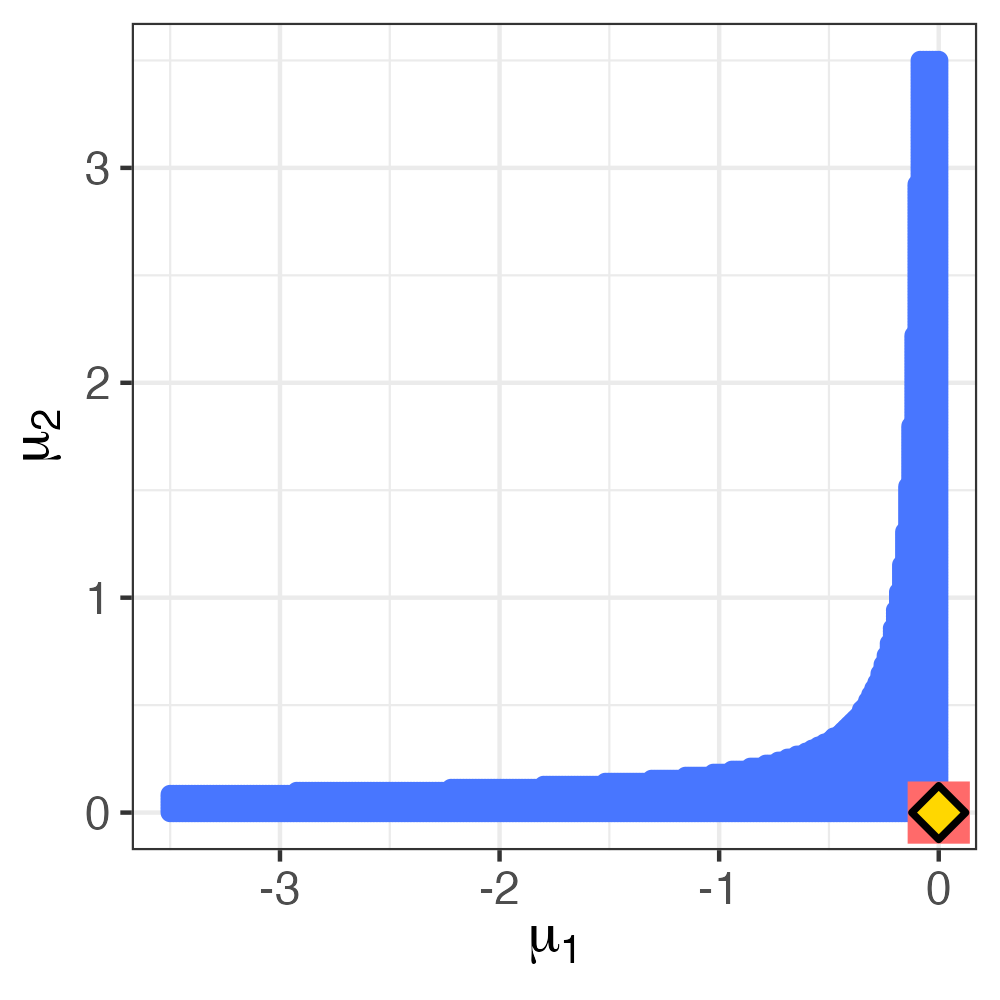}}
    \subfloat[a) {\sl $L_2$ discrepancy}]{\includegraphics[width=3.2cm]{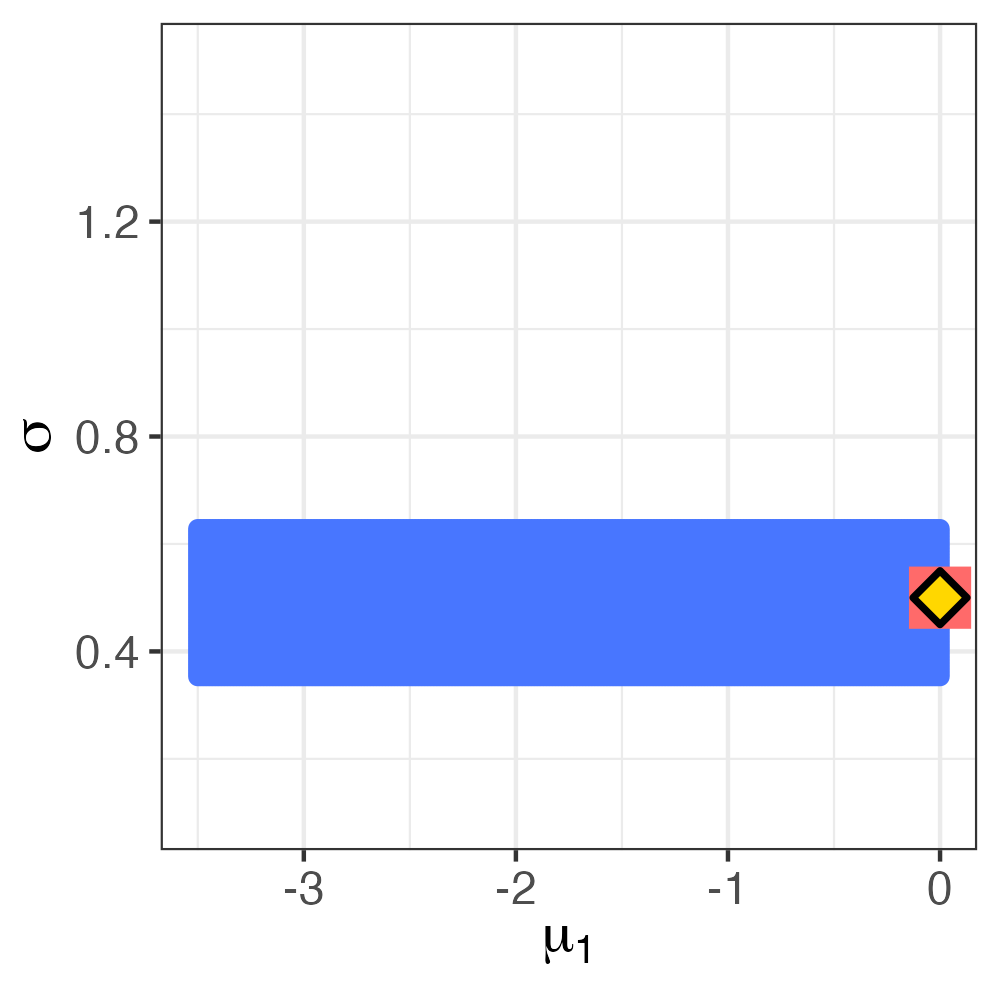}}
    \subfloat{\includegraphics[width=3.2cm]{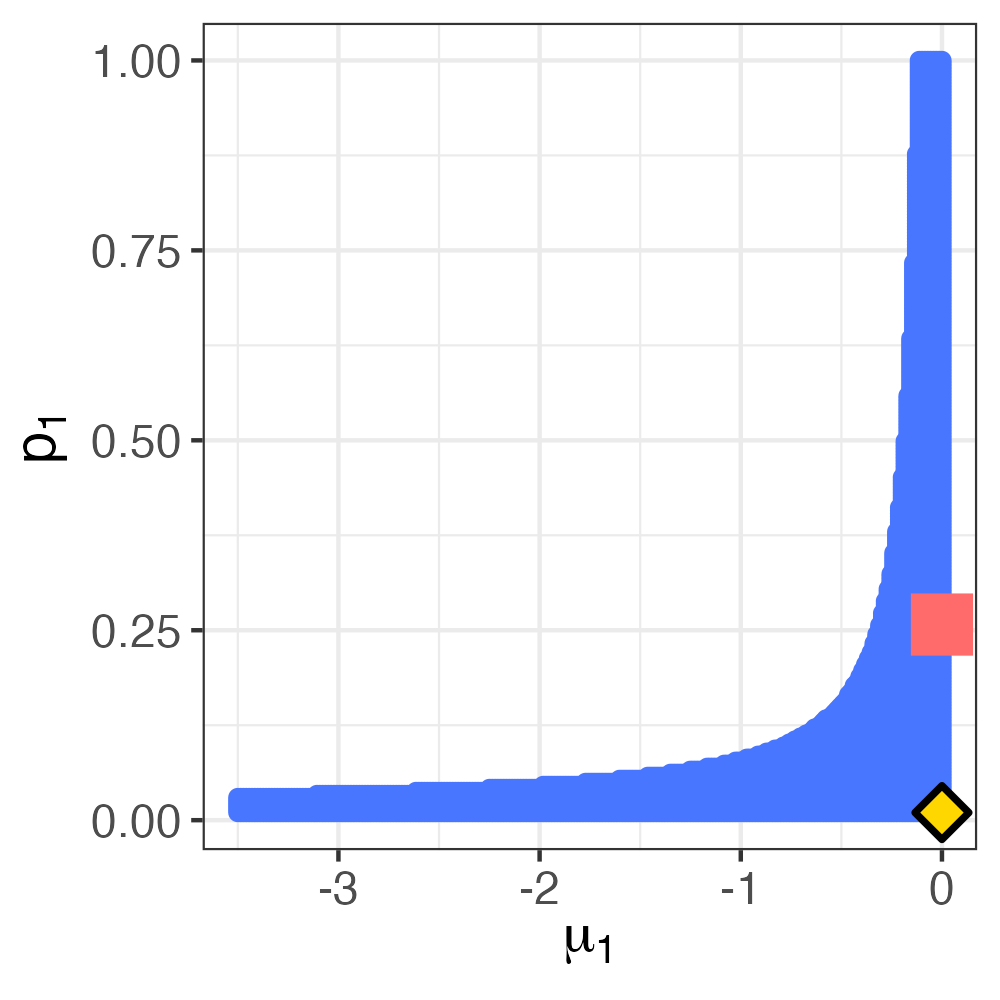}}   \\
    \subfloat{\includegraphics[width=3.2cm]{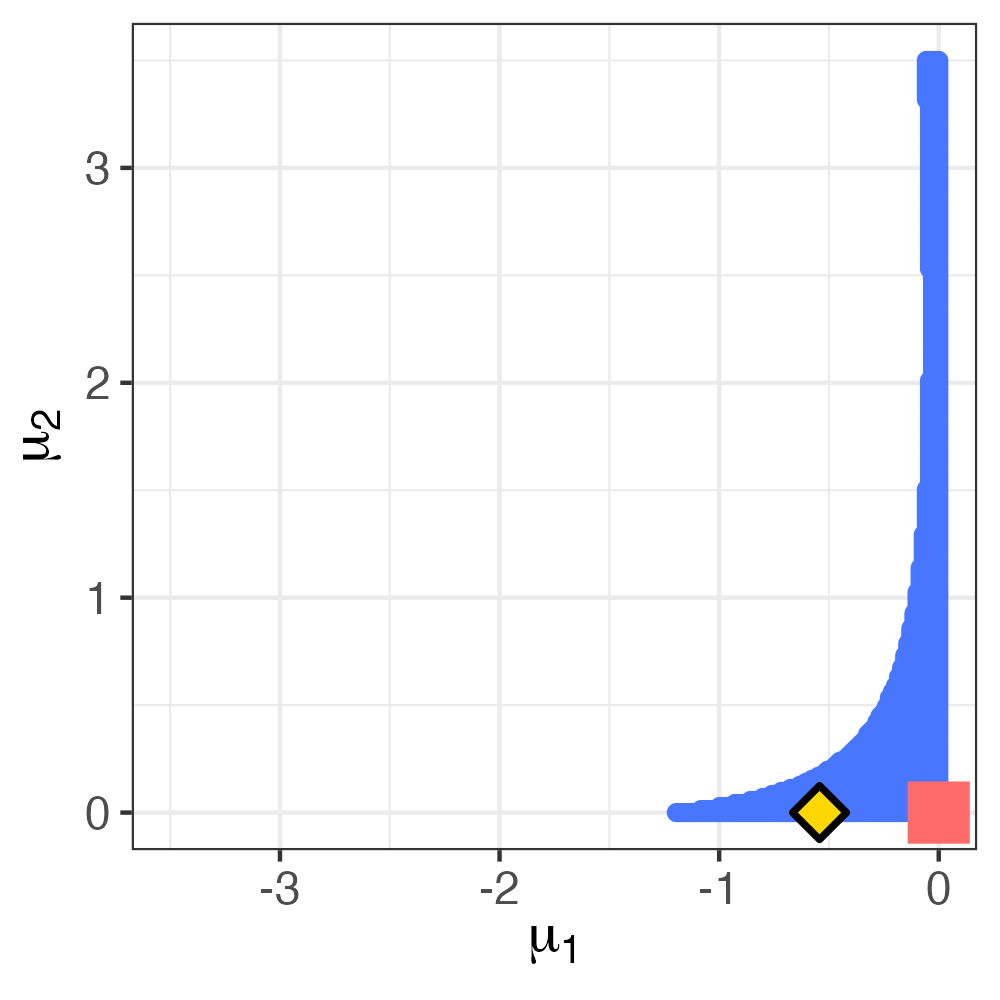}}
    \subfloat[  b) {\sl  Hellinger discrepancy}]{\includegraphics[width=3.2cm]{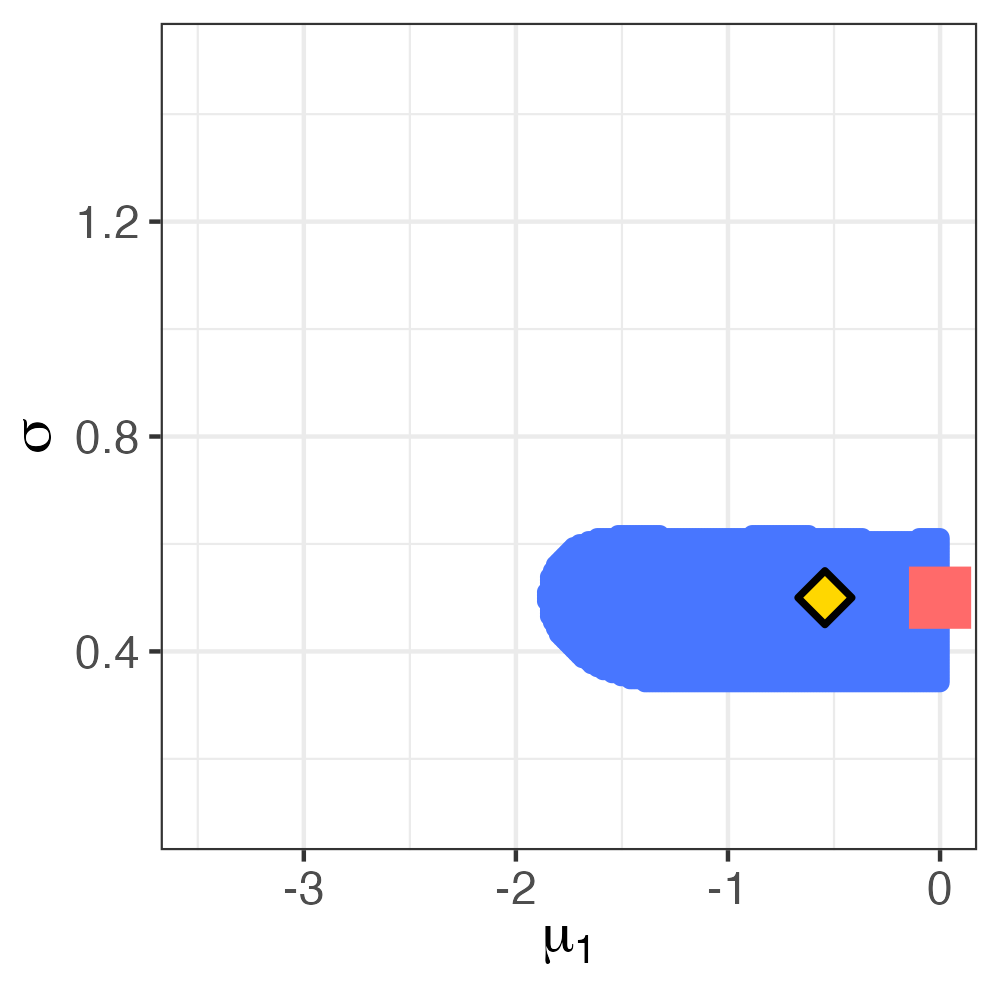}}
    \subfloat{\includegraphics[width=3.2cm]{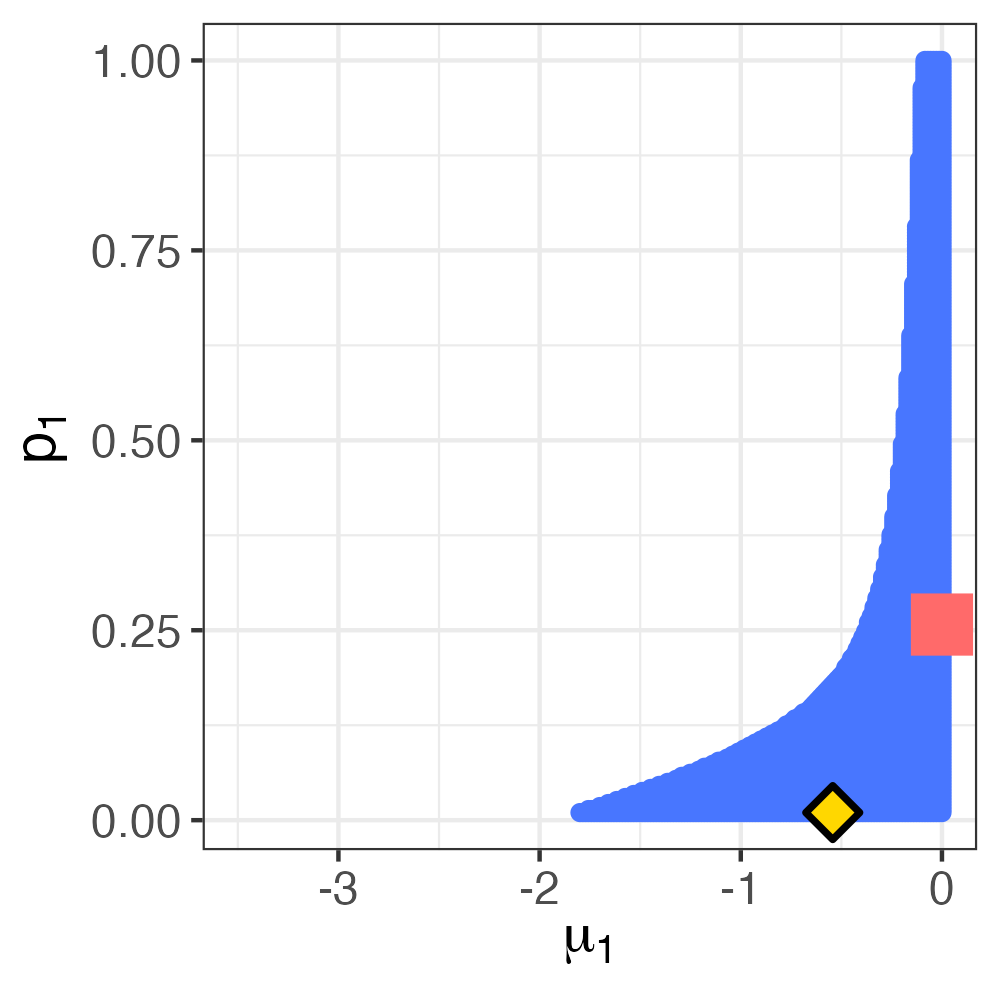}}   
    \captionof{figure}{\sl{{\bf Confidence sets for the four parameters of the GMM in the  unidentifiable case} when $\mu_1=\mu_2$, 
    for several slices of the parameter space.
    True (red) and estimated (gold) parameters, and relative fit confidence sets (blue) using the (a) efficient $L_2$ and (b) Hellinger discrepancies.}}
    \label{fig:gmm_complex_unidetifiable}

\end{table}

Assume that we have $n=2000$ observations from a Gaussian mixture models (GMM)
with density 
$p_\theta(x)=p_\cdot\phi(x, \mu_1, \sigma)+(1-p)\cdot \phi(x, \mu_2, \sigma)$, 
where $\phi(x, \mu, \sigma)$  is the Gaussian density 
with mean $\mu$ and variance $\sigma^2$. We fit a GMM model, and the target of inference 
is $\theta=(\mu_1, \mu_2, \sigma, p)$.
We use this example to illustrate the application of SBI is to obtain confidence 
sets when asymptotic results do not apply, as happens, for example,
for GMMs.

We start with 
the identifiable case when $\mu_1<\mu_2$ (and 
there is enough separation between the mixture components, for
convenience).
Table~\ref{table::results_GMM} reports the averages over 50 
repeat simulations of the coverages
and lengths of the relative fit $95\%$ confidence sets
for the efficient L2 and Hellinger discrepancies.
For comparison we also report these metrics for the asymptotic theoretical 
confidence set using inversion of likelihood ratio hypothesis tests. 
We did not report results for the MMD discrepancy as it failed to estimate parameters close to the truth and yielded wide, uninformative confidence sets.
For the MMD approach we experimented with both Gaussian and polynomial kernels.
%This is likely due to the GMM 
%having several components each with its own mean and standard deviation. 
%The MMD provides a measure of global similarity between distributions and may 
%overlook some local features of the GMM which translates into inaccurate inference 
%for this application. \V{I don't understand what this MMD part means.}
%
All the other discrepancies produced valid confidence sets. The confidence sets for
$\mu_1$ based on the 
Hellinger discrepancy are shorter than the confidence sets for $L_2$, 
which is in line with theoretical results. For the other parameters, they are overlapping.
However, they are 4 to 5 times
wider than their theoretical likelihood inversion test counterparts 
due to 
estimation errors of the density ratios and sample splitting in the relative 
fit procedure.

Next we consider the unidentifiable case, with data simulated from
$p_\theta$ with $\mu_1=\mu_2 = 0$, so that the true 
distribution is effectively Gaussian. However,
we fit a GMM, so that $\mu_1$, $\mu_2$ and $p$ are not identifiable.
Figure~\ref{fig:gmm_complex_unidetifiable} shows slices of the Hellinger 
and $L_2$ estimated confidence sets based on $n=3000$ observations,
for several subspaces of the 4-dimensional parameter space.
(An asymptotic confidence set is invalid.) 
The robust SBI 
confidence set highlights the unidentifiability.
In the upper left plot we see that if $\mu_1$ is far from 0,
then $p$ must be close to 0 and $\mu_2$ must be close to 0.
The upper middle plot shows that $\sigma$ is identified.
The upper right plot shows that when $p>0$, $\mu_1$ is constrained to be near 0
but when $p\approx 0$, $\mu_1$ is not identified.
The situation is similar for the plots in the bottom row
(based on the Hellinger discrepancy).
In this example, the $L_2$ discrepancy seems to lead 
to tighter confidence sets.
In a non-identified models, it would be
difficult to say which discrepancy should lead to
smaller confidence sets in general.

%\subsection{Finite Mixtures  of Turkey's g-\& -h distributions.} 

\subsection{SBI Goodness-of-Fit Test}
\label{sec::exgof}

We apply the Wasserstein distance and Kolmogorov-Smirnov based GoF tests to 
simulated data from 
(i) the normal distribution $\mathcal{N}(5, 1)$, 
(ii) the exponentially tilted distribution (\ref{eq::exponential_tilt_normal}) 
with $\theta=5$, $\sigma=1$ and $\alpha=(0.075, -0.0075)$
and
(iii) the t-distribution with $df=3$ and shifted to have mean 5.
The assumed model is $P_\theta = \mathcal{N}(\theta, 1)$.

\begin{figure}[t!]
    \centering
    \includegraphics[width=.6\linewidth]{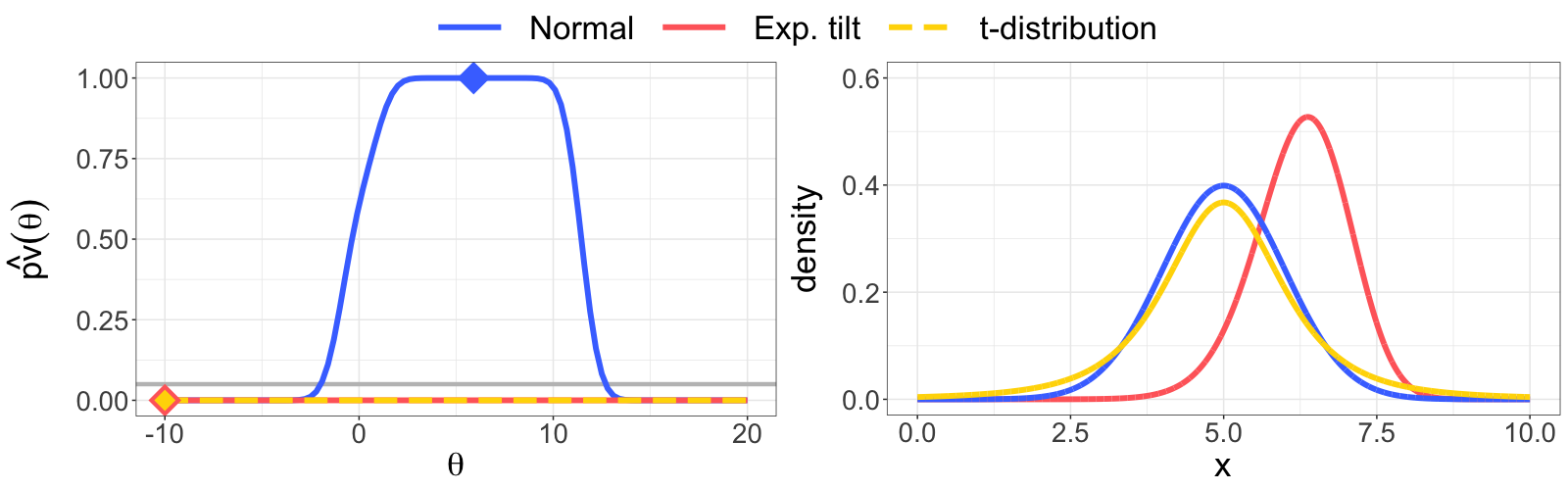}
\caption{\sl{{\bf SBI goodness of fit.} 
Left: Estimated p-value function (\ref{eq::gof.p}) for the Wasserstein distance based 
GoF test with null hypothesis $P \in \{p_\theta: \mathcal{N}(\theta, 1)\}$, 
when the true distribution $P$ is $\mathcal{N}(5, 1)$ (blue), exponentially 
tilted Gaussian in (\ref{eq::exponential_tilt_normal}) with $\alpha=(0.075, -0.0075)$ (red) and $\text{t}(\text{df}=3)$ (yellow).
The densities are shown in the right panel.
%\V{Please plot the normal distributions fitted to the respective data, not the EDFs.} 
The GoF test estimated p-values in (\ref{eq::gof.pv}) are 1, 0 and 0 (diamond-shaped points).
The test yields the correct decisions in the three cases.}}
\label{fig:gof}

\bigskip

\centering
%\sisetup{parse-numbers = false}
\begin{tabular}{
  l 
  *{8}{S[table-format=1.5]}
  S[table-format=5]
  %S[table-format=3.2]
}
\toprule
\multicolumn{1}{c}{Test statistic} &
  \multicolumn{2}{c}{Wasserstein}  &
  \multicolumn{2}{c}{Kolmogorov-Smirnov} \\  
\cmidrule(lr){2-3} \cmidrule(lr){4-5}  
  & \multicolumn{1}{c}{Average $\widehat{\text{p}}$} & \multicolumn{1}{c}{Rejection prob.} &  \multicolumn{1}{c}{Average $\widehat{\text{p}}$} & \multicolumn{1}{c}{Rejection prob.} 
  \\ 
\midrule
\midrule
$Y\sim\mathcal{N}(\mu, \sigma^2)\in \mathcal{P}_\theta$ & \multicolumn{1}{c}{0.561 \small $\pm$ .079} &  \multicolumn{1}{c}{0 \small $\pm$ 0.0} & 
 \multicolumn{1}{c}{0.546 \small $\pm$ .099} &  \multicolumn{1}{c}{0.06 \small $\pm$ .068} \vspace{.15cm}\\
$Y\sim Exp.\ Tilt \notin \mathcal{P}_\theta$ & \multicolumn{1}{c}{0 \small $\pm$ 0.0}  & \multicolumn{1}{c}{1 \small $\pm$ 0.0} 
& \multicolumn{1}{c}{0 \small $\pm$ 0.0} 
& 
\multicolumn{1}{c}{1 \small $\pm$ 0.0} \vspace{.15cm}\\ 
$Y\sim t_{3} \notin \mathcal{P}_\theta$ & \multicolumn{1}{c}{0 \small $\pm$ 0.0} 
& \multicolumn{1}{c}{1 \small $\pm$ 0.0} 
& \multicolumn{1}{c}{0.026 \small $\pm$ .013} 
& \multicolumn{1}{c}{0.84 \small $\pm$ .105}  \\
\bottomrule
\end{tabular}
\captionof{table}{
\sl{{\bf SBI goodness of fit test properties} 
over 50 repetitions at $\alpha=0.05$, 
with settings as in Figure~\ref{fig:gof}. 
Estimated average p-value and rejection rate (with 95\% confidence intervals) 
for tests based on the Wasserstein and Kolmogorov–Smirnov (KS) statistics. Under 
the null (first row), the rejection rate estimates the nominal level $\alpha$;
under the alternatives, it estimates the 
power of the test.
The Wasserstein-based test reliably detects model misspecification. 
The KS-based test has lower power, particularly under the $t$-distribution alternative. 
}}
\label{table::gof_results}

\end{figure}

For these 1D exsamples, we estimate the Wasserstein distance using~(\ref{eq::w1d_qntl}). 
We observe $Y_1, \dots, Y_n\sim P$ and estimate the quantile function  
$\widehat{F}_n^{-1}(x)$. Then for each $\theta_1, \dots, \theta_N\sim \pi$ we
\begin{enumerate}
\item simulate $\mathcal{Y}_n(\theta_j)=Y_1(\theta_j), \dots, Y_n(\theta_j)\sim P_{\theta_j}$ and $\mathcal{Y}_M^*(\theta_j)=Y_1^*(\theta_j), \dots, Y_M^*(\theta_j)\sim P_{\theta_j}$,
\item estimate the quantile functions  
$\widehat{F}^{-1}_n(\theta_j, x)$ and $\widehat{F}^{-1}_M(\theta_j, x)$ using $\mathcal{Y}_n(\theta_j)$ and $\mathcal{Y}^*_M(\theta_j)$ respectively,
\item estimate the distance $\widehat{W}(P_M^*(\theta_j), P_n(\theta_j)) = 
\left( \int_0^1 |\widehat{F}_M^{-1}(\theta_j, u) - \widehat{F}_n^{-1}(\theta_j, u)|^2 du \right)^{1/2}$
and similarly for $\widehat{W}(P_M^*(\theta_j), P_n)$.
\end{enumerate}
The test statistics $\widehat{T}_n$ and $\widehat{T}_n(\theta)$ are derived using the estimated quantities. 
The Kolmogorov-Smirnov statistic is instead computed by first estimating the empirical CDFs from the observed, $\mathcal{Y}_n$, and simulated data, $\mathcal{Y}_n(\theta_j)$, $\mathcal{Y}_M^*(\theta_j)$, then  deriving $\widehat{KS}(P_M^*(\theta_j), P_n(\theta_j))=
\max_{x} |F_M^*(\theta_j, x)-F_n(\theta_j, x)|$.
Fig.~\ref{fig:gof} (right) compares the three empirical distributions 
to the assumed normal distribution fitted to the respective data. 
The discrepancy between true and assumed models in (ii) is 
clearly visible. 
In example (iii) we chose the degrees of freedom of the true distribution for it to be distinguishable from the model but not easily.
Fig.~\ref{fig:gof} (left) shows the estimates of $p(\theta)$ in (\ref{eq::gof.p}) 
as functions of $\theta$ for the three datasets, as well as 
the p-values in~(\ref{eq::gof.pv}) for the null hypothesis that 
the data has distribution $\mathcal{N}(\theta, 1)$. 
The test leads to the correct decisions
in the three cases.
We repeated this simulation 50 times to estimate the powers of the tests: when the true distribution belongs to the model (scenario (i)), 
the GoF test never rejected the null hypothesis when using Wasserstein distance and three times for KS statistic; see Table~\ref{table::gof_results}. 
In cases (ii) and (iii), when the truth does not belong to the model, the Wasserstein-based GoF correctly rejected the null all 50 times. The KS-based test fails to reject the null 8 out of 50 times for (iii) while correctly rejecting it for all repetitions of the experiment in (ii). The choice of distribution for (iii) highlights how the Wasserstein-based test is more powerful than the KS-based test 
in this example as it better captures subtle differences in the true and model 
distributions as seen in Fig.~\ref{fig:gof} right panel.

\section{Accoutrements}
\label{sec::accoutrements}

We now present two additional results that are useful in the SBI framework beyond the model
misspecification situation: a closed form approximation to an intractable model $p_\theta$ and an 
active learning method to sample the parameter space efficiently, which
should be useful particularly in higher dimensions.

\subsection{Model Approximation via SBI} \label{sec::approx}

In cases where
$p_\theta$ is intractable,
we have used SBI to construct a confidence set for $\theta$.
But in some cases
it might be useful
to have a closed form
expression that approximates $p_\theta$.
In this section we show how SBI 
can be used to find such an expression.
This is distinct from
constructing inferences for $\theta$.

We approximate $p_\theta(x)$ with a varying coefficient model
$$
p_{\theta,f}(y) = \sum_{r=1}^k f_r(\theta) b_r(y)
$$
where $b_1,\ldots, b_k$ are given basis functions
and $f(\theta) = (f_1(\theta),\ldots, f_k(\theta))$
are smooth functions mapping $\Theta$ to $\mathbb{R}$.
We want to find $f^*$ to minimize
$$
\int (p_\theta(y) - p_{\theta,f}(y))^2 dy.
$$
The mininimizer is
$$
f^*(\theta) = B^{-1} \E_\theta[b(Y)]
$$
where $B$ is the $k\times k$ matrix with
$B_{rs} = \int b_r(y) b_s(y)$
and $b(Y) = (b_1(Y),\ldots, b_k(Y))$.
We estimate this by
$$
\hat f^*(\theta_j) = \frac{1}{m}\sum_{i=1}^m B^{-1} b(Y_i(\theta_j)) = B^{-1} \overline{b}_{\theta_j}
$$
where
$\overline{b}_{\theta_j} = m^{-1}\sum_{i=1}^m b_{\theta_j}(Y_i(\theta_j))$.
For $r=1,\ldots, k$,
we estimate $f^*_r(\theta)$ by
nonparametric regression, e.g. kernel or local polynomial regression.
If we use the former we have
$$
\hat f_r(\theta) = 
\frac{\sum_{j=1}^N K_h(\theta-\theta_j) B^{-1} \overline{b}_{\theta_j,r}}
{\sum_{j=1}^N K_h(\theta-\theta_j)}
$$
where $K_h$ is a kernel with bandwidth $h$ and
$\overline{b}_{\theta_j,r}$ is the $r$-th element of
$\overline{b}_{\theta_j}$.
We are essentially doing $N$
density estimation problems
but the $N$ densities are related to each other by the
smooth functions $f_j$.

Then we approximate $p_\theta$ by
$$
p_{\theta,\hat f}(y) = \sum_{r=1}^k \hat f_r(\theta) b_r(y).
$$
Now
$\overline{b}_{\theta_j} - b_{\theta_j} = O_P(m^{-1/2})$ and so,
if the regression estimator has accuracy
$O_P(n^{-\gamma/(2\gamma+d)})$ then
$$
\int ( p_{\theta,f^*}(y) - p_{\theta,\hat f}(y))^2 dy = O_P(1/m) + O_P(n^{-2\gamma/(2\gamma+d)}).
$$

Rather than estimating $f(\theta)$ at each $\theta_j$ and then appying smoothing,
we can instead use smoothing to estimate $\E_\theta[b(Y)]$.
This is useful if $m$ is taken to be small
for computational expediency.
In fact, we can even take $m=1$.
In that case, we let
$$
\hat{b}_\theta = 
\frac{\sum_j K_h(\theta-\theta_j)\overline{b}_{\theta_j}}
{\sum_j K_h(\theta-\theta_j)}
$$
and then we set
$$
\hat f(\theta) = B^{-1}\hat {b}_\theta.
$$

Now we consider how one might choose the number of basis functions $k$.
We fix an upper bound $K$ and choose
$1\leq k \leq K$.
For a fixed $\theta$,
one approach is to minimize an estimate of the $L_2$ error
$$
\int (p_{\theta,\hat f,k}(y) - p_{\theta}(y))^2 dy
$$
where we now include the subscript $k$.
This is equivalent to minimizing
$$
L(\theta,k)=\int p_{\theta,\hat f,k}^2(y)dy - 2\int p_{\theta,\hat f,k}(y) p_\theta(y) dy
$$
which we can estimate by
$$
\hat L(\theta,k) = 
\int p_{\theta,\hat f,k}^2(y)dy - \frac{2}{m}\sum_i  \hat p_{\theta,\hat f,k}( Y_i(\theta)).
$$
However, the result would be a different $k$ for each $\theta$.
Instead, we minimize
\begin{equation}
    \hat L(k) = \frac{1}{N} \sum_j \hat L(\theta_j,k).\label{eq::loss_l2_approx}
\end{equation}
(An alternative is to maximize the maximum over $\theta_j$.)

As proof of concept, we show that a Beta$(\alpha,\beta)$ distribution can be well approximated
by an exponentially
tilted uniform distribution.
The unknown parameter is \( \theta = \alpha \) and,
for simplicity,
we fix \( \beta = 1.5\, \alpha \) to make this a 1-dimensional problem. 
Figure~\ref{fig:approx_model} shows the true density and approximations that use
$k=4$ and $k=8$ basis functions $b_r$, 
for $\theta=1,3$ and $5$.  
The loss (\ref{eq::loss_l2_approx}) was minimized at $k=8$, and the corresponding estimates are
close to the true densities for all $\theta$.

\begin{figure}[t!]
    \centering
    %\subfloat[]{
    \includegraphics[width=.7\linewidth%, %height=.30\linewidth
    ]{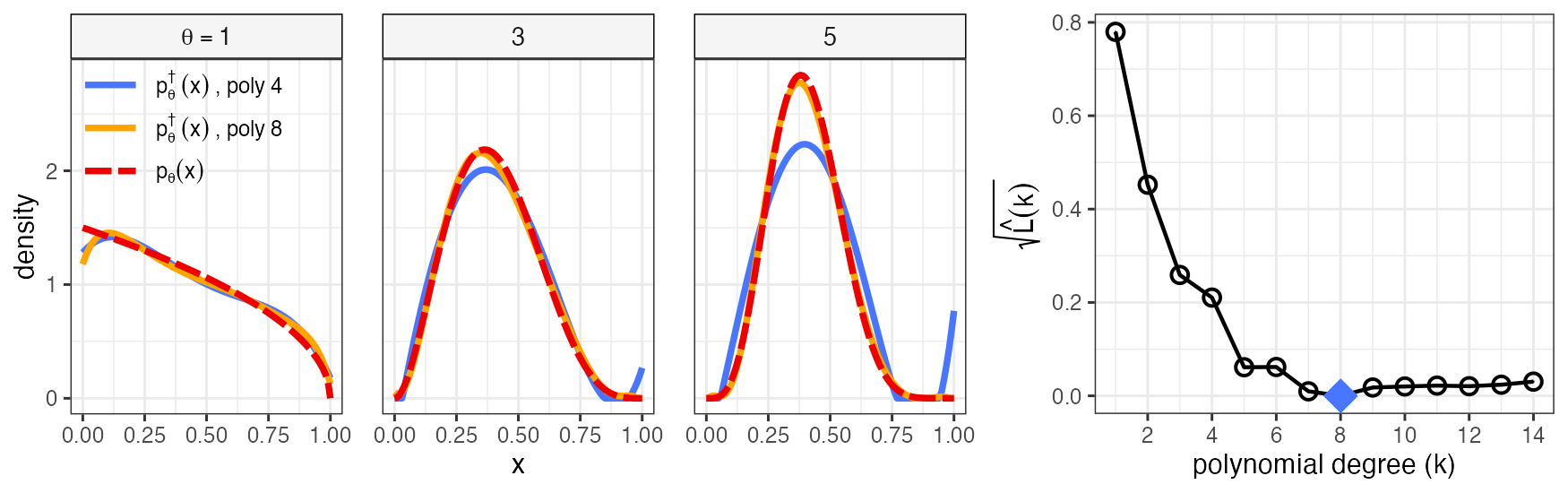}
    %}\\
    %\subfloat[]{\includegraphics[width=.65\linewidth
    %]{figures/exponential_tilt/model_expansion/fig_estimated_beta_modelapprox_L2.png}}
    \caption{\sl{{\bf Approximation of $p_\theta=\text{Beta}(\theta, 1.5 \cdot \theta)$},
    $\theta=1,3, 5$ (red curves).
    The blue (orange) approximations use polynomial functions with 
    $k=4$ ($k=8$) basis functions. 
    The loss (\ref{eq::loss_l2_approx}) is minimized at $k=8$.
  %\lortomas{instead of chi squared test need to use the l2 distance for choosing power. Implement it for figure 7 b and update 7a with l2.} 
    }
    }
    \label{fig:approx_model}
\end{figure}

\subsection{Active Learning to Explore the Parameter Space Efficiently} \label{sec::AL}

The first step of SBI is to generate $N$ independent values of $\theta$ from $\pi$;
see Section~\ref{sec::SBI}.
Here we consider choosing $\theta$ values sequentially to explore the parameter space more efficiently.

Suppose we have already drawn
$\theta_1,\ldots, \theta_j$
and let $\hat{\mathrm{pv}}(\theta)$
be the current estimate of the p-value function~(\ref{eq::B}).
Let $ C=\{\theta:\ {\text{pv}}(\theta)\geq \alpha\}$
and 
$\hat C=\{\theta:\ \hat {\text{pv}} (\theta)\geq \alpha\}$
denote the confidence set and its current estimate.
We aim
to minimize an estimate of the error
between
$\hat C$ and $C$:
$$
R(\hat C,C)=
\int P (\hat B(\theta) \neq B(\theta)) d\theta, 
$$
where 
$\hat B(\theta) = \mathbb{I} (\hat{\text{pv}}(\theta)\geq \alpha)$ and
$B(\theta) = \mathbb{I} ({\text{pv}}(\theta)\geq \alpha)$.
Therefore, to reduce $R(\hat C,C)$, we choose the next value $\theta_{j+1}$ 
for which
$P (\hat B(\theta) \neq B(\theta))$ is large.

If
\begin{equation}\label{eq::kclt}
\frac{\hat{\text{pv}}(\theta)  - \text{pv}(\theta)}{s(\theta)}\rightsquigarrow N(0,1)
\end{equation}
then
$P(\hat B(\theta) \neq B(\theta)) \to
\Phi \left(- {|\alpha - \text{pv}(\theta)| \over s(\theta)}\right)$.
To see this,
suppose that $\text{pv}(\theta)>\alpha$.
Then
$$
P(\hat B(\theta) \neq B(\theta)) =
P \left(\hat{\text{pv}}(\theta) < \alpha \right)=
P( \left(\hat{\text{pv}}(\theta)-\text{pv}(\theta))/s(\theta) < 
(\alpha - \text{pv}(\theta))/s(\theta) \right) \to
\Phi(- |\alpha - \text{pv}(\theta)|/s(\theta)).
$$
Similarly for $\text{pv}(\theta)<\alpha$. 
Condition (\ref{eq::kclt}) holds 
if, for example, $\hat{\text{pv}}$ is the kernel estimator
with appropriate bandwidth.

We estimate the bound by
$e(\theta) = \Phi(- |\alpha - \hat{\text{pv}}(\theta)|/\hat s(\theta))$. 
To reduce 
$R(\hat C,C)$,
we choose $\theta_{j+1}$ to maximize $e(\theta)$
or we sample it from a density  
that puts high probability on $\theta$'s where
$e(\theta)$ is large.
Note that $e(\theta)$ is large when
$\hat{\text{pv}}(\theta)$ is close to $\alpha$ (we are close to the boundary of the confidence set)
or when $s(\theta)$ is large (the $p$-value is poorly estimated).

\RestyleAlgo{ruled}
\begin{algorithm}[t!]
\caption{Active learning for confidence set estimation 
\vspace{0.1cm}}\label{alg::ALnpm}
\SetAlgoLined\vspace{0.1cm}
\SetKwInOut{Input}{Input}
\SetKwInOut{Output}{Output}
\Input{\begin{itemize}
    \item[] Observed data  ${\cal Y}=Y_1, \dots, Y_n\sim p$;\vspace{-.075in}
    \item[] Initial parameter set and simulated data set $S_\theta = S_Y = \{\varnothing\}$;\vspace{-.075in} 
    \item[] Number of total simulated parameters $N$;\vspace{-.075in}
    \item[] Number of active learning steps $\eta$;\vspace{-.075in}\item[] Initial parameter values $\{\theta_1^{(0)},\ldots,\theta_{N/\eta}^{(0)}\}$, $\theta_j^{(0)} \sim \pi$; \vspace{-.05in}
\end{itemize}
}
\Output{$1-\alpha$ confidence set $\hat{C}=\{\theta: \hat{
\text{pv}}(\theta) \geq \alpha \}$}\vspace{.025in}
\For{$i = 1,\dots,\eta$}{
Augment the set of parameter values: $S_\theta \leftarrow S_\theta \, \bigcup \ \left\{\theta_{1}^{(i-1)}, \dots, \theta_{N/\eta}^{(i-1)} \right\}$ \\
Augment the set of simulated data: $S_Y \leftarrow S_Y\, \bigcup\ \left\{\mathcal{Y}(\theta_1^{(i-1)}), \ldots, \mathcal{Y}(\theta_{N/\eta}^{(i-1)}) \right \}$, where $\mathcal{Y}(\theta_j^{(i-1)})= \left( Y_1(\theta_j^{(i-1)}),\dots, Y_n(\theta_j^{(i-1)}) \right)$ and $\ Y_j(\theta_j^{(i-1)}) \sim p_{\theta_j^{(i-1)}}$  \\
Build the dataset for SBI using $(S_\theta,S_Y)$ and estimate the likelihood at all $\theta \in S_\theta$ (see Section~\ref{sec::SBI})   \\
Compute the indicators $B(\theta_j) = \mathbb{I} \{\ell(\mathcal{Y}(\theta_j), \theta_j)\leq \ell(\mathcal{Y}, \theta_j)\}$,
$\theta_j \in S_\theta$\\
%
%
%- compute the indicators $B(\theta_j) = \mathbb{I} \{\ell(\mathcal{Y}(\theta_j), \theta_j)\leq \ell(\mathcal{Y}, \theta_j)\}$, $\theta_j \in S_\theta$\\
Estimate $\hat{\text{pv}}(\theta)$ and $\hat{s}_{\text{pv}}(\theta)$ via kernel regression of $B(\theta_j)$ on $\theta_j$, $\theta_j \in S_\theta$\\
%sample $\theta_{1}^{(i)}, \dots, \theta_{N}^{(i)}\sim f_{\theta}$ where for example $f_{\theta_j} =\dfrac{e(\theta_j)}{\sum_k e(\theta_k)}$ and $e(\theta) = \Phi(- |\alpha - \hat{\text{pv}}(\theta)|/\hat s(\theta))$\\
Sample $\theta_{1}^{(i)}, \dots, \theta_{N/\eta}^{(i)}\sim f_{\theta}$, where $\ f_{\theta} \propto e(\theta)$ and $e(\theta) = \Phi(- |\alpha - \hat{\text{pv}}(\theta)|/\hat s(\theta))$ 
\\
%- update $S_\theta \leftarrow S_\theta \, \bigcup \ \{\theta_{1}^{(i)}, \dots, \theta_{N}^{(i)} \}$ \\
}
{\bf return} \quad $\widehat{C}=\{\theta:\ \widehat{\text{pv}}(\theta)>\alpha\}$.\vspace{.1cm}
\end{algorithm}

To illustrate the method, we estimated the mean vector of a 
bivariate normal distribution 
$\mathcal{N}(\mu, \Sigma)$ with mean $\theta = (\mu_1, \mu_2)$ 
and known covariance $\Sigma = \sigma^2 I$ with $\sigma = \sqrt{2}$.
The true target parameter is $\theta^* = (1, 2)$. 
We used a small SBI simulation
to emulate a high dimensional situation, when active learning is most valuable.
We started with $N_0=100$ parameters equally spaced over 
a grid in the parameter space $\Theta=[-5,5]\times[-5,5]$. 
Because $N_0$ is small, we used a guided
approach for the first five iterations: we uniformly 
sampled $50$ points from the level set corresponding to the 
$1-\alpha$ quantile of a chi-squared distribution estimated 
by a quadratic approximation of the SBI likelihood function.
We then applied the proposed active learning procedure, 
as summarized in Algorithm~\ref{alg::ALnpm}, 
sampling 25 additional $\theta$ values
at each iteration.
The likelihood was estimated anew at each iteration,
per~(\ref{eq::likelihood_trick})
using deep learning to solve the classification problem 
(see Appendix~\ref{app::al_sbi_nn} for details).
For comparison, we
also estimated the likelihood by 
SBI over a regular grid of parameter values, 
matching the sample size of the active learning approach.

Figure~\ref{fig:AL_results} 
shows the true and estimated confidence sets for iterations 
3, 7, 11, 15 and 19, with a focus on the high-likelihood region of the parameter space.
The usual SBI confidence sets are variable and do not 
improve markedly as 
the simulation size increases. Their AL counterparts are more stable and
improve steadily.
The $\theta$ values sampled at each step of the 
AL procedure (grey ‘+’) tend to concentrate in areas where the 
confidence set can be improved in subsequent iterations.

A more formal comparison of the two approaches can be based on the excess risk defined in \cite{willett_minimax_2007} Eq.~(6), without normalization:
\begin{equation}
\mathcal{R}(C) - \mathcal{R}(C^*) = \int_{\Delta(C^*, C)} |\alpha - pv(\theta)| \, d\mathbb{P}_\theta,\label{eq::excess_risk}
\end{equation}
where $\Delta(C^*, C)$ is the symmetric difference of the true and estimated confidence sets $C$ and $C^*$,
\begin{equation}
\Delta(C^*, C) = \left\{\theta \in \Theta : (\theta \in C \setminus C^*) \cup (\theta \in C^* \setminus C)\right\}.\nonumber
\end{equation}
%The sequential approach shows significantly higher excess risk in early ($i \leq 4$) and later ($i > 7$) iterations, with significantly less computational time. 
%Its lower computational cost stems from the absence of bandwidth selection via cross-validation, which is necessary for the kernel-based p-value estimation.
%\textcolor{red}{Lorenzo: did you use a quadratic approx? Add explanation}
%Both methods eventually yield CS close to the ground truth. As a final note, the slower convergence and higher initial stability of the guided approach suggest a potential hybrid strategy: use the guided method in early stages, then switch to the fully nonparametric AL approach as more data accumulates.

A possible alternative active learning approach from \cite{zhao_sequential_2012}
relies on nonparametric quantile regression, but it is more delicate to 
implement because several parameters have to be chosen. Details 
and example are provided in
Appendix~\ref{app::AL}.

\begin{figure}[t!]
    \centering
    %\subfloat[\sl Updated CS over AL iterations (1-10).]{\includegraphics[width=15cm]{figures/active_learning/successive_cs_comparison_methods.png}}\\
    %\subfloat[\sl Comparison of risk (not normalized) and time for the two AL approaches.]{\includegraphics[width=8.5cm]{figures/active_learning/risk_not_normalizedvstime_AL.png}}
    \includegraphics[width=.95\linewidth]{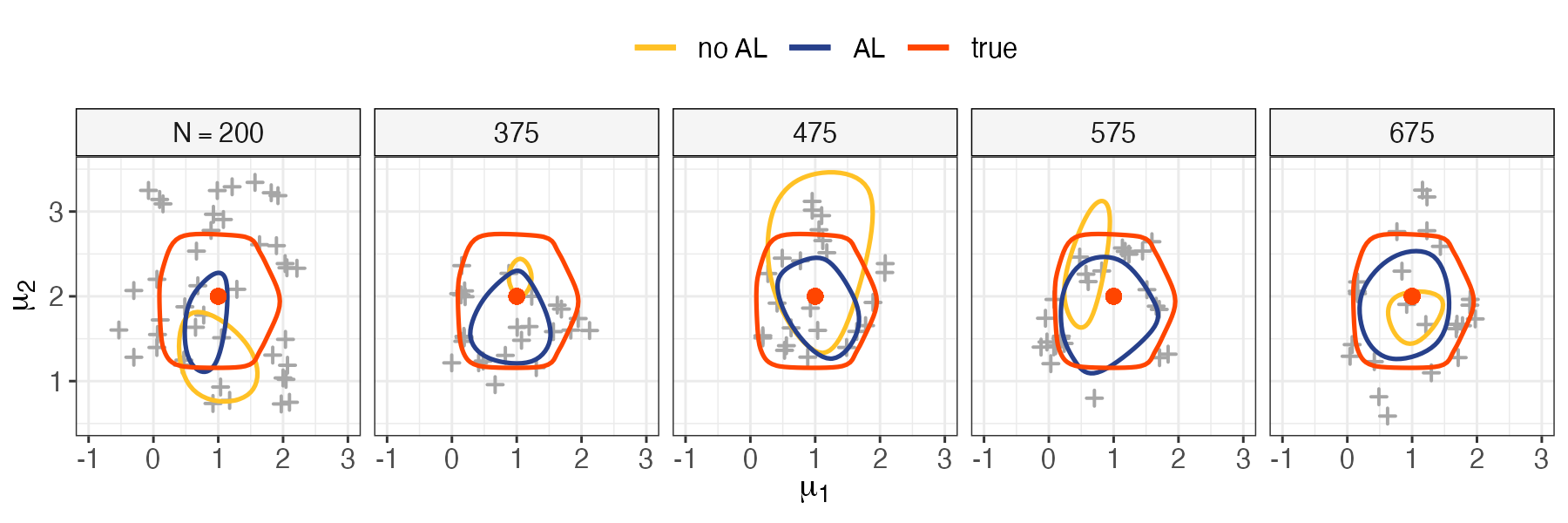}
    \caption{\sl{{\bf Active Learning Confidence Sets.} 
    Inference is for the mean $\theta=(\mu_1,\mu_2)$ of a bivariate normal distribution 
    $N(\theta, \Sigma)$, with true value $\theta^*= (1,2)$ (red dot). 
    The red perimeter is the true confidence set.
    The yellow perimeters are confidence sets obtained by estimating the likelihood by SBI on a 
    regular grid of $N$ points (specified at the top).
    The navy perimeters are the sequential active learning SBI confidence sets (Algorithm~\ref{alg::ALnpm}) 
    using the same $N$. The grey + are the new points generated by the active learning
    procedure. 
  %\lortomas{instead of iteration include the number of points, only plot 3-7-11-15-19}
    }} 
    \label{fig:AL_results}
\end{figure}

\section{Conclusion}

We presented an
approach to robust simulation-based inference (SBI) approach.
We proposed discrepancy-based estimators and discussed the theoretical guarantees for the validity of the proposed confidence sets, which are based on a relative-fit approach. This method relaxes many of the assumptions required for confidence sets.

We demonstrated the 
validity of our inference approaches across a range of applications, 
from simple Gaussian location and scale inference to more complex 
settings with unknown model densities (e.g., the g-and-k distribution), 
and an example  where regularity conditions fail
(e.g., Gaussian mixture models with same components). 
We derived empirical 
coverage of the proposed confidence sets, which achieved or exceeded the nominal coverage under both correctly specified and misspecified models. This showcases 
the robustness of our discrepancy-based inference method. 
We also proposed an approach to expand the model via exponential tilt to address model misspecification, demonstrating its validity as well with several examples.

We conclude by discussing considerations for future research. 
As detailed in Section~\ref{sec::density_ratio_estim}, our approach relies on a kernel method for density ratio estimation, 
which requires selecting a reference distribution $g$. 
There is no general rule for this choice, except selecting a distribution with larger variance to avoid exploding ratios. This choice however needs 
to be balanced against computational stability and accuracy.
In Section~\ref{sec::robust_sbi}, saw
how the reference distribution $g$ affects the asymptotic variance of the estimator, 
raising an interesting research question about the 
dependence of the variance of estimated density ratios on $g$.

As the dimensionality of the parameter space increases, 
computational costs rise exponentially.
\cite{walchessen_neural_2024} proposed
a method based on
Latin Hypercube Sampling (LHS)\footnote{\cite{carnell_lhs_2023}, https://cran.r-project.org/web/packages/lhs/lhs.pdf} which guarantees uniform coverage of the parameter space.
However, LHS does not mitigate increasing complexity with the dimensionality of the parameter space. 
We have proposed an active learning approach
but much more research on this topic is needed.

\bibliography{references}

\appendix

\section{Algorithms}
\label{appendix::algo} 

\subsection{Simulated dataset for classification-based likelihood estimation}
\label{alg:data_SBI}

\begin{algorithm}[h!]
\SetAlgoLined
\SetKwInOut{Input}{Input}
\SetKwInOut{Output}{Output}
\Input{
number $N$ of $\theta$ values\\
\ sampling distribution $\pi$ over the parameter space $\Theta$
 }
\Output{
dataset for SBI $(Z_j,{\cal Y}_j, \theta_j)$, $1\leq j \leq 2N$}
sample $\theta_1, \dots, \theta_N \sim \pi$\\
generate a permutation of the index set $s=[I_1, \dots, I_N]$  \\
\For{$j=1,\ldots, N$}{
  draw $\mathcal{Y}_j = \mathcal{Y}(\theta_j)\sim p_{\theta_j}$ and $Z_j=1$ \\
  set $\mathcal{Y}_{N+j}=\mathcal{Y}_{j}$, $Z_{N+j}=0$ and $\theta_{N+j}=\theta_{s_j}$
}
\end{algorithm}

\subsection{Newton-Raphson for Algorithm~\ref{alg:tilt} } \label{app::NR}

We start at $\hat \beta^{(0)} = (0,\ldots, 0)^T$, where the tilted model coincides with $p_\theta$,
and then iterate
$$
\hat \beta^{(k+1)}(\theta) \leftarrow \hat \beta^{(k)}(\theta) - V(\theta)^{-1} S(\theta)
$$ 
where $S(\theta)=\sum_i b(Y_i) -n\left(\dfrac{ \sum_i b(Y_i(\theta)) e^{\beta^T b(Y_i(\theta))}}{ \sum_i e^{\beta^T b(Y_i(\theta))}}\right)$ and $$V(\theta)= n  \left(\dfrac{\left( \sum_i b(Y_i(\theta)) e^{\beta^T b(Y_i(\theta))}\right)\left( \sum_i b(Y_i(\theta)) e^{\beta^T b(Y_i(\theta))}\right)^\top}{\left( \sum_i e^{\beta^T b(Y_i(\theta))}\right)^2}-\dfrac{\sum_i b(Y_i(\theta)) b(Y_i(\theta))^\top e^{\beta^T b(Y_i(\theta))}}{\sum_i e^{\beta^T b(Y_i(\theta))}}\right).$$

\section{Proofs}
\label{appendix::proofs} 

{\bf Proof of Lemma \ref{lemma::hellinger}.} 
We first derive the one-step estimator.
Let the influence function for $r$ be
\begin{align*}
    \phi_r(x) &= \frac{\partial}{\partial \epsilon}
\psi(\theta, P + \epsilon (\delta_x - P))\Biggm|_{\epsilon=0} \\
&= 
\frac{\partial}{\partial \epsilon}
\frac{1}{2}\int g\cdot \sqrt{\frac{p+\epsilon(\delta_x - p)}{g}\cdot s_\theta} 
\Biggm|_{\epsilon=0} \\
&= 
\frac{1}{2}
\int \sqrt{g\cdot\frac{s_\theta}{p}} (\delta_x - p)\\
&=
\frac{1}{2} \sqrt{\frac{s_\theta}{r}} - \frac{\psi(\theta)}{2}.
\end{align*}
Similarly by perturbating $P_\theta$ we get the influence function
\begin{align*}
    \phi_{s_\theta}(x) &= \frac{\partial}{\partial \epsilon}
\psi(\theta, P_\theta + \epsilon (\delta_x - P_\theta))\Biggm|_{\epsilon=0} \\
&= 
\frac{\partial}{\partial \epsilon}
\frac{1}{2}\int g \cdot \sqrt{\frac{p_\theta +\epsilon(\delta_x - p_\theta) }{g}\cdot r} 
\Biggm|_{\epsilon=0} \\
&= 
\frac{1}{2}
\int \sqrt{g\cdot\frac{r}{p_\theta}} (\delta_x - p_\theta)\\
&=
\frac{1}{2} \sqrt{\frac{r(x)}{s_\theta(x)}} - \frac{\psi(\theta)}{2}.
\end{align*}
Now let $Y_i\sim p$ and $Y_i^*\sim p_\theta$. The above lead to the one-step estimator
\begin{align*}
\widehat\psi(p_\theta, p) &= \int g \cdot \sqrt{\widehat{r}\widehat{s}_\theta} + \dfrac{1}{n}\sum_i \widehat\phi_r (Y_i) + \dfrac{1}{m}\sum_i \widehat\phi_{s_\theta} (Y_i^*) = \dfrac{1}{2n}\sum_i   \sqrt{\frac{\widehat{s}_\theta(Y_i)}{\widehat{r}(Y_i)}} +
\frac{1}{2m}\sum_i \sqrt{\frac{\widehat{r}(Y^*_i)}{\widehat{s}_\theta(Y^*_i)}} .
\end{align*}
From a standard von Mises expansion we get 
\begin{align*}
   \sqrt{n}(\widehat\psi -\psi)&=\sqrt{n}\Bigl(\int \phi_r(x,\hat{P}) \mathrm{d} P +
\int \phi_{s_\theta}(x,\hat{P}_\theta) \mathrm{d} P_\theta + R_n \Bigr)\\
&\approx \dfrac{1}{n}\sum_i \widehat\phi_r (Y_i) + \dfrac{1}{m}\sum_i \widehat\phi_{s_\theta} (Y_i^*) + \mathbb{G}_n(\widehat\phi_r (Y_i)-\phi_r (Y_i))+ \mathbb{G}_n(\widehat\phi_{s_\theta} (Y_i^*)-\phi_{s_\theta} (Y_i^*) + R_n)
\end{align*}
where $\mathbb{G}_nf=\sqrt{n}(\frac{1}{n}\sum_i f(X_i)-\mathbb{E}_X[f])$. Under proper assumptions (see Section 19 in \cite{vaart_asymptotic_1998}), the empirical processes
$$\frac{1}{\sqrt{n}}\sum \widehat\phi_{s_\theta} (Y_i^*)\rightsquigarrow\mathcal{N}(0, \mathbb{E}[\phi_{s_\theta}^2]),\quad \frac{1}{\sqrt{n}}\sum \widehat\phi_{r} (Y_i)\rightsquigarrow\mathcal{N}(0, \mathbb{E}[\phi_{r}^2])$$ 
with the remainder term  $R_n=o_P(n^{-1/2})$ and $\mathbb{G}_n(\widehat\phi_{r} (Y_i)-\phi_{r} (Y_i))=o_P(1)$ and $\mathbb{G}_n(\widehat\phi_{s_\theta} (Y_i^*)-\phi_{s_\theta} (Y_i^*))=o_P(1)$. 

The variance of the estimator is $\sigma^2=\mathbb{E}_{PP_\theta}[\phi_{r}^2+\phi_{s_\theta}^2]$, where 
\begin{align*}
    \mathbb{E}[\varphi^2] &= \mathbb{E}_P\left[\Biggl(\dfrac{1}{2}\sqrt{\dfrac{s_\theta}{r}}(Y)-\dfrac{1}{2}\psi\Biggr)^2\right] + \mathbb{E}_{P_\theta}\left[\Biggl(\dfrac{1}{2}\sqrt{\dfrac{r}{s_\theta}}(Y^*)-\dfrac{1}{2}\psi\Biggr)^2\right] \\
    &= \dfrac{1}{4}\int \dfrac{s_\theta}{r} p 
    + \dfrac{1}{4}\psi^2 
    - \dfrac{\psi}{2}\int \sqrt{\dfrac{s_\theta}{r}}p
    + \dfrac{1}{4} \int \dfrac{r}{s_\theta}p_\theta
    + \dfrac{1}{4}\psi^2
    - \dfrac{\psi}{2}\int \sqrt{\dfrac{r}{s_\theta}}p_\theta\\
    &= \dfrac{1}{4}\int p_\theta
    + \dfrac{1}{4} \int p
    + \dfrac{1}{2}\psi^2 
    - \psi \int \sqrt{p p_\theta}
    \\
    &=\dfrac{1-\psi^2}{2} 
\end{align*}

\begin{comment}
    Lemma~\ref{lemma::hellinger2} below provides the density analogue of Lemma~\ref{lemma::hellinger}.

\begin{lemma}\label{lemma::hellinger2}
Let
$Y_1,\ldots, Y_n \sim p$
and
$Y_1(\theta),\ldots, Y_m(\theta) \sim p_\theta$.
Let $\hat p$ and $\hat p_\theta$ be estimates of $p$ and $p_\theta$.
The one-step estimator is 
\begin{align}
\hat\psi(p_{\theta}, p) &=
\frac{1}{2n}\sum_i \sqrt{\frac{\hat p_\theta(Y_i)}{\hat p(Y_i)}} +
\frac{1}{2m}\sum_i \sqrt{\frac{\hat p(Y_i(\theta))}{\hat p_\theta(Y_i(\theta))}}.
\end{align}
Suppose that $p\in \text{Holder}(\beta_1)$ and 
$p_\theta\in \text{Holder}(\beta_2)$ for each $\theta$,
where $\beta_1,\beta_2 > d/2$ 
and $m,k\geq n$.
Assume that
$||\hat p - p|| = o_P(n^{-1/4})$ and
$||\hat p_{\theta^*} - p_{\theta^*}|| = o_P(n^{-1/4})$.
Then
$$
\sqrt{n}(\hat \psi -\psi)\rightsquigarrow N(0,\sigma^2)
$$
where $\sigma^2 = \dfrac{1-\psi^2}{2}$.
\end{lemma}

\end{comment}

{\bf Proof of Lemma \ref{lemma::power}.}
We first derive the one-step estimator.
Let the influence function for $r$ be
\begin{align*}
    \phi_r(x) &= \frac{\partial}{\partial \epsilon}
\psi_\gamma(\theta, P + \epsilon (\delta_x - P))\Biggm|_{\epsilon=0} \\
&= 
\frac{\partial}{\partial \epsilon}
\int s_\theta^{1+\gamma} g^{1+\gamma}-\int \Bigl(1+\frac{1}{\gamma}\Bigr) \frac{p+\epsilon(\delta_x-p)}{g}s_\theta^\gamma g^{1+\gamma}
\Biggm|_{\epsilon=0} \\
&= -\Bigl(1+\frac{1}{\gamma}\Bigr)\int  \frac{(\delta_x-p)}{g}s_\theta^\gamma g^{1+\gamma}\\
&=-\Bigl(1+\frac{1}{\gamma}\Bigr)\Bigl[ s_\theta^\gamma(x) g^\gamma(x) +\int  r s_\theta^\gamma g^{1+\gamma}\Bigr].
\end{align*}
Similarly by perturbating $P_\theta$ we get the influence function
\begin{align*}
    \phi_{s_\theta}(x) &= \frac{\partial}{\partial \epsilon}
\psi(\theta, P_\theta + \epsilon (\delta_x - P_\theta))\Biggm|_{\epsilon=0} \\
&= 
\frac{\partial}{\partial \epsilon}
\int \Bigl(\dfrac{p_\theta+\epsilon(\delta_x - p_\theta)}{g}\Bigr)^{1+\gamma} g^{1+\gamma}-\int \Bigl(1+\frac{1}{\gamma}\Bigr) r \Bigl(
\dfrac{p_\theta+\epsilon(\delta_x - p_\theta)}{g}\Bigr)^\gamma g^{1+\gamma}
\Biggm|_{\epsilon=0}  \\
&= \int 
(1+\gamma)\dfrac{p_\theta^\gamma}{g^{1+\gamma}}(\delta_x - p_\theta) g^{1+\gamma}-\Bigl(1+\frac{1}{\gamma}\Bigr) \int  \gamma\cdot r 
\dfrac{p_\theta^{\gamma-1}}{g^\gamma}(\delta_x - p_\theta) g^{1+\gamma}
\\
&= (1+\gamma)s_\theta^\gamma(x)g^\gamma(x)- \gamma\Bigl(1+\frac{1}{\gamma}\Bigr) r (x)
s_\theta^{\gamma-1}(x) g^\gamma (x)-(1+\gamma)\int 
s_\theta^{1+\gamma} g^{1+\gamma} + \gamma\Bigl(1+\frac{1}{\gamma}\Bigr) \int r 
s_\theta^{\gamma} g^{1+\gamma} 
\end{align*}
Now let $Y_i\sim p$, $Y_i^*\sim p_\theta$ and $X_i\sim g$. The above lead to the one-step estimator
\begin{align*}
\widehat\psi_\gamma(p_\theta, p) &= \int \widehat{s}_\theta^{1+\gamma} g^{1+\gamma} - \left(1 + \frac{1}{\gamma}\right) 
\widehat{r}\, \widehat{s}_\theta^\gamma g^{1+\gamma} + \dfrac{1}{n}\sum_i \phi_r (Y_i) + \dfrac{1}{m}\sum_i \phi_{s_\theta} (Y_i^*) \\
&= \dfrac{1+\gamma}{2m}\sum_i   \widehat{s}_\theta^\gamma(Y_i^*) g^\gamma(Y_i^*) -  \dfrac{1+\gamma}{2m}\sum_i \widehat{r} (Y_i^*)
\widehat{s}_\theta^{\gamma-1}(Y_i^*) g^\gamma (Y_i^*) 
-\Bigl(1+\frac{1}{\gamma}\Bigr) \frac{1}{n}\sum_i \widehat{s}_\theta^\gamma(Y_i) g^\gamma(Y_i)\\
&- \gamma \Bigl(\int \widehat{s}_\theta^{1+\gamma} g^{1+\gamma} - \left(1 + \frac{1}{\gamma}\right)\int \widehat{r}\widehat{s}_\theta^\gamma g^{1+\gamma}\Bigr).
\end{align*}
where we estimate the integrals using samples from $X_1, \dots, X_{\tilde{m}}\sim g$, so that for large enough $\tilde{m}$ 
$$
\int \widehat{s}_\theta^{1+\gamma} g^{1+\gamma}\approx \frac{1}{\tilde{m}}\sum_i  \widehat{s}_\theta^{1+\gamma}(X_i) g^{\gamma}(X_i)
,\quad 
\int \widehat{r}\,\widehat{s}_\theta^\gamma g^{1+\gamma} \approx \frac{1}{\tilde{m}}\sum_i \widehat{r}(X_i)\widehat{s}_\theta^\gamma (X_i) g^{\gamma}(X_i)
$$
Asymptotic normality of the one-step estimator follows from the same argument in the proof of Lemma \ref{lemma::hellinger}.
The asymptotic variance is 
\begin{align*}
    \mathbb{E}[\phi_r^2+\phi_{s_\theta}^2] &= \mathbb{E}_P\left[
    \Bigl(1+\frac{1}{\gamma}\Bigr)^2\Biggl(-s_\theta^\gamma(x) g^\gamma(x) + \int  r s_\theta^\gamma g^{1+\gamma}\Biggr)^2\right] \\
    &+ \mathbb{E}_{P_\theta}\left[(1+\gamma)^2\Biggl(s_\theta^\gamma(x)g^\gamma(x)- r (x)s_\theta^{\gamma-1}(x) g^\gamma (x)-\int s_\theta^{1+\gamma} g^{1+\gamma} +  \int r s_\theta^{\gamma} g^{1+\gamma}\Biggr)^2\right] \\
    &= \Bigl(1+\frac{1}{\gamma}\Bigr)^2 \Biggl(\int s_\theta^{2\gamma} g^{2\gamma} p +\Bigl(\int  r s_\theta^\gamma g^{1+\gamma}\Bigr)^2 - 2 \int  r s_\theta^\gamma g^{1+\gamma}\int s_\theta^\gamma g^\gamma p \Biggr)\\
    &+(1+\gamma)^2\Biggl(\int s_\theta^{2\gamma}g^{2\gamma}p_\theta  + \int r^2 s_\theta^{2(\gamma-1)} g^{2\gamma} p_\theta 
    + \Bigl(\int s_\theta^{1+\gamma} g^{1+\gamma}\Bigr)^2 + \Bigl(\int r s_\theta^{\gamma} g^{1+\gamma}\Bigr)^2\\
    & - 2\int s_\theta^\gamma r s_\theta^{\gamma-1} g^{2\gamma}p_\theta - 2\int s_\theta^{1+\gamma} g^{1+\gamma}\int s_\theta^\gamma g^\gamma p_\theta 
    +  2\int r s_\theta^{\gamma} g^{1+\gamma} \int s_\theta^\gamma g^\gamma p_\theta \\
    & + 2\int s_\theta^{1+\gamma} g^{1+\gamma}\int r s_\theta^{\gamma-1} g^\gamma p_\theta 
    -  2\int r s_\theta^{\gamma} g^{1+\gamma} \int r s_\theta^{\gamma-1} g^\gamma p_\theta -2\int s_\theta^{1+\gamma} g^{1+\gamma}  \int r s_\theta^{\gamma} g^{1+\gamma}\Biggr)
    \\
    &= \Bigl(1+\frac{1}{\gamma}\Bigr)^2 \int s_\theta^{2\gamma} g^{2\gamma} p 
    +(1+\gamma)^2 \int s_\theta^{2\gamma}g^{2\gamma}p_\theta 
    + (1+\gamma)^2 \int r^2 s_\theta^{2\gamma-1} g^{2\gamma+1} 
    - 2(1+\gamma)^2\int  r s_\theta^{2\gamma} g^{2\gamma+1} \\ 
    &\quad 
    - (1+\gamma)^2\Bigl(\int s_\theta^{1+\gamma} g^{1+\gamma}\Bigr)^2
    - \Bigl(1+\frac{1}{\gamma}\Bigr)^2\Bigl(\int  r s_\theta^\gamma g^{1+\gamma}\Bigr)^2
    - (1+\gamma)^2 \Bigl(\int r s_\theta^{\gamma} g^{1+\gamma} \Bigr)^2\\
    &\quad 
    +  2(1+\gamma)^2\int r s_\theta^{\gamma} g^{1+\gamma} \int s_\theta^{1+\gamma} g^{1+\gamma} \\
    &= \Bigl(1+\frac{1}{\gamma}\Bigr)^2 \E_P \Bigl[ s_\theta^{2\gamma} g^{2\gamma} \Bigr]
    +(1+\gamma)^2 \E_{p_\theta}\Bigl[ s_\theta^{2\gamma}g^{2\gamma}\Bigr]
    + (1+\gamma)^2 \E_g \Bigl[r^2 s_\theta^{2\gamma-1} g^{2\gamma}\Bigr] 
    - 2(1+\gamma)^2\E_g \Bigl[  r s_\theta^{2\gamma} g^{2\gamma}\Bigr]  \\ 
    &\quad 
    - (1+\gamma)^2\Biggl(
    \Bigl(\int s_\theta^{1+\gamma} g^{1+\gamma}\Bigr)^2
    + \Bigl(1+\frac{1}{\gamma^2}\Bigr) \Bigl(\int  r s_\theta^\gamma g^{1+\gamma}\Bigr)^2 -2 \int r s_\theta^{\gamma} g^{1+\gamma} \int s_\theta^{1+\gamma} g^{1+\gamma} \\
    &\qquad\qquad\qquad\ 
    \pm \frac{2}{\gamma}  \Bigl(\int  r s_\theta^\gamma g^{1+\gamma}\Bigr)^2
    \pm \frac{2}{\gamma} \int r s_\theta^{\gamma} g^{1+\gamma} \int s_\theta^{1+\gamma} g^{1+\gamma} \Biggr) \\
    &= \Bigl(1+\frac{1}{\gamma}\Bigr)^2 \E_p \Bigl[ s_\theta^{2\gamma} g^{2\gamma} \Bigr]
    - \frac{2(1+\gamma)^2}{\gamma} \E_{p}\Bigl[ s_\theta^{\gamma} g^{\gamma}\Bigr] \E_{p_\theta}\Bigl[ s_\theta^{\gamma} g^{\gamma}\Bigr]
    +(1+\gamma)^2 \E_{p_\theta}\Bigl[ s_\theta^{2\gamma}g^{2\gamma}\Bigr]\\
    &\quad + (1+\gamma)^2 \E_g \Bigl[r^2 s_\theta^{2\gamma-1} g^{2\gamma}\Bigr] 
    - 2(1+\gamma)^2\E_g \Bigl[  r s_\theta^{2\gamma} g^{2\gamma}\Bigr]  
    + \frac{2(1+\gamma)^2}{\gamma}  \Bigl(\E_g \Bigl[ r s_\theta^\gamma g^{\gamma}\Bigr]\Bigr)^2
     \\
    &\quad- (1+\gamma)^2\psi_\gamma^2\\
    &=  
    \E_{pp_\theta} \Biggl[\Biggl( \Bigl(1+\frac{1}{\gamma}\Bigr) s_\theta^{\gamma}(Y) g^{\gamma}(Y) - 
    (1+\gamma)^2 s_\theta^{\gamma}(Y^*)g^{\gamma}(Y^*)\Biggl)^2\Biggr]\\
    &\quad + (1+\gamma)^2 \E_g \Bigl[r^2 s_\theta^{2\gamma-1} g^{2\gamma}\Bigr] 
    - 2(1+\gamma)^2\E_g \Bigl[  r s_\theta^{2\gamma} g^{2\gamma}\Bigr]  
    + \frac{2(1+\gamma)^2}{\gamma}  \Bigl(\E_g \Bigl[ r s_\theta^\gamma g^{\gamma}\Bigr]\Bigr)^2
     \\
    &\quad- (1+\gamma)^2\psi_\gamma^2
\end{align*}

{\bf Proof of Theorem \ref{thm::whatisthis}.}
We can write $\hat\Delta(\theta_1,\theta_2)$ as in eq.\ref{eq::diff_discr}.
Under weak conditions 
(which do not require regularity conditions on the model),
$$
\frac{ \hat\Delta(\theta_1,\theta_2) - \Delta(\theta_1,\theta_2) }
{s(\theta_1, \theta_2)}
\rightsquigarrow N(0,1).
$$

Now split the data into two parts
${\cal D}_0$ and
${\cal D}_1$.
For notational simplicity, assume each has sample size $n_0=n_1=n$. 
%\V{(we used $n_0$ and $n_1$ 
%somewhere. We need to settle on a notation)}.
Let $\hat\theta$ be any estimator
computed from ${\cal D}_0$.
Let
\begin{equation}
C = \{ \theta:\ \hat\Delta(\theta,\hat\theta)\leq t_n(\theta)\}\label{eq::cs_relative_fit}
\end{equation}
where
$t_n(\theta) = s(\theta,\hat\theta)z_\alpha/\sqrt{n}$.
Note that
$$\Delta(\theta^*,\hat \theta) = 
d(p_{\theta^*},p)-d(p_{\hat \theta},p)
\leq 0
$$
since $\theta^*$ minimizes
$D(\theta)$.
Therefore, conditional on ${\cal D}_0$,
\begin{align*}
P(\theta^* \notin C) &=
P \left( \hat\Delta(\theta,\hat\theta)\leq t_n(\theta)\right)\\
&=
P \left( \sqrt{n}(\hat\Delta(\theta,\hat\theta)- \Delta(\theta,\hat\theta) )
\leq \sqrt{n}(t_n(\theta) - \Delta(\theta,\hat\theta)) \right)\\
& \leq
P \left( \sqrt{n}(\hat\Delta(\theta,\hat\theta)- \Delta(\theta,\hat\theta))
\leq \sqrt{n}t_n(\theta) \right)\ \ 
\text{since\ } \Delta(\theta^*,\hat\theta) \leq 0\\
& \to
P(Z > z_\alpha) = \alpha.
\end{align*}

\paragraph{Proof of Theorem \ref{thm::4}.}
From 
Lemma \ref{lemma::Lip_Wass}, using assumption (2),
$\hat T_n = T_n + O_P(\sqrt{\log N/N})+ O_P(\sqrt{\log M/M})$.
Similarly,
$\hat T_n(\theta) = T_n(\theta) + O_P(\sqrt{\log N/N})+ O_P(\sqrt{\log M/M})$
uniformly in $\theta$. 
In this proof we write $w_r(\theta):= \frac{K_h(\theta_r - \theta)}{\sum_r K_h(\theta_r - \theta)} $.
We have,
\begin{align*}
    \widehat{p}(\theta)= \sum_r w_r(\theta) I( \hat T_n(\theta_r) \geq \hat T_n) &=
\sum_r w_r(\theta) I(  T_n(\theta_r) \geq  T_n) +
\sum_r w_r(\theta) \Bigl[ I( \hat T_n(\theta_r) \geq \hat T_n)-  I(  T_n(\theta_r) \geq  T_n)\Bigr] \\
&=
\sum_r w_r(\theta) I(  T_n(\theta_r) \geq  T_n) +
O_P(\sqrt{\log N/N}) + O_P(\sqrt{\log M/M})
\end{align*}
Indeed, let $\widehat{D}_r = \hat T_n(\theta_r) - \hat T_n$ and $D_r =T_n(\theta_r) - T_n$.
Then $\widehat{D}_r=D_r+\delta_{N,M}$ where $$\delta_{N,M}=O_P(\sqrt{\log N/N})+ O_P(\sqrt{\log M/M})$$. 
We want to show that 
$\mathbb{I}(\widehat{D}_r>0)-\mathbb{I}(D_r>0) = O_P(\sqrt{\log N/N})+ O_P(\sqrt{\log M/M})$. 
We note that 
\begin{align*}
    |\mathbb{I}(\widehat{D}_r>0)-\mathbb{I}(D_r>0)| &= |\mathbb{I}(\widehat{D}_r>0)\neq\mathbb{I}(D_r>0)|=\mathbb{I}(\text{sign} (D_r + \delta_{N,M})\neq\text{sign} (D_r))\\
    &\leq \mathbb{I}(D_r \in [-\delta_{N,M}, \delta_{N,M}]) = \mathbb{I}(|D_r| \leq |\delta_{N,M}|).
\end{align*}
Since $\delta_{N,M}\leq C_1  \sqrt{\log N/N}+C_2  \sqrt{\log M/M}$ for some $C_1,C_2>0$, then 
$$
|\mathbb{I}(\widehat{D}_r>0)-\mathbb{I}(D_r>0)|\leq \mathbb{I}\Bigl(|D_r| \leq  C_1  \sqrt{\log N/N}+C_2  \sqrt{\log M/M}\Bigr).
$$
Let $Z_r=|\mathbb{I}(\widehat{D}_r>0)-\mathbb{I}(D_r>0)|$, then
$Z_r=O_P\Biggl(\mathbb{P}\Bigl(|D_r| \leq C_1  \sqrt{\log N/N}+C_2  \sqrt{\log M/M}\Bigr)\Biggr)$. 
By definition of boundedness of probability, we want to find $a_{N,M}$ such that for some large $A>0$ and small $\varepsilon>0$,
$\mathbb{P}(Z_r > A\cdot a_{N,M})\leq \varepsilon$. We use Markov inequality fo find such $a_{N,M}$ as follows
\begin{align*}
    \mathbb{P}(Z_r > A \cdot a_{N,M})\leq \dfrac{\mathbb{E}[Z_r]}{A \cdot a_{N,M}} = 
    \dfrac{\mathbb{P}(Z_r=1)}{A \cdot a_{N,M}}\leq \dfrac{\mathbb{P}\Bigl(|D_r| \leq C_1  \sqrt{\log N/N}+C_2  \sqrt{\log M/M}\Bigr)}{A \cdot a_{N,M}}
\end{align*}
By setting $a_{N,M}=\mathbb{P}\Bigl(|D_r| \leq C_1  \sqrt{\log N/N}+C_2  \sqrt{\log M/M}\Bigr)$ and $\varepsilon=1/A$, we get the desired result.
Setting $\varphi_{N,M} =  n^\xi\Bigl(C_1 \sqrt{\log N/N}+C_2 \sqrt{\log M/M}\Bigr)$,
$$
\mathbb{P}\Bigl(|D_r| \leq \varphi_{N,M}\Bigr) = \int_{-\varphi_{N,M}}^{\varphi_{N,M}} f_{D}(u)\text{d} u\leq C_{\max} \int_{-\varphi_{N,M}}^{\varphi_{N,M}}\text{d} u = \tilde{C}\cdot n^\xi\cdot\Bigl(\sqrt{\log N/N}+ \sqrt{\log M/M}\Bigr) 
$$
where  $C_{\max} = \max_{u\in [-\varphi_{N,M}, \varphi_{N,M}]} f_D(u)$ and $\tilde{C} = C_{\max} \cdot \max{\{C_1,C_2\}} $. 
This shows that 
$$
\mathbb{I}(\widehat{D}_r>0)-\mathbb{I}(D_r>0) = O_P\left(n^\xi\Bigl(\sqrt{\log N/N} + \sqrt{\log M/M} \Bigr)\right)
$$
Since $\sum_r w_r(\theta) = 1$ and
$w_r(\theta)\leq 1$,
$$
\sum_{r=1}^N w_r(\theta)[\mathbb{I}(\widehat{D}_r>0)-\mathbb{I}(D_r>0)] = O_P\left(n^\xi\Bigl(\sqrt{\log N/N} + \sqrt{\log M/M} \Bigr)\right).
$$

By standard kernel arguments,
\begin{align*}
\hat p (\theta) &=
p (\theta) + 
O_P(h^\beta + (Nh^d)^{-1/2}) + O_P(\sqrt{\log N/N}) + O_P(\sqrt{\log M/M}) +\\
&\qquad + O_P\left(n^\xi\Bigl(\sqrt{\log N/N} + \sqrt{\log M/M} \Bigr)\right)
\end{align*}
Moreover, this bound is uniform in $\theta$.
So 
$$
\hat p = \max_j \hat p(\theta_j) \leq 
\sup_\theta p(\theta) + 
O_P(h^\beta + (Nh^d)^{-1/2})  + 
O_P\left(\left(1+n^\xi\right)\Bigl(\sqrt{\log N/N} + \sqrt{\log M/M} \Bigr)\right)
$$
and the result follows.
Note the optimal kernel bandwidth is obtained by setting 
$$
h^\beta\asymp \Bigl(\dfrac{1}{Nh^d}\Bigr)^{1/2} \iff h\asymp N^{-1/(d+2\beta)}
$$
which leads to the final bound.

\paragraph{Proof of Lemma~\ref{lemma::Lip_Wass}.}
We have
\begin{align*}
|\min_j W(P_M(\theta_j),Q) - \inf_\theta W(P_\theta,Q)| &\leq
|\min_j W(P_M(\theta_j),Q) - \min_j W(P_{\theta_j},Q)|\\
&\ \ \ +
|\min_j W(P_{\theta_j},Q) - \inf_\theta W(P_\theta,Q)|.
\end{align*}

For the first term,
\begin{align*}
P(|\min_j W(P_M(\theta_j),Q) - \min_j W(P_{\theta_j},Q)|> \epsilon) & \leq P(\max_j |W(P_M(\theta_j),Q) - W(P_{\theta_j},Q)|> \epsilon)\\
& \leq \sum_j P( |W(P_M(\theta_j),Q) - W(P_{\theta_j},Q)|> \epsilon)\\
&\leq N e^{- c M \epsilon^2} \leq M e^{- c M \epsilon^2} 
\end{align*}

from Theorem 2 of \cite{fournier2015rate},
where we have assumed that the dimension of $Y$ is less than 4.
This implies that
the first term $O_P(\sqrt{\log N/N})$.
A similar argument applies when dimension is greater than or equal to 4
with a slight change in the form of the exponential term.

By the Lipschitz condition,
the second term is $O(\delta)$ where
$\delta = \sup_{\theta\in\Theta}\min_j ||\theta - \theta_j||$
and
$\delta = O_P(\sqrt{\log N/N})$
since $\Theta$ is compact and
$\pi$ is strictly positive. 

As a side note, since $N\leq M$, we also have 
$$P(|\min_j W(P_M(\theta_j),Q) - \min_j W(P_{\theta_j},Q)|> \epsilon)\leq N e^{- c M \epsilon^2} \leq N e^{- c N \epsilon^2}$$
This implies that both terms  appearing at the start of this proof are of order $O_P(\sqrt{\log N / N})$, yielding a sharper bound than the one stated in the current version of the lemma. Nevertheless, we retain both terms for the sake of generality.

\paragraph{Proof of Theorem~\ref{thm::6}.}
Let $k$ denote the dimension of $Y$.
By Theorem 2 of
\cite{fournier2015rate},
there are constants $c$ and $C$ such that
$$
P_\theta(W(P_M^*(\theta),P_\theta) > \epsilon) \leq
\begin{cases}
C e^{-cM\epsilon^2} & \mathrm{if}\ 4 > k\\
C e^{- cM (\epsilon /\log(2 + 1/\epsilon))^2} & \mathrm{if}\ 4=k\\
C e^{-cM\epsilon^{q/2}} & \mathrm{if}\ 4 <k.
\end{cases}
$$
To avoid repetition, we'll assume that $4>k$
but the other cases are similar.
Let
$\mathcal{C} = \{ \theta_1,\ldots, \theta_R\}$
be a $\epsilon/4$ covering set of $\Theta$.
Thus, for each $\theta$ there is a $\theta_j \in \mathcal{C} $ such that
$||\theta-\theta_j|| \leq \epsilon/4$.
Note that $R \leq (c_1/\epsilon)^k$ for some $c_1$.
Let
$\mathbb{P} = \prod_{j=1}^N P_{\theta_j}^M$
denote the product measure.
By the above exponential inequality above
and the Lipschitz property,
%\lortomas{are we assuming that the function is L=1 Lipschitz? or else are we going to build a $\epsilon/2L$ covering of the parameter space (just for clarity in the first inequality, it is going to be incorporated in a different constant for the covering number and final bound)? }
\begin{align*}
\mathbb{P}(\sup_\theta |W(P_M^*(\theta),P_n)-W(P_\theta,P_n)| > \epsilon) & \leq
\mathbb{P}(\max_{\theta\in \in \mathcal{C}} |W(P_M^*(\theta),P_n)-W(P_\theta,P_n)| +\epsilon/2 > \epsilon)\\
\leq \left(\frac{\tilde{c}}{\epsilon}\right)^{k} C e^{-cM\epsilon^2}.
\end{align*}
To show the above inequality for some $\theta$ and $ \theta_c\in \Theta$, we focus on
\begin{align*}
    |W(P_M(\theta), P_n)-W(P(\theta), P_n))|
    &\leq 
    \underset{(i)}{
    |W(P_M(\theta), P_n)-W(P_M(\theta_c),P_n)|}
    +
    \underset{(ii)}{|W(P(\theta), P_n))-W(P(\theta_c), P_n))|}
    \\
    &\quad +\underset{(iii)}{|W(P_M(\theta_c), P_n))-W(P(\theta_c), P_n))|}\\
    & \leq |W(P_M(\theta_c), P_n))-W(P(\theta_c), P_n))| + \epsilon/2
\end{align*}
where we used triangle inequality and $(i)\leq L||\theta -\theta_c||\leq \epsilon/4$, assuming $L=1$, alternatively we construct $\epsilon/4L$ covering sets for the inequality to hold, and $(ii)\leq \epsilon/4$ similarly. By taking the supremum of the LHS and the maximum over the covering set on the RHS, then using the union bound, we get the result
\begin{align*}
    \mathbb{P}(\sup_\theta |W(P_M^*(\theta),P_n)-W(P_\theta,P_n)| > \epsilon) &\leq
\mathbb{P}(\max_{\theta\in \mathcal{C}} |W(P_M^*(\theta),P_n)-W(P_\theta,P_n)|> \epsilon/2)\\
    &\leq \sum_{\tilde{\theta}\in\mathcal{C}} \mathbb{P}(|W(P_M^*(\tilde\theta),P_n)-W(P_{\tilde\theta},P_n)|>\epsilon/2)\\
    &\leq |\mathcal{C}| \ \mathbb{P}(W(P_M^*(\theta), P_\theta) >\epsilon/2) 
\end{align*}
Where $|\mathcal{C}|\leq (\frac{\tilde{c}}{\epsilon})^k$ is the covering number, for some constant $\tilde{c}>0$, and for any $P_M(\theta),P(\theta), P_n$ by non-negativity of the distance $W(\cdot,\cdot)$  and the triangle inequality 
\begin{align*}
W(P_M(\theta),P_\theta) \geq |W(P_M(\theta),P_n)-
W(P_\theta,P_n) | 
\end{align*}
From Theorem 2 of \cite{fournier2015rate} with $k<4$, it follows:
$$
\mathbb{P}(\sup_\theta |W(P_M^*(\theta),P_n)-W(P_\theta,P_n)| > \epsilon)
\leq
\Bigl(\ \dfrac{\tilde{c}}{\epsilon}\ \Bigr)^k C e^{-cM\epsilon^2} 
$$
Setting $\epsilon = \sqrt{
\log M/cM}$ we conclude that
$\sup_\theta |W(P_M^*(\theta),P_n)-W(P_\theta,P_n)| = O_P(\sqrt{\log M/M})$.%, or $O_P(\sqrt{k\log M/M})$ when accounting for the dimension of $Y$.
This comes by setting 
the right hand side to some $\delta>0$
then analyzing the asymptotic behavior (up to a constant) of both sides of the derived equation $k\log (1/\epsilon)\approx -cM\epsilon^2$. 
A similar argument shows that
$\sup_\theta |W(P_M^*(\theta),P_n(\theta))-W(P_\theta,P_n(\theta))| = O_P(\sqrt{\log M/M})$.

\section{Confidence Sets}

\subsection{Asymptotic confidence sets for the Projection Estimator Under Regularity}

Since 
the projection estimator
$\hat\theta$
is the minimizer of the distance,
it is an $m$-estimator
so, if the standard regularity conditions hold,
then
one can use the usual
asymptotic confidence set
$$
C = \Bigl\{ \theta:\ n(\theta-\hat\theta)^T \hat\Sigma^{-1} (\theta-\hat\theta)\leq 
\chi^2_{\alpha,d}\Bigr\}
$$
where 
$\Sigma = G^{-1} M (G^{-1})^\top$,
$G = \E[\nabla U(Y,\theta^*)]+\E[\nabla V(Y(\theta^*),\theta^*)]$ and 
$M=-\left(\E[U(Y,\theta^*)]+\E[V(Y(\theta^*),\theta^*)]\right)$.
This approach might be problematic in SBI for two reasons.
First, we may not have access to 
$\nabla U(Y,\theta^*)$.
The HulC (\cite{kuchibhotla_hulc_2024}) provides a solution in that case.
The data are split into $B$ batches, with 
$B = \lceil \log(2/\alpha)/\log(2)\rceil$, and
estimates $\hat \theta_1,\ldots,\hat \theta_B$ are obtained 
from each batch.
If the median bias of $\hat\theta$
tends to zero 
(which holds under the usual regularity conditions)
then $[\min_j \hat\theta_j(r),\max_j \hat\theta_j(r)]$
is a $1-\alpha$ confidence set for $\theta(r)$, 
the $r$-th component $\theta$. 
The bootstrap provides an alternative.
The usual bootstrap requires obtaining estimates of $\theta^*$ in each of many 
bootstrap samples, which is undesirable in SBI
since computations can be expensive.
Instead, the {\em cheap bootstrap}
\citep{lam_cheap_2023} allows us to construct confidence intervals
using only $b$
bootstrap samples, with $b$ small.
The interval for a parameter $\psi$ is then
$$
C = \Bigl[\hat\psi-t_{b,\alpha/2}|\hat\psi-\psi^*|,\ 
\hat\psi+t_{b,\alpha/2}|\hat\psi-\psi^*|\Bigr]
$$
where
$\hat\psi$ is the original estimator,
$\hat\psi^*$ is the estimator
constructed from bootstrap sample 
and $t_{b,\alpha/2}$
is the upper $\alpha/2$ quantile
of a $t$ distribution with $b$ degrees of freedom. 
Note that too small a $b$ produce conservative confidence sets (larger width), 
while $b$ large defies the purpose of using the cheap bootstrap.
In our applications we found that $b$ as low as 
$5$ produce satisfactory results; see Appendix Fig.~\ref{fig:cheap_bootstrap_CS}
for more details.

\subsection{Asymptotic confidence sets for Kernel distance estimator}

We present the asymptotic confidence sets for the projection parameters obtained by minimizing the MMD, as used in the applications in Section~\ref{sec::applications_iid}.

\paragraph{Gaussian location.} Let $Y_1,\dots,Y_m\sim\mathcal{N}(\theta^*, \sigma^2I_{d\times d})$ and $Y_1^*,\dots,Y_m^*\sim\mathcal{N}(\theta, \sigma^2I_{d\times d})$.
Under the assumptions of Theorem 2 in \cite{briol_statistical_2019} we have
$$\sqrt{(n\wedge m)}(\hat{\theta}_{n,m}-\theta^*)\xrightarrow{d} \mathcal{N}(0, C_\lambda)$$
where 
$\hat{\theta}_{n,m}$ is the minimizer of Eq.~(\ref{eq_mmd}) using a Gaussian kernel with bandwidth $h^2$, $C_\lambda=\frac{1}{\lambda(1-\lambda)}C$, with $\lambda=\frac{n}{n+m}$ and 
$$C=\sigma^2((h^2+\sigma^2)(3\sigma^2+h^2))^{-\frac{d}{2}-1}(h^2+2\sigma^2)^{d+2}$$
The CI in the 1d case is
$$\hat\theta_{n,m}\pm z_{\frac{\alpha}{2}}\sqrt{\frac{C}{(n\wedge m)}}$$

\paragraph{Gaussian scale.} 
Let $Y_1,\dots,Y_m\sim\mathcal{N}(\mu, \theta_*^2 I_{d\times d})$ and $Y_1^*,\dots,Y_m^*\sim\mathcal{N}(\mu, \theta^2 I_{d\times d})$.
From the proof of proposition 7 (Appendix C.6 in \cite{briol_statistical_2019}) the asymptotic variance of the CLT corresponds to: $$C=\dfrac{(h^2+2s)^2\left(((h^2+s)^{-\frac{d}{2}-2}(h^2+2s)^{d+2}(h^2+3s)^{-\frac{d}{2}-2}\left((h^2+2s)^2+2\frac{s^2}{d}\right)-1\right)}{(d+2)^2s^2}
$$
where $s=2\theta^*$. 
We estimate $\hat{C}$ by plugging in $\hat\theta_{n,m}$ in place of $\theta^*$ and build the
the CI in the 1d case as done in the Gaussian location case.

\subsection{Confidence sets using the cheap bootstrap}
\label{section::cheap}

The cheap bootstrap procedure proposed by \citep{lam_cheap_2023} allows us to derive confidence sets with desired theoretical guarantees with limited computational efforts. This is particularly valuable for the discrepancy-based exponential tilt approach 
(Section \ref{sec::exponential_tilt}),
which requires minimizing the loss function for all $\theta$ over a grid to estimate the exponential tilt model parameters, i.e. $\left(\hat{\theta}, \hat{\beta}(\hat\theta)\right)$. 
In that section we computed relative fit confidence sets for $\theta^*$
using the profile likelihood.
To compute confidence sets for both the model and tilting parameters, multiple repetitions of the minimization procedure would be necessary. We limit the number of bootstrap iterations using the cheap bootstrap approach described in Algorithm~\ref{alg::cheap_bootstrap}. 
Fig.~\ref{fig:cheap_bootstrap_CS} presents confidence sets ($B=15$ bootstrap iterations) for the exponential tilt parameters in the settings of Fig.~\ref{fig:exponential_tilt_normal_profile_likelihood} when minimizing the L2 loss.
These sets are informative and cover the true parameter values (red line).
In a separate work, we  investigated  how confidence sets coverage and width vary as a function of bootstrap iterations, for inference on the variance of a folded standard normal distribution (i.e., $|\mathcal{N}(0,1)|$, true variance $\theta^*=1-2/\pi$).
We achieved relatively short confidence sets with as few as 5 bootstrap iterations, while maintaining coverage at or above the nominal level. 

\RestyleAlgo{ruled}
\begin{algorithm}[hbt!]
\caption{Confidence sets construction via the cheap bootstrap approach. \vspace{0.05cm}}\label{alg::cheap_bootstrap}
\SetAlgoLined
\SetKwInOut{Input}{Input}
\SetKwInOut{Output}{Output}
\Input{
\begin{itemize}
    \item[] observed data  ${\cal Y}=Y_1, \dots, Y_n\sim p$\vspace{-.075in}
    \item[] parameter values $\theta_1,\ldots,\theta_N \sim \pi$; \vspace{-.075in}
    \item[] initial values $\beta_1^{init},\dots, \beta_N^{init}$ obtained using Algorithm \ref{alg:tilt};\vspace{-.075in}
    \item[] loss tolerance $\epsilon$; max number of iterations $\iota_X$; learning rate $\delta$;\vspace{-.075in}
    \item[] number of bootstrap iterations $B$;\vspace{-.075in}
\end{itemize}
}
\Output{Confidence sets for $(\theta, \beta(\theta))$.\vspace{0.05cm}}
\For{$b=1,\dots,B$}{
  {\bf sample} $\mathcal{Y}^{\,b}=Y_{i_1},\ldots, Y_{i_n}$ where $i_1,\dots,i_n$ is a permutation of the index set;\\ 
  {\bf simulate} $\mathcal{X}^{\,b}=X_1, \dots,X_m\sim g$ \\
  {\bf estimate} the density ratio $\hat{r}_\theta^{\,b}:\mathbb{R}\mapsto \mathbb{R}$ from $\mathcal{Y}^{\,b}$ and $\mathcal{X}^{\,b}$;\\
  {\bf estimate} $$(\widehat{\theta}^{\,b},\widehat{\beta}^{\,b}(\widehat{\theta}^{\,b}))=\argmin_j S_n(\theta_j^{\,b}, \widehat{\beta}^{\,b}(\theta_j^{\,b}))$$ 
  \qquad\qquad\quad  
  e.g., via NR method in \ref{app::NR} using gradient and Hessian in \ref{sec::appx_mdpd} or \ref{sec::appx_hellinger} depending on $S_n$%using Algorithm~\ref{alg:L2_loss_tilt_minimizer}; \V{XXXXX}
}
{\bf compute} \quad $\bar{\theta}^{\,B}=\frac{1}{b}\sum_{b=1}^B \widehat{\theta}^{\,b}$, $S_{\bar{\theta}^{\,B}}=\sqrt{\frac{1}{b}\sum_{b=1}^B \left(\widehat{\theta}^{\,b}-\bar{\theta}^{\,B}\right)^2}$ and similarly for each tilting parameter;\\
{\bf return} $$\bar{\theta}^{\,B}\pm q_{\alpha/2, t_B}\cdot S_{\bar{\theta}^{\,B}}\,,\quad \bar{\beta}_i^{\,B}\pm q_{\alpha/2, t_B}\cdot S_{\bar{\beta}_i^{\,B}}\ \text{ for}\ i=1,\dots,k$$
\qquad\qquad where $q_{\alpha/2, t_B}$ is the $1-\alpha/2$ quantile of the $t$ distribution with $B$ degrees of freedom.
\vspace{.1cm}
\end{algorithm}

\begin{figure}[t!]
    \centering
%    \subfloat[Analysis of coverage and width for cheap bootstrap confidence sets for the variance of a folded normal distribution.]{\includegraphics[width=7.75cm]{figures/appendix/cheap_bootstrap_figure_folded_normal.png}} 
 %   \qquad
    \includegraphics[width=7.75cm,height=3cm]{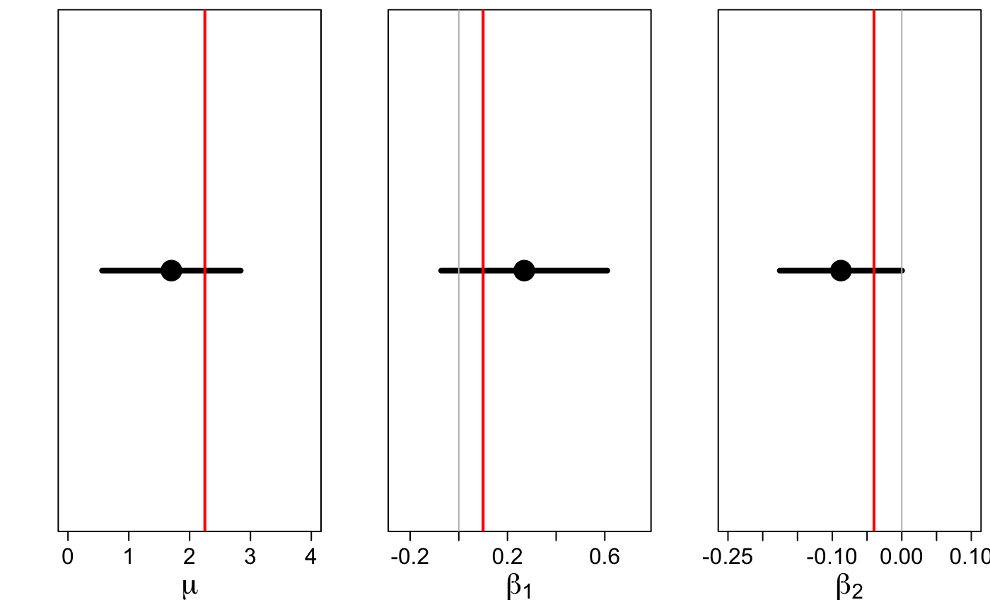}
    \caption{\sl 
     {\bf Cheap bootstrap} confidence sets for the parameters of an exponentially tilted normal distribution (with unknown location) when minimizing the $L_2$ loss. Same settings as in Fig.~\ref{fig:exponential_tilt_normal_profile_likelihood}.
}
    \label{fig:cheap_bootstrap_CS}
\end{figure}

\section{Exponential tilt - Gradient and Hessian}\label{app::grad_hess}

We derive the analytical gradient and Hessian for the discrepancies discussed in Section~\ref{sec::exponential_tilt}. Note that for the results in Section~\ref{sec::applications_iid}, we used gradients and Hessians estimated via optimization algorithms.

\subsection{Efficient MDPD loss}\label{sec::appx_mdpd}
\begin{align}
    S_n\left(\theta, \beta\right) &=  
    \dfrac{1+\gamma}{m} \sum_i   \widehat{r}_\theta^{\gamma}(Y^*_i)g^{\gamma}(Y^*_i)\Biggl(\dfrac{e^{\,\beta^\top b(Y_i^*)}}{c(\theta, \beta)}\Biggr)^\gamma 
    -   \dfrac{1+\gamma}{m} \sum_i \widehat{r}_\theta^{\gamma-1}(Y^*_i)\widehat{r}(Y^*_i)g^{\gamma}(Y^*_i)\Biggl(\dfrac{e^{\,\beta^\top b(Y^*_i)}}{c(\theta, \beta)}\Biggr)^{\gamma-1}\nonumber\\
	&- \left(1+\dfrac{1}{\gamma}\right)\dfrac{1}{n} \sum_i \widehat{r}_\theta^\gamma (Y_i) g(Y_i)^\gamma \Biggl(\dfrac{e^{\,\beta^\top b(Y_i)}}{c(\theta, \beta)}\Biggr)^{\gamma} 
 - \dfrac{\gamma }{\tilde{m}} \sum_i \widehat{r}^{1+\gamma}_\theta(\tilde{Y}_i) g^\gamma(\tilde{Y}_i)\Biggl(\dfrac{e^{\,\beta^\top b(\tilde{Y}_i)}}{c(\theta, \beta)}\Biggr)^{1+\gamma}\nonumber\\
     &+ \dfrac{1+\gamma}{\tilde{m}} \sum_i 
     \widehat{r}(\tilde{Y}_i)\widehat{r}^{\gamma}_\theta(\tilde{Y}_i) g^\gamma(\tilde{Y}_i) \Biggl(\dfrac{e^{\,\beta^\top b(\tilde{Y}_i)}}{c(\theta, \beta)}\Biggr)^{\gamma}
\end{align}

The gradient is
\begin{align}
    \nabla_\beta S_n\left(\theta, \beta\right) &=  \frac{1+\gamma}{m}\sum_i \hat r_\theta^\gamma (Y_i^*) g^\gamma (Y_i^*) \dfrac{\gamma e^{\,\gamma\beta^\top b(Y_i^*)}}{c(\theta, \beta)^{\gamma}}\left(\dfrac{ b(Y_i^*)c(\theta, \beta)- \nabla_\beta c(\theta, \beta)}{c(\theta, \beta)}\right) \nonumber\\ 
    &- \dfrac{1+\gamma}{m} \sum_i \widehat{r}_\theta^{\gamma-1}(Y^*_i)\widehat{r}(Y^*_i)g^{\gamma}(Y^*_i)
    \dfrac{(\gamma-1)e^{(\gamma-1)\,\beta^\top b(Y^*_i)}}{c(\theta, \beta)^{\gamma-1}}\left(\dfrac{ b(Y_i^*)c(\theta, \beta)- \nabla_\beta c(\theta, \beta)}{c(\theta, \beta)}\right)\nonumber\\
    &- \left(1+\dfrac{1}{\gamma}\right)\dfrac{1}{n} \sum_i \widehat{r}_\theta^\gamma (Y_i) g(Y_i)^\gamma 
    \dfrac{\gamma e^{\,\gamma\beta^\top b(Y_i)}}{c(\theta, \beta)^{\gamma}} \left(\dfrac{ b(Y_i)c(\theta, \beta)- \nabla_\beta c(\theta, \beta)}{c(\theta, \beta)}\right) \nonumber\\
    &- \dfrac{\gamma }{\tilde{m}} \sum_i \widehat{r}^{1+\gamma}_\theta(\tilde{Y}_i) g^\gamma(\tilde{Y}_i)\dfrac{(1+\gamma)e^{\,(1+\gamma)\beta^\top b(\tilde{Y}_i)}}{c(\theta, \beta)^{1+\gamma}}\left(\dfrac{ b(\tilde{Y}_i)c(\theta, \beta)- \nabla_\beta c(\theta, \beta)}{c(\theta, \beta)}\right) \nonumber\\
    &+ \dfrac{1+\gamma}{\tilde{m}} \sum_i 
    \widehat{r}(\tilde{Y}_i)\widehat{r}^{\gamma}_\theta(\tilde{Y}_i) g^\gamma(\tilde{Y}_i) \dfrac{\gamma e^{\,\gamma\beta^\top b(\tilde{Y}_i)}}{c(\theta, \beta)^\gamma} \left(\dfrac{ b(\tilde{Y}_i)c(\theta, \beta)- \nabla_\beta c(\theta, \beta)}{c(\theta, \beta)}\right)
\label{eq::grad_MDPD_exp_tilt}
\end{align}
while the Hessian has $rs$-th entry  
\begin{align}
    \dfrac{\partial^2 S_n\left(\theta, \beta\right)}{\partial \beta_r\partial \beta_s} &=
     \frac{\gamma (1+\gamma)}{m}\sum_i \hat r_\theta^\gamma (Y_i^*) g^\gamma (Y_i^*) \dfrac{e^{\,\gamma\beta^\top b(Y_i^*)}}{c(\theta, \beta)^{\gamma+2}}\Biggl[
     \nonumber\\
    &\qquad\dfrac{\gamma e^{\,\gamma\beta^\top b(Y_i^*)}}{c(\theta, \beta)^{\gamma}}\Bigl(b(Y_i^*)c(\theta, \widehat\beta) - 
    \dfrac{\partial c(\theta, \widehat\beta)}{\partial \beta_r }\Bigr) \Bigl( b(Y_i^*)c(\theta, \widehat\beta) - 
    \dfrac{\partial c(\theta, \widehat\beta)}{\partial \beta_s }\Bigr)+\Biggl(
    \dfrac{\partial c(\theta, \widehat\beta)}{\partial \beta_r }\dfrac{\partial c(\theta, \widehat\beta)}{\partial \beta_s }-\dfrac{\partial^2 c(\theta, \widehat\beta)}{\partial \beta_r \partial \beta_s }
    \Biggr)
    \Biggr]\nonumber\\ 
    &- \dfrac{(1+\gamma)(\gamma-1)}{m} \sum_i \widehat{r}_\theta^{\gamma-1}(Y^*_i)\widehat{r}(Y^*_i)g^{\gamma}(Y^*_i)\dfrac{e^{(\gamma-1)\,\beta^\top b(Y^*_i)}}{c(\theta, \beta)^{\gamma+1}}\Biggl[
     \nonumber\\
&\qquad\dfrac{(\gamma-1)e^{(\gamma-1)\,\beta^\top b(Y^*_i)}}{c(\theta, \beta)^{\gamma-1}}\Bigl(b(Y_i^*)c(\theta, \widehat\beta) - 
    \dfrac{\partial c(\theta, \widehat\beta)}{\partial \beta_r }\Bigr) \Bigl( b(Y_i^*)c(\theta, \widehat\beta) - 
    \dfrac{\partial c(\theta, \widehat\beta)}{\partial \beta_s }\Bigr)\nonumber\\
    &\quad
    +\Biggl(
    \dfrac{\partial c(\theta, \widehat\beta)}{\partial \beta_r }\dfrac{\partial c(\theta, \widehat\beta)}{\partial \beta_s }-\dfrac{\partial^2 c(\theta, \widehat\beta)}{\partial \beta_r \partial \beta_s }
    \Biggr)
    \Biggr]\nonumber\\
    &- \gamma\left(1+\dfrac{1}{\gamma}\right)\dfrac{1}{n} \sum_i \widehat{r}_\theta^\gamma (Y_i) g(Y_i)^\gamma 
    \dfrac{ e^{\,\gamma\beta^\top b(Y_i)}}{c(\theta, \beta)^{\gamma+2}} \Biggl[
     \nonumber\\
&\quad
\dfrac{\gamma e^{\,\gamma\beta^\top b(Y_i)}}{c(\theta, \beta)^{\gamma}}\Bigl(b(Y_i)c(\theta, \widehat\beta) - 
    \dfrac{\partial c(\theta, \widehat\beta)}{\partial \beta_r }\Bigr) \Bigl( b(Y_i)c(\theta, \widehat\beta) - 
    \dfrac{\partial c(\theta, \widehat\beta)}{\partial \beta_s }\Bigr)
    +\Biggl(
    \dfrac{\partial c(\theta, \widehat\beta)}{\partial \beta_r }\dfrac{\partial c(\theta, \widehat\beta)}{\partial \beta_s }-\dfrac{\partial^2 c(\theta, \widehat\beta)}{\partial \beta_r \partial \beta_s }
    \Biggr)
    \Biggr]\nonumber\\ 
    &- \dfrac{\gamma(1+\gamma)}{\tilde{m}} \sum_i \widehat{r}^{1+\gamma}_\theta(\tilde{Y}_i) g^\gamma(\tilde{Y}_i)
    \dfrac{e^{\,(1+\gamma)\beta^\top b(\tilde{Y}_i)}}{c(\theta, \beta)^{3+\gamma}}
    \Biggl[
     \nonumber\\
&\quad\dfrac{(1+\gamma) e^{\,(1+\gamma)\beta^\top b(\tilde{Y}_i)}}{c(\theta, \beta)^{1+\gamma}}\Bigl(b(\tilde{Y}_i)c(\theta, \widehat\beta) - 
    \dfrac{\partial c(\theta, \widehat\beta)}{\partial \beta_r }\Bigr) \Bigl( b(\tilde{Y}_i)c(\theta, \widehat\beta) - 
    \dfrac{\partial c(\theta, \widehat\beta)}{\partial \beta_s }\Bigr)\nonumber\\
    &\quad
    +\Biggl(
    \dfrac{\partial c(\theta, \widehat\beta)}{\partial \beta_r }\dfrac{\partial c(\theta, \widehat\beta)}{\partial \beta_s }-\dfrac{\partial^2 c(\theta, \widehat\beta)}{\partial \beta_r \partial \beta_s }
    \Biggr)
    \Biggr]\nonumber\\
    &+ \dfrac{\gamma (1+\gamma)}{\tilde{m}} \sum_i 
    \widehat{r}(\tilde{Y}_i)\widehat{r}^{\gamma}_\theta(\tilde{Y}_i) g^\gamma(\tilde{Y}_i) \dfrac{e^{\,\gamma\beta^\top b(\tilde{Y}_i)}}{c(\theta, \beta)^{2+\gamma}} 
    \Biggl[
     \nonumber\\
&\quad\dfrac{\gamma e^{\,\gamma\beta^\top b(\tilde{Y}_i)}}{c(\theta, \beta)^\gamma}\Bigl(b(\tilde{Y}_i)c(\theta, \widehat\beta) - 
    \dfrac{\partial c(\theta, \widehat\beta)}{\partial \beta_r }\Bigr) \Bigl( b(\tilde{Y}_i)c(\theta, \widehat\beta) - 
    \dfrac{\partial c(\theta, \widehat\beta)}{\partial \beta_s }\Bigr) +\Biggl(
    \dfrac{\partial c(\theta, \widehat\beta)}{\partial \beta_r }\dfrac{\partial c(\theta, \widehat\beta)}{\partial \beta_s }-\dfrac{\partial^2 c(\theta, \widehat\beta)}{\partial \beta_r \partial \beta_s }
    \Biggr)
    \Biggr]
\label{eq::hessian_MDPD_exp_tilt}
\end{align}
Since, for fixed $\theta$,
$p_{\theta,\beta}$
is an exponential family,
we can use properties of the cumulant function 
to estimate the partial first and second derivatives in~(\ref{eq::grad_MDPD_exp_tilt}) and (\ref{eq::hessian_MDPD_exp_tilt}).
Specifically, let $X_1,\dots,X_n\sim p_{\theta, \beta}$, we have 
$\nabla_\beta c(\theta, \beta)=\mathbb{E}[b(X)]$ with estimator $\widehat{\frac{\partial c(\theta, \beta)}{\partial\beta_r}}=\dfrac{1}{n}\sum_i b_r(X_i)=:\hat\mu_r$. The estimate for second-order partial derivatives is $\widehat{\frac{\partial^2c(\theta, \beta)}{\partial\beta_r\partial\beta_s}}=\dfrac{1}{n}\sum_{i}(b_r(X_i)-\hat\mu_r)(b_s(X_i)-\hat\mu_s)=:\hat\sigma^2_{rs}$.
NR method (algorithm in \ref{app::NR}) is used to estimate $\beta$ via minimization of
$S_n(\theta,\beta)$.

\subsection{Hellinger loss}\label{sec::appx_hellinger}
$$
S_n(\theta, \beta) = 
\frac{1}{2n}\sum_i \sqrt{\dfrac{\widehat{r}_\theta(Y_i)}{\widehat{r}(Y_i)}} \dfrac{e^{\frac{1}{2} \widehat\beta^T b(Y_i)}}{c(\theta,\widehat\beta)^\frac{1}{2}}
+ \frac{1}{2m}\sum_i \sqrt{\dfrac{\widehat{r}(Y_i^*)}{\widehat{r}_\theta(Y_i^*)}}  e^{-\frac{1}{2}\widehat\beta^T b(Y_i^*)}c(\theta,\widehat\beta)^\frac{1}{2}
$$

The gradient corresponds to
\begin{align}
\nabla_\beta S_n(\theta, \beta) &= 
\frac{1}{2n}\sum_i \sqrt{\dfrac{\widehat{r}_\theta(Y_i)}{\widehat{r}(Y_i)}} \dfrac{e^{\frac{1}{2} \widehat\beta^T b(Y_i)}}{2 c(\theta,\widehat\beta)^\frac{1}{2}}\left(\dfrac{ b(Y_i)c(\theta, \widehat\beta)- \nabla_\beta c(\theta, \widehat\beta)}{c(\theta, \widehat\beta)}\right)\nonumber\\
&+ \frac{1}{2m}\sum_i \sqrt{\dfrac{\widehat{r}(Y_i^*)}{\widehat{r}_\theta(Y_i^*)}}  \dfrac{e^{-\frac{1}{2}\widehat\beta^T b(Y_i^*)}}{2} \left(\dfrac{\nabla_\beta c(\theta, \widehat\beta)- b(Y_i^*)c(\theta, \widehat\beta)}{c(\theta,\widehat\beta)^\frac{1}{2}}\right)\label{eq::gradient_hellinger_exp_tilt}
\end{align}

while the Hessian has $rs$-th entry  
\begin{align}
    \dfrac{\partial^2 S_n\left(\theta, \beta\right)}{\partial \beta_r\partial \beta_s} &=
    \frac{1}{2n}\sum_i \sqrt{\dfrac{\widehat{r}_\theta(Y_i)}{\widehat{r}(Y_i)}} 
    \Biggl[
    \dfrac{e^{\frac{1}{2} \widehat\beta^T b(Y_i)}}{2 c(\theta,\widehat\beta)^\frac{1}{2}}\dfrac{ \Bigr(b(Y_i)c(\theta, \widehat\beta) - 
    \dfrac{\partial c(\theta, \widehat\beta)}{\partial \beta_r }\Bigr)}{c(\theta, \widehat\beta)}\times\dfrac{e^{\frac{1}{2} \widehat\beta^T b(Y_i)}}{2 c(\theta,\widehat\beta)^\frac{1}{2}}\dfrac{\Bigl( b(Y_i)c(\theta, \widehat\beta) - 
    \dfrac{\partial c(\theta, \widehat\beta)}{\partial \beta_s }\Bigr)}{c(\theta, \widehat\beta)}\nonumber\\ 
    &\qquad\qquad\qquad\qquad\ + \dfrac{e^{\frac{1}{2} \widehat\beta^T b(Y_i)}}{2 c(\theta,\widehat\beta)^\frac{1}{2}}\Biggl(
    \dfrac{\dfrac{\partial c(\theta, \widehat\beta)}{\partial \beta_r }\dfrac{\partial c(\theta, \widehat\beta)}{\partial \beta_s }-\dfrac{\partial^2 c(\theta, \widehat\beta)}{\partial \beta_r \partial \beta_s }}{c(\theta, \widehat{\beta})^2}
    \Biggr)
    \Biggr]\nonumber\\
    &+ \frac{1}{2m}\sum_i \sqrt{\dfrac{\widehat{r}(Y_i^*)}{\widehat{r}_\theta(Y_i^*)}}  \Biggl[
    \dfrac{-\frac{1}{2}e^{-\frac{1}{2} \widehat\beta^T b(Y_i^*)}}{ c(\theta,\widehat\beta)^{-\frac{1}{2}}}\dfrac{ \Bigr(b(Y_i^*)c(\theta, \widehat\beta) - 
    \dfrac{\partial c(\theta, \widehat\beta)}{\partial \beta_r }\Bigr)}{c(\theta, \widehat\beta)}\times
    \dfrac{-\frac{1}{2}e^{-\frac{1}{2} \widehat\beta^T b(Y_i^*)}}{ c(\theta,\widehat\beta)^{-\frac{1}{2}}}
    \dfrac{\Bigl( b(Y_i^*)c(\theta, \widehat\beta) - 
    \dfrac{\partial c(\theta, \widehat\beta)}{\partial \beta_s }\Bigr)}{c(\theta, \widehat\beta)}\nonumber\\ 
    &\qquad\qquad\qquad\qquad\qquad + \dfrac{-\frac{1}{2}e^{-\frac{1}{2} \widehat\beta^T b(Y_i^*)}}{ c(\theta,\widehat\beta)^{-\frac{1}{2}}}\Biggl(
    \dfrac{\dfrac{\partial c(\theta, \widehat\beta)}{\partial \beta_r }\dfrac{\partial c(\theta, \widehat\beta)}{\partial \beta_s }-\dfrac{\partial^2 c(\theta, \widehat\beta)}{\partial \beta_r \partial \beta_s }}{c(\theta, \widehat{\beta})^2}
    \Biggr)
    \Biggr]\nonumber\\
    &=
    \frac{1}{4n}\sum_i \sqrt{\dfrac{\widehat{r}_\theta(Y_i)}{\widehat{r}(Y_i)}} \dfrac{e^{\frac{1}{2} \widehat\beta^T b(Y_i)}}{c(\theta,\widehat\beta)^\frac{5}{2}}
    \Biggl[
    \dfrac{e^{\frac{1}{2} \widehat\beta^T b(Y_i)}}{2 c(\theta,\widehat\beta)^\frac{1}{2}}
    \Bigl(b(Y_i)c(\theta, \widehat\beta) - 
    \dfrac{\partial c(\theta, \widehat\beta)}{\partial \beta_r }\Bigr) \Bigl(b(Y_i)c(\theta, \widehat\beta)-
    \dfrac{\partial c(\theta, \widehat\beta)}{\partial \beta_s }\Bigr)\nonumber \\
    &\qquad\qquad\qquad\qquad\qquad\qquad\qquad + \Biggl(
    \dfrac{\partial c(\theta, \widehat\beta)}{\partial \beta_r }\dfrac{\partial c(\theta, \widehat\beta)}{\partial \beta_s }-\dfrac{\partial^2 c(\theta, \widehat\beta)}{\partial \beta_r \partial \beta_s }
    \Biggr)
    \Biggr]\nonumber\\
    &- \frac{1}{4m}\sum_i \sqrt{\dfrac{\widehat{r}(Y_i^*)}{\widehat{r}_\theta(Y_i^*)}}  \dfrac{e^{-\frac{1}{2} \widehat\beta^T b(Y_i^*)}}{ c(\theta,\widehat\beta)^{\frac{3}{2}}} \Biggl[
    \dfrac{-\frac{1}{2}e^{-\frac{1}{2} \widehat\beta^T b(Y_i^*)}}{ c(\theta,\widehat\beta)^{-\frac{1}{2}}} \Bigl(b(Y_i^*)c(\theta, \widehat\beta) - 
    \dfrac{\partial c(\theta, \widehat\beta)}{\partial \beta_r }\Bigr) \Bigl( b(Y_i^*)c(\theta, \widehat\beta) - 
    \dfrac{\partial c(\theta, \widehat\beta)}{\partial \beta_s }\Bigr)\nonumber\\
    &\qquad\qquad\qquad\qquad\qquad\qquad\qquad\qquad 
    +\Biggl(
    \dfrac{\partial c(\theta, \widehat\beta)}{\partial \beta_r }\dfrac{\partial c(\theta, \widehat\beta)}{\partial \beta_s }-\dfrac{\partial^2 c(\theta, \widehat\beta)}{\partial \beta_r \partial \beta_s }
    \Biggr)
    \Biggr]\label{eq::hessian_hellinger_exp_tilt}
\end{align}
Similarly to the MDPD loss, 
we can use properties of the cumulant function 
to estimate the partial first and second derivatives in Eqs.~(\ref{eq::gradient_hellinger_exp_tilt}, \ref{eq::hessian_hellinger_exp_tilt}).
Specifically, let $X_1,\dots,X_n\sim p_{\theta, \beta}$, we have 
$\nabla_\beta c(\theta, \beta)=\mathbb{E}[b(X)]$ with estimator $\widehat{\frac{\partial c(\theta, \beta)}{\partial\beta_r}}=\dfrac{1}{n}\sum_i b_r(X_i)=:\hat\mu_r$. The estimate for second-order partial derivatives is $\widehat{\frac{\partial^2c(\theta, \beta)}{\partial\beta_r\partial\beta_s}}=\dfrac{1}{n}\sum_{i}(b_r(X_i)-\hat\mu_r)(b_s(X_i)-\hat\mu_s)=:\hat\sigma^2_{rs}$.

\section{Additional examples} \label{app::add}

\subsection{Ricker's Model}

Ricker's model
\citep{ricker1954stock, bortolato2025box}
describes the evolution of a population over time. The observed members of the population at a time $t$ are  a random variable of an underline, latent, number of individuals $N(t)$ which is modeled at a time $t$ by
\begin{align}
\log N(t) &= \log (r) + \log(N(t-1))- N(t-1) + \sigma Z_t\nonumber\\
Y_t &\sim \mathrm{Poisson}(\phi N(t))\label{eq::rickers_model}
\end{align}
where
$Z_1,Z_2, \sim N(0,1)$.
Here $Y_t$ is the observed population, and $r$ is the growth rate and $\phi$ is regarded as a known scale parameter.
The parameters are $\sigma$ and $r$. Because of the latent variable $N(t)$, the likelihood is intractable. 

In Fig.~\ref{fig:rickers_model}, we report the $L_2$ and Hellinger profile losses, along with the corresponding relative-fit confidence sets discussed in Sec.~\ref{sec::cs_misspecified}. The estimated parameter (yellow diamond), which minimizes the profile losses, is found close to the true value (red line). Both confidence sets (blue segments) correctly cover the true parameter, with the Hellinger-based set being slightly narrower than the one derived from the $L_2$ loss. This result is consistent with prior findings on CS lengths.

The non-monotonic behavior of the $L_2$ profile loss in the top-right panel for $\sigma > 0.75$ may be attributed to a strong presence of outliers in the distribution of simulated data for such parameters combinations, combined with similarly shaped distributions to those simulated using $\sigma\in [0.6,0.75]$ in high-density areas of the distribution of observed data.
Robustness of the $L_2$ to outliers might thus result in lower loss values beyond the $\sigma$ threshold.

\begin{figure}[t!]
    \centering    
    \includegraphics[width=.4\linewidth]{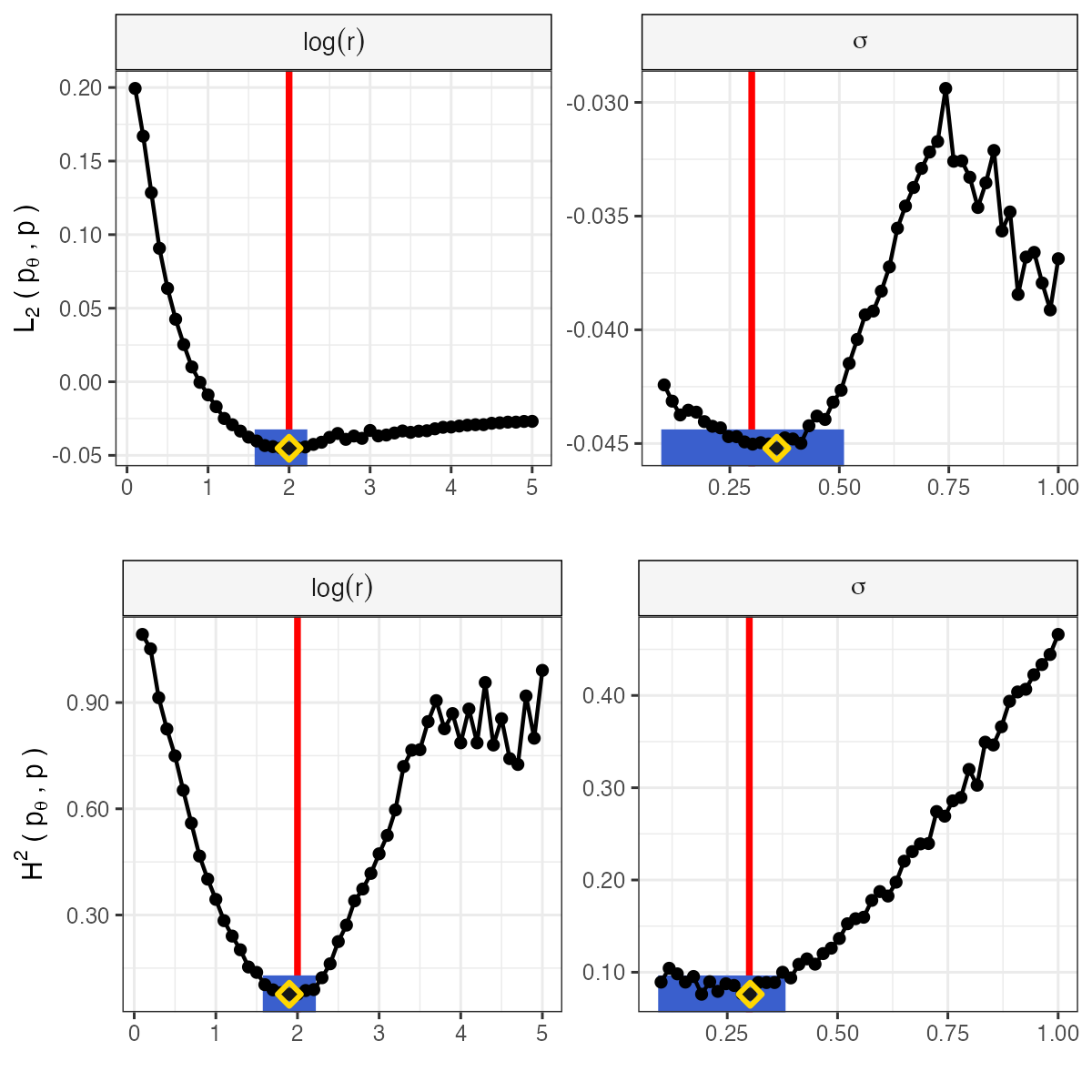}
    \caption{\sl{{\bf Inference on Ricker's model parameters.}
    The Ricker's model describes the evolution of a population over time according to the set of equations~(\ref{eq::rickers_model}). The estimated parameters are $\theta=(\log(r), \sigma)$ with known scale parameter $\phi=5$. The true parameter is $\theta^*=(2, 0.3)$. We observe the total population over $n=T_{obs}=2000$ time steps. 
    For estimation, for each parameter in the grid we simulate a total of $T=3500$ points from the model and consider the last  $2000$ points to allow for the model to converge to its stationary distribution. 
    We report results for the $L_2$ (top) and Hellinger (bottom) based SBI. 
    The estimated parameter (yellow diamond) is close to the truth (red line). Both confidence sets (blue segments) correctly cover the true parameter. The Hellinger-based sets are slightly narrower than the one derived from the $L_2$ loss, in line with findings presented in this paper about CS lengths of the approaches.
    }}
    \label{fig:rickers_model}
%\end{figure}

\bigskip

%\begin{figure}[t!]
    \centering
    \subfloat[\sl Efficient L2 discrepancy]{\includegraphics[width=5cm]{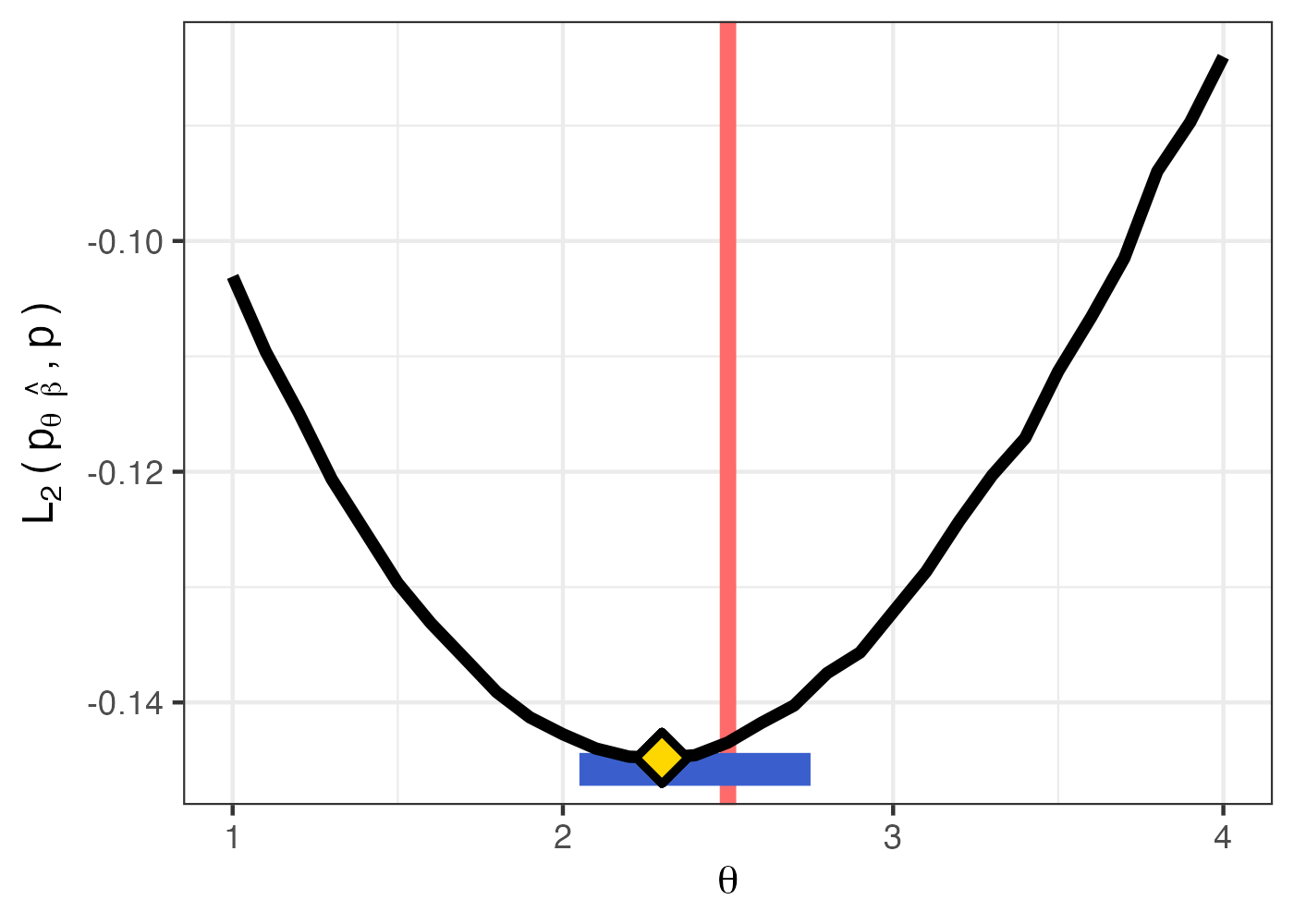}}
    \subfloat[\sl  Hellinger discrepancy]{\includegraphics[width=5cm]{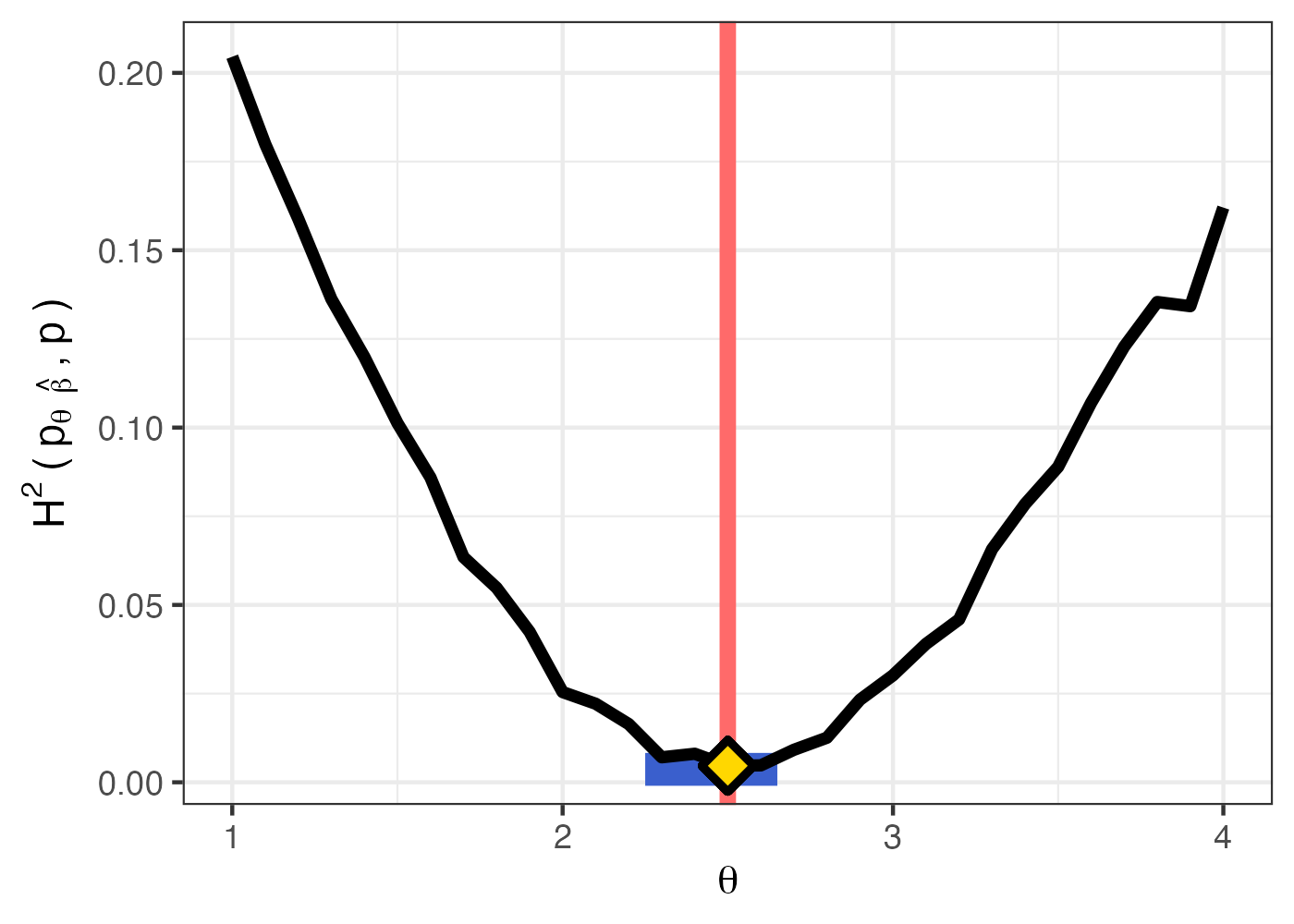}}
    \caption{\sl
    {{\bf
    Robust inference for expanded model.} The data has distribution (\ref{eq::exponential_tilt_normal}).
    The assumed model is (\ref{eq::exponential_tilt_norm}).
    {\bf (a)} Loss function for the $L_2$ discrepancy with true parameter $\theta$ (red line),
    estimated projection parameter (diamond-shaped point), and relative fit confidence set. 
    {\bf (b)} Same for the Hellinger discrepancy.
    The two discrepancies are minimized close to the true value and the confidence sets contain the true value.
 %\V{Lorenzo: please remove the 2 histogram panels.}\lortomas{Done and updated code}
}}
    \label{fig:exponential_tilt}
\end{figure}

\subsection{Inference by Projection for an Expanded Model} \label{app::expansion}

The data have distribution (\ref{eq::exponential_tilt_normal}). 
The target of inference is $\theta$.
Fig.~\ref{fig:exponential_tilt} shows the estimated $L_2$ and Hellinger discrepancies, 
the estimate of the projection parameter of $\theta$ and 
the relative fit confidence sets, when the assumed model is (\ref{eq::exponential_tilt_norm}).
%The confidence sets include the true $\theta$ and have about the same width as in Figure~\ref{fig:exponential_tilt_normal_profile_likelihood}(a) \V{there are no CS in fig 1a}\lortomas{Add a CS with loglikelihood?}, which makes sense given that the true model in (\ref{eq::exponential_tilt_normal}) is close to (\ref{eq::exponential_tilt_norm}). 
The tilting parameters were estimated using a one-step procedure, following the approach in \cite{karunamuni2011one}. In detail, initial values were obtained by maximizing the log-likelihood using the NR algorithm described in section~\ref{app::NR}. Starting from these initial values, a single additional NR step was performed to obtain the final parameter estimates, replacing the log-likelihood with either the negative Hellinger or L2 loss function, and using the corresponding gradients and Hessians derived in section \ref{app::grad_hess}. 
%then we to  using the BOBYQA \lortomas{Comment to future lorenzo...we ended up using a one-step approach for estimation of tilting parameters similarly to \cite{karunamuni2011one} instead of BOBYQA...update! xxxx} minimization algorithm of \cite{powell_bobyqa_2009}, a stochastic gradient descent that uses quadratic approximation  to the loss function and 
Good starting values are essential for convergence; the one-step approach preserves efficiency of the final estimates, as discussed in the original paper.
We used the profiled values 
$\hat \beta_1(\theta)$ and $\hat \beta_2(\theta)$.

\section{Active Learning}

\subsection{SBI specifics for the example in section~\ref{sec::AL}}\label{app::al_sbi_nn}
In the example in fig.~\ref{fig:AL_results} we used the likelihood-based SBI approach in section~\ref{sec::SBI}.  
We estimated the likelihood of data by solving a classification problem with a deep learning approach for automatic feature extraction and classification, based on a multilayer perceptron with the specifics presented in figure~\ref{fig:NN_archt}. 

\begin{figure}[b!]
    \centering
    \includegraphics[width=.4\linewidth]{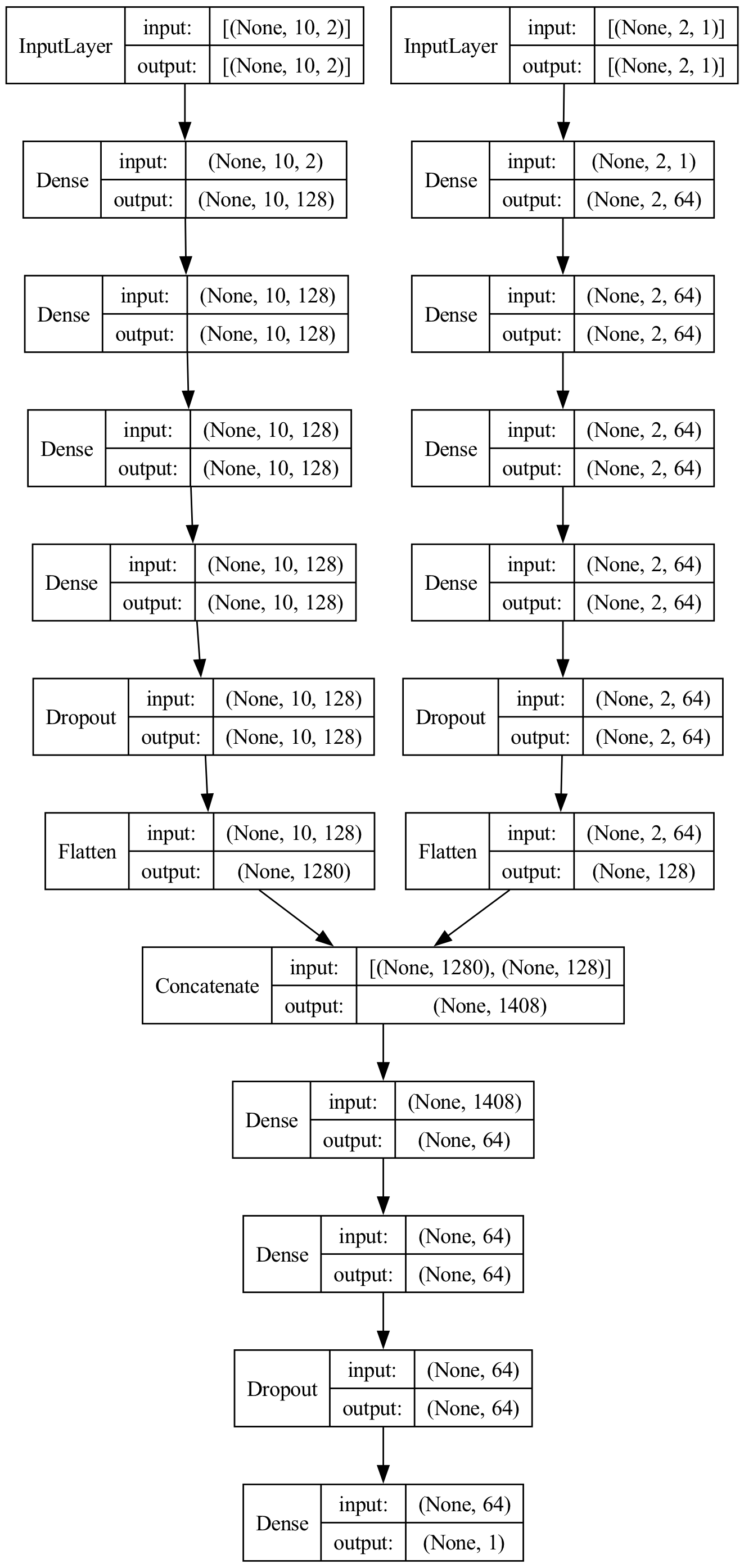}
    \caption{\sl 
     {\bf Neural network specifics.  
    } Neural network architecture used to solve the classification problem to estimate the likelihood of data for the active learning example in figure~\ref{fig:AL_results}. Data is the input of the left branch, while the parameter vector, $\theta$, is input of the right branch. 
    The output is binary for the classification problem. The neural network is trained on a $50\%$ train-validation split with binary crossentropy loss and a decaying learning rate with  5 epochs patience, $90\%$ decay factor, initial rate $10^{-3}$ and lower bound $10^{-5}$.
}
\label{fig:NN_archt}
\end{figure}

\subsection{Alternative Approach for Active Learning}
\label{app::AL}

The second active learning approach is
from \cite{zhao_sequential_2012}.
In many cases, we can write
$C =\{\theta:\ T(\theta) \geq q(\theta)\}$
where $T(\theta)$ is some statistic and
$q(\theta)$ is the $\alpha$ quantile of $T(\theta)$.
In this case
$\hat C =\{\theta:\ T(\theta) \geq \hat q(\theta)\}$
where
$\hat q(\theta)$ is the estimated quantile.
We can estimate $q(\theta)$ by local linear
quantile regression with kernel $K$ and bandwidth $h$:
choose $\hat q(\theta)$ and $\hat\mu(\theta)$ to minimize
$$
\sum_{i=1}^j L(T(\theta_i) - \mu(\theta) -\beta(\theta-\theta_i))K_h(\theta-\theta_i)
$$
where
$\mu(\theta) = \E[T|\theta]$ and
$L(t) = |t| + (2 (1-\alpha)-1)t$ is the pinball loss.
If the next $\theta$ is sampled from $f(\theta)$,
then, under regularity conditions,
$$
\E| \hat q(\theta) - q(\theta)|^2 =
h^4 \rho(\theta) +
\frac{W(\theta)}{n h f(\theta)} +
o(h^4 + (nh)^{-1})
$$
where $\rho(\theta)$ is some function of $\theta$ given in \cite{zhao_sequential_2012},
$$
W(\theta) = \frac{\alpha (1-\alpha) D_K}{m^2(q(\theta)|\theta)},
$$
$D_K = \int K^2$ and
$m$ is the density of $T$ given $\theta$.
The bias term
$h^4 \rho(\theta)$  is not affected by $f(\theta)$.
\cite{zhao_sequential_2012}
show that 
the density $f$ that minimizes
$\E| \hat q(\theta) - q(\theta)|^2$ is
$f(\theta) \propto m(q(\theta)|\theta)$.
They recommend estimating 
$m(t|\theta)$
using the conditional kernel estimator
$$
\hat m(t|\theta) =
\frac{ (j \nu h)^{-1}\sum_{i=1}^j K_h(\theta-\theta_i) K_\nu(T(\theta_i) - t)}
{ (j  h)^{-1}\sum_{i=1}^j K_h(\theta-\theta_i)}
$$
using bandwidths $h$ and $\nu$.

\end{document}